\PassOptionsToPackage{unicode}{hyperref}
\PassOptionsToPackage{hyphens}{url}
\PassOptionsToPackage{dvipsnames,svgnames,x11names}{xcolor}
\documentclass[
  12pt]{article}

\usepackage{amsmath,amssymb,bm}
\usepackage{iftex}
\usepackage{algpseudocode}
\usepackage{algorithmicx}
\usepackage{adjustbox} 
\usepackage{algorithm}
\ifPDFTeX
  \usepackage[T1]{fontenc}
  \usepackage[utf8]{inputenc}
  \usepackage{textcomp} 
\else 
  \usepackage{unicode-math}
  \defaultfontfeatures{Scale=MatchLowercase}
  \defaultfontfeatures[\rmfamily]{Ligatures=TeX,Scale=1}
\fi
\usepackage{lmodern}
\ifPDFTeX\else  
\fi
\IfFileExists{upquote.sty}{\usepackage{upquote}}{}
\IfFileExists{microtype.sty}{
  \usepackage[]{microtype}
  \UseMicrotypeSet[protrusion]{basicmath} 
}{}
\makeatletter
\@ifundefined{KOMAClassName}{
  \setlength{\parindent}{15pt}
  \setlength{\parskip}{0pt}
}{
  \KOMAoptions{parskip=never}}
\makeatother
\usepackage{xcolor}
\makeatletter
\ifx\paragraph\undefined\else
  \let\oldparagraph\paragraph
  \renewcommand{\paragraph}{
    \@ifstar
      \xxxParagraphStar
      \xxxParagraphNoStar
  }
  \newcommand{\xxxParagraphStar}[1]{\oldparagraph*{#1}\mbox{}}
  \newcommand{\xxxParagraphNoStar}[1]{\oldparagraph{#1}\mbox{}}
\fi
\ifx\subparagraph\undefined\else
  \let\oldsubparagraph\subparagraph
  \renewcommand{\subparagraph}{
    \@ifstar
      \xxxSubParagraphStar
      \xxxSubParagraphNoStar
  }
  \newcommand{\xxxSubParagraphStar}[1]{\oldsubparagraph*{#1}\mbox{}}
  \newcommand{\xxxSubParagraphNoStar}[1]{\oldsubparagraph{#1}\mbox{}}
\fi
\makeatother

\usepackage{longtable,booktabs,array}
\usepackage{calc} 
\usepackage{etoolbox}
\makeatletter
\patchcmd\longtable{\par}{\if@noskipsec\mbox{}\fi\par}{}{}
\makeatother
\IfFileExists{footnotehyper.sty}{\usepackage{footnotehyper}}{\usepackage{footnote}}
\makesavenoteenv{longtable}
\usepackage{graphicx}
\makeatletter
\def\maxwidth{\ifdim\Gin@nat@width>\linewidth\linewidth\else\Gin@nat@width\fi}
\def\maxheight{\ifdim\Gin@nat@height>\textheight\textheight\else\Gin@nat@height\fi}
\makeatother
\setkeys{Gin}{width=\maxwidth,height=\maxheight,keepaspectratio}
\makeatletter
\def\fps@figure{htbp}
\makeatother

\makeatletter
\@ifpackageloaded{caption}{}{\usepackage{caption}}
\AtBeginDocument{%
\ifdefined\contentsname
  \renewcommand*\contentsname{Table of contents}
\else
  \newcommand\contentsname{Table of contents}
\fi
\ifdefined\listfigurename
  \renewcommand*\listfigurename{List of Figures}
\else
  \newcommand\listfigurename{List of Figures}
\fi
\ifdefined\listtablename
  \renewcommand*\listtablename{List of Tables}
\else
  \newcommand\listtablename{List of Tables}
\fi
\ifdefined\figurename
  \renewcommand*\figurename{Figure}
\else
  \newcommand\figurename{Figure}
\fi
\ifdefined\tablename
  \renewcommand*\tablename{Table}
\else
  \newcommand\tablename{Table}
\fi
}
\@ifpackageloaded{float}{}{\usepackage{float}}
\floatstyle{ruled}
\@ifundefined{c@chapter}{\newfloat{codelisting}{h}{lop}}{\newfloat{codelisting}{h}{lop}[chapter]}
\floatname{codelisting}{Listing}

\makeatother
\makeatletter
\@ifpackageloaded{caption}{}{\usepackage{caption}}
\@ifpackageloaded{subcaption}{}{\usepackage{subcaption}}
\makeatother

\ifLuaTeX
  \usepackage{selnolig}  
\fi
\usepackage[]{natbib}
\usepackage{bookmark}
\usepackage{amsthm}
\newtheorem{proposition}{Proposition}
\IfFileExists{xurl.sty}{\usepackage{xurl}}{} 
\hypersetup{
  pdftitle={Title},
  pdfauthor={Author 1; Author 2},
  pdfkeywords={3 to 6 keywords, that do not appear in the title},
  colorlinks=true,
  linkcolor={blue},
  filecolor={black},
  citecolor={Blue},
  urlcolor={Blue},
  pdfcreator={LaTeX via pandoc}}

\newcommand{\anon}{1}

\usepackage{xr-hyper}  
\usepackage{geometry}
\usepackage{tikz}
\usepackage{adjustbox}
\usepackage{xcolor}
\usepackage{caption}
\usepackage{setspace}

\usetikzlibrary{arrows.meta, calc, positioning}

\newcommand{\boxtitle}[1]{{\sffamily\bfseries #1}}
\newcommand{\boxnote}[1]{{\sffamily\itshape\textcolor{blue}{#1}}}

\makeatletter
\renewcommand{\maketitle}{
	\begin{center}
		\vspace*{-1cm} 
		{\LARGE\@title\par}
		\vskip 0.5em
		{\large\@author\par}
		\vskip 0.5em
		{\@date\par}
	\end{center}
}
\makeatletter
\date{}

\begin{document}

\def\spacingset#1{\renewcommand{\baselinestretch}%
{#1}\small\normalsize} \spacingset{1}




\if1\anon
{
  \title{\bf Bayesian Analysis Using a Constrained Mixture of Normal-Inverse-Gamma Models}
  \author{Madelyn Clinch\\
    Department of Statistics, Florida State University\\
    Jonathan R Bradley\\
    Department of Statistics and Data Science, University of Missouri\\
    Andr\'{e}s F Barrientos \\
    Department of Statistics, Florida State University\\
    Garritt L Page \\
    Department of Statistics, Brigham Young University}
  \maketitle
} \fi

\if0\anon
{
  \bigskip
  \bigskip
  \bigskip
  \begin{center}
    {\LARGE\bf Bayesian analysis using a constrained mixture of normal-inverse-gamma models}
  \end{center}
} \fi

\begin{abstract}
Gaussian mixtures of regressions are commonly implemented via a Gibbs sampler.  This  Markov chain Monte Carlo (MCMC) algorithm can be computationally burdensome because of the need to update discrete-valued latent component allocation parameters  whose dimension increases as the sample size increases. In this article, we propose applying the method of composition to a Gaussian finite mixture model with a Normal-Inverse-Gamma (NIG) prior which allows one to write the posterior distribution as the product of conditional distributions. Namely, the conditional distribution of parameters given the data and mixture labels, times the marginal posterior of the mixture labels. The conditional distribution of parameters given the data and mixture labels, can be sampled from directly, instead of using MCMC. The expression of the marginal posterior of the mixture labels is known up to a proportionality constant and we adapt existing approaches in Bayesian selective inference to constrain the space of component labels to those arising from preliminary estimators,  which alleviates a commonly encountered bottleneck.  In simulation studies, we consider several settings and compare several versions of our constrained mixture of NIG models to two different MCMC-based strategies and demonstrate their use on natality data from the CDC.

\end{abstract}

\noindent%
\vfill
\thispagestyle{empty}  

\newpage

\setcounter{page}{1}
\spacingset{1.8} 

\section{Introduction}\label{gmm:s:intro}
Bayesian finite mixture models (FMMs) have become popular as they provide a general framework for density estimation and model-based clustering. Posterior sampling from the FMM is typically done via Markov chain Monte Carlo (MCMC) and is commonly carried out by augmenting the model with latent component variables \citep{diebolt1994estimation, richardson1997bayesian,  celeux2000computational, stephens2000bayesian}. This augmented model leads to a Gibbs sampler that can easily implemented when using conditionally conjugate priors \citep{diebolt1994estimation, viele2002modeling}. More generally, a range of MCMC methods have been developed for finite mixtures, including reversible jump MCMC \citep{richardson1997bayesian}, birth-and-death MCMC \citep{stephens2000bayesian}, split-merge MCMC \citep{jain2004split}, and the allocation sampler \citep{nobile2007bayesian}, among others.

Although Gibbs samplers are straightforward to implement, it is well known that posterior distributions for mixture models can exhibit multimodality and are invariant under permutations of component labels. That is, if the labels of the mixture components are swapped, the data model does not change. This gives rise to the label switching problem and complicates inferential summaries for component specific parameters \citep{jasra2005markov, celeux2000computational, stephens2000dealing}. Additional difficulties arise when the sample size, number of covariates, and number of potential components (i.e., the maximum \(K\)) are large, as standard MCMC approaches may mix slowly and yield highly autocorrelated chains. This increases computational costs and leads to nontrivial convergence assessments  \citep{celeux2000computational, jain2004split}.

Because of this, there is a need for computationally efficient Bayesian approaches for Gaussian FMMs and mixture regressions. Using the method of composition \citep{press2009subjective} to decompose a Gaussian FMM with a Normal-Inverse-Gamma (NIG) prior, we can write the posterior distribution as a product of conditional distributions that can be sampled from directly. In particular, this product consists of (i) the conditional distributions of parameters given the data and mixture labels, and (ii) the marginal posterior distribution of the mixture labels available up to a proportionality constant. {In some cases we allow for uncertainty in the number of components as well.} This structure enables posterior sampling schemes that do not require MCMC which in turn alleviates the mixing and convergence concerns. {To our knowledge, Bayesian finite mixture model implementations generally rely on some form of MCMC or variational inference strategies \citep{mcgrory2007variational, fan2012variational}, neither of which are needed using this direct sampler.}

NIG priors have been used in recent work to enable exact Bayesian inference for Gaussian regression models by providing closed form posteriors that allow for MCMC-free sampling. \citet{banerjee2020modeling} considers point-referenced Gaussian spatial regression and shows that, conditional on fixing a small set of spatial covariance/process parameters, the remaining model reduces to a conjugate Gaussian regression. This conditional conjugacy yields closed form posteriors and enables direct sampling for the regression and variance parameters. \citet{finley2019efficient} use a related strategy for nearest-neighbor Gaussian process (NNGP) models where conditioning on fixed spatial process parameters, conjugacy allows for posterior sampling of regression and variance parameters without computationally intensive MCMC. \citet{zhang2019practical} extend this conjugate NNGP framework by enabling scalable, MCMC-free inference for the latent spatial process as well. We adopt a similar parametric strategy in the finite mixture regression setting where we obtain a closed from expression of the posterior distribution by decomposing the joint posterior distribution into the product of conditional distributions that we can sample from directly. 

For FMMs, the NIG model specification yields exact updates for the continuous parameters given the data and the mixture labels. Although the marginal posterior of the labels can be derived, it belongs to a very large discrete parameter space, making it computationally expensive to sample from. We consider several strategies to solve this issue. First, a Gibbs sampler can be used within the method of composition to sample from the marginal posterior of the mixture labels. {We call this approach the method of composition with MCMC where needed, or MC-MCMC.}  A second strategy, uses a truncated (or constrained) support for mixture labels, which is a general strategy common in the selective inference literature.  That is, rather than considering all possible label configurations, we constrain the space of the cluster labels to a smaller set of plausible labels, which we call a ``candidate set'' of cluster labels.

One of the primary challenges of selective inference in the Bayesian setting is that the implied likelihood has an intractable normalizing constant leading to a doubly intractable posterior distribution \citep{park2020function}. \citet{bradley2021empirical} and \citet{zong2023criterion} consider constraining the joint distribution of the data, random effects, and parameters instead of the likelihood. They show several benefits of incorporating a constraint in this way, including improved predictive performance relative to the unconstrained model, conditions for a type of posterior consistency, and other desirable properties. Constraining the joint distribution of the data, random effects, and parameters instead of the likelihood also avoids the introduction of an intractable normalizing constant, which leads to computation gains. We adopt this approach in the present work.

The selective inference literature develops procedures that properly account for uncertainty introduced by a data-dependent selection step by conditioning on the selection event \citep{taylor2015statistical, fithian2014optimal}. Bayesian selective inference first introduced in \citet{yekutieli2012adjusted}, involves conditioning the joint model of the parameters and the data on the event that a specified selection rule is satisfied, which leads to a ``selection-adjusted'' likelihood and a corresponding posterior distribution that conditions on both the data and the selection rule. Selective inference ideas have been developed for clustering \citep{gao2024selective, chen2023selective}, though existing work is primarily frequentest and does not address the use of selective inference in Bayesian finite mixture modeling. In our setting, we utilize selective inference primarily to reduce the support of the mixture labels so that one can sample directly from their marginal posterior distribution. 

For small sample sizes selective inference is not required as we can enumerate all possible label orientations. However, for large sample sizes, we construct a candidate set of mixture labels that are used within the Bayesian selective inference framework. We consider {three} strategies to construct a candidate set of mixture labels. {First, we implement a Gibbs sampler for the unconstrained mixture of NIG models using a subset of the data. We then use the resulting posterior draws from this subset fit to predict allocation labels for all observations, and the unique full data allocations form the candidate set. {We call this strategy ``DS-Const,'' where DS is an acronym for ``direct sampler.''} The second strategy uses several existing machine learning (ML) clustering algorithms {with random initializations} to generate candidate labels. Each algorithm produces allocation labels, and the collection of the predicted allocations across the selected algorithms forms the candidate set. {Consequently, the size of the candidate set corresponds to the number of ML algorithms used and specifications of the finite mixture model (e.g., number of components, choice of covariates, etc.).} {We call this strategy DS-ML. The third strategy uses the posterior mode of the DS-Const approach as the initialization for an ML method to create the candidate set for direct sampling.} {We refer to this third strategy as DS-Const-MAP.} {Figure~\ref{fig:method_compare} summarizes the contributions of this article, which includes several constrained and unconstrained models and their corresponding sampling strategies.}

\begin{figure}[!h]
\captionsetup{font={stretch=1}}
\centering
\scalebox{0.85}{\begin{tikzpicture}[
    scale=0.6,
    transform shape,
    font=\sffamily,
    execute at begin node={\hyphenpenalty=10000\exhyphenpenalty=10000\relax},
    >=Stealth,
    line/.style={->, very thick, draw=black!75},
    box/.style={
        draw=black!70,
        thick,
        rounded corners=8pt,
        align=center,
        inner sep=7pt,
        text width=#1
    },
    root/.style={box=2.6cm, fill=black!5, minimum height=1.35cm},
    branch/.style={box=3.7cm, fill=blue!6, minimum height=2.15cm},
    method/.style={box=9.8cm, fill=white, minimum height=1.85cm},
    output/.style={box=6.4cm, fill=green!7, minimum height=2.45cm},
    note/.style={font=\sffamily\small\itshape, text=blue!70!black}
]

\node[root] (fmm) at (0,0) {
    \boxtitle{Gaussian}\\
    \boxtitle{FMM}
};

\node[branch] (uncon) at (5.0,3.0) {
    \boxtitle{Unconstrained}\\
    \boxtitle{Model}\\[2pt]
    Sample labels from full
    label space
};

\node[branch] (con) at (5.0,-3.0) {
    \boxtitle{Constrained Model}\\[2pt]
    Restrict space of
    labels to candidate set
};

\node[method] (direct) at (14.2,6.8) {
    \boxtitle{Direct Sampler}\\[2pt]
    Enumerate full label space\par
    \boxnote{Feasible for $N$~$\leq$~$13$ only}
};

\node[method] (mcmc) at (14.2,3.8) {
    \boxtitle{MC-MCMC}\\[2pt]
    Gibbs sampler to sample from marginal posterior of
    labels and method of composition for continuous
    parameters\par
    \boxnote{Slow for large $N$}
};

\node[method] (gibbs) at (14.2,0.9) {
    \boxtitle{Traditional Gibbs Sampler}\\[2pt]
    Full MCMC baseline\par
    \boxnote{Slow for large $N$}
};

\node[method] (dsconst) at (14.2,-1.9) {
    \boxtitle{DS-Const}\\[2pt]
    Gibbs sampler for unconstrained Gaussian FMM
    using a subset of the data then predict labels for
    entire data
};

\node[method] (dsml) at (14.2,-4.6) {
    \boxtitle{DS-ML}\\[2pt]
    Build candidate set via ML clustering on full dataset
};

\node[method] (dsconstmap) at (14.2,-7.3) {
    \boxtitle{DS-Const-MAP}\\[2pt]
    Use the marginal posterior mode of DS-Const as
    the initialization for ML clustering
};

\node[output] (final) at (25.0,-4.6) {
    (1) Sample from reduced label
    space\par
    (2) Sample continuous parameters
    via method of composition\par
    \boxnote{Scalable for large $N$}
};

\draw[line] (fmm.north east) -- (uncon.west);
\draw[line] (fmm.south east) -- (con.west);

\draw[line] (uncon.east) -- (direct.west);
\draw[line] (uncon.east) -- (mcmc.west);
\draw[line] (uncon.east) -- (gibbs.west);

\draw[line] (con.east) -- (dsconst.west);
\draw[line] (con.east) -- (dsml.west);
\draw[line] (con.east) -- (dsconstmap.west);

\draw[line] (dsconst.east) -- ($(final.west)+(0,0.75)$);
\draw[line] (dsml.east) -- (final.west);
\draw[line] (dsconstmap.east) -- ($(final.west)+(0,-0.75)$);

\end{tikzpicture}}
\caption{Overview of the unconstrained model, which samples labels over the full label space, and the constrained model, which restricts inference to a candidate set of plausible label configurations. Let \(N\) denote the sample size}
\label{fig:method_compare}
\end{figure}
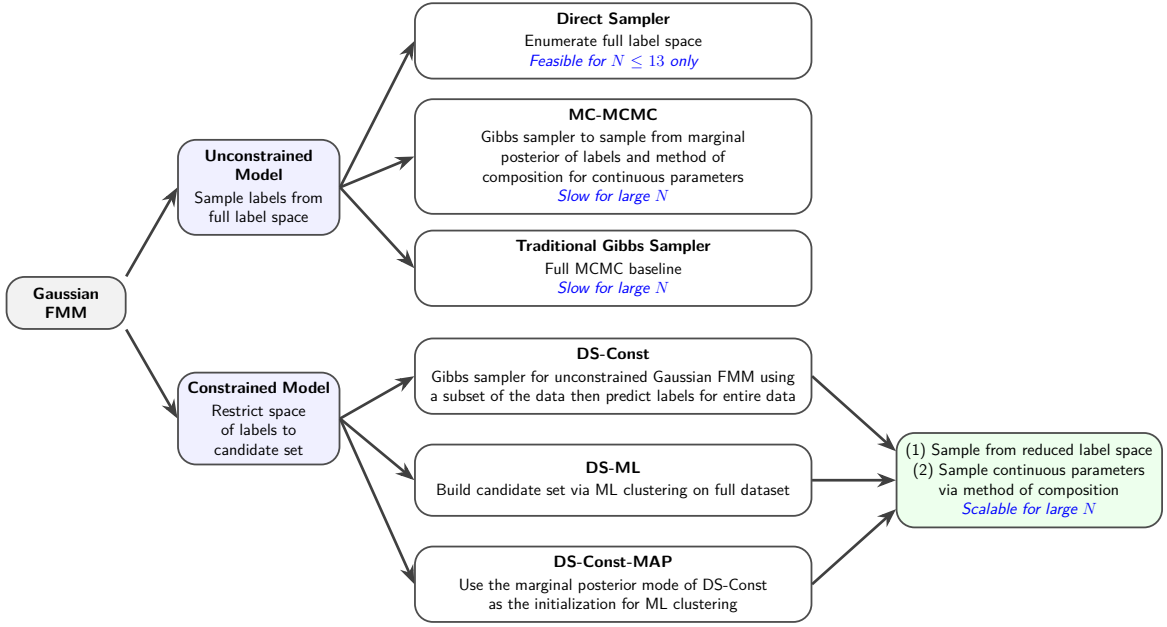

The remainder of this paper is organized as follows. Section \ref{gmm:s:method} introduces the constrained NIG finite mixture regression model and uses the method of composition to derive the posterior distribution. The details for the  various strategies to constrain the model are also presented. Section \ref{gmm:s:simstudy} presents simulation studies comparing our methods with two different Gibbs sampler strategies across a variety of scenarios. A discussion comparing the techniques for constructing the candidate set of labels is also provided in this section. Section \ref{gmm:s:birthweight} presents data illustration using a birth weight dataset obtained for the Centers for Disease Control and Prevention (CDC). Finally, Section \ref{gmm:s:discuss} concludes with a discussion.

\section{Methodology}\label{gmm:s:method}
In this section, we define the NIG model use the method of composition to derive an expression for the posterior distribution. We then present the constrained model and strategies to define the constraint.

\subsection{Normal Inverse Gamma Model}\label{gmm:ss:nig_hm}
Let \(\mathbf{y}=(Y_1, \ldots, Y_N)^{\prime}\) be the \(N\)-dimensional vector of observed data from a mixture of \(K\) Gaussian regressions. Let \(f\) denote a generic probability density function (pdf) or probability mass function (pmf), \(\mathbf{x}_i\) be a generic known \(p\)-dimensional covariate vector for \(i = 1, \dots, N\), \(\mathbf{z}_i = (Z_{i1}, \dots, Z_{iK})'\) be the allocation vector with \(Z_{ik} = 1\) indicating that observation \(i\) is assigned to component \(k\), and \(\boldsymbol{\pi} = (\pi_1, \dots, \pi_K)'\) are the mixing proportions. Consider the following hierarchical model:
\begin{equation}\label{gmm:eq:hm}
\begin{aligned}
    f(\mathbf{y}\mid \{\boldsymbol{\beta}_k\}, \{\sigma_k^2\}, \{Z_{ik}\}, K) &= \prod_{i = 1}^N \prod_{k = 1}^K \left\{\frac{(\sigma_k^2)^{-1/2}}{\sqrt{2\pi}}\exp\left(-\frac{(Y_i - \mathbf{x}_i^{\prime}\boldsymbol{\beta}_k)^2}{2\sigma_k^2}\right) \right\}^{Z_{ik}} \\ 
    f(\mathbf{z}_{i}\mid{\boldsymbol{\pi}},
    K) &= \prod_{k = 1}^{K}\pi_k^{Z_{ik}}; \quad Z_{ik}\in \{0,1\}, \sum_{k = 1}^KZ_{ik} = 1  \\
    f(\boldsymbol{\pi}\mid \alpha_1,\dots, \alpha_K, K) &= \frac{\Gamma(\sum_{k =1}^K\alpha_k)}{\prod_{k =1}^K\Gamma(\alpha_k)}\prod_{k = 1}^K \pi_k^{\alpha_k -1}  \\
   f({\{}\boldsymbol{\beta}_k{\}}\mid {\{}\sigma_k^2{\}}, \sigma_{\beta}^2{, K})
&= {\prod_{k=1}^K}
\frac{(\sigma_k^2\sigma_{\beta}^2)^{-p/2}}{(2\pi)^{p/2}}
\exp\left(-\frac{\boldsymbol{\beta}_k^{\prime}\boldsymbol{\beta}_k}{2\sigma_k^2\sigma_{\beta}^2}\right)  \\
f({\{}\sigma_k^2{\}}\mid \omega, \kappa{, K})
&= {\prod_{k=1}^K}
\frac{\kappa^{\omega}}{\Gamma(\omega)}
(\sigma_k^2)^{-\omega-1}
\exp\left(-\frac{\kappa}{\sigma_k^2}\right) \\
     K & { \sim} \text{Categorical}(\rho_1, \dots, \rho_M),
\end{aligned}
\end{equation}
where \(\boldsymbol{\pi} = (\pi_1, \dots, \pi_K)\) are the mixing proportions with $\sum_{k=1}^K\pi_k = 1$. The Dirichlet prior on \(\boldsymbol{\pi}\) has hyperparameters \((\alpha_1, \dots, \alpha_K)\) with \(\alpha_k > 0\). Let $\mathbf{x}_i = (x_{1i}, \ldots, x_{pi})$ denote a $p$-dimensional vector of explanatory variables, \(\boldsymbol{\beta}_k\) is a \(p\)-dimensional vector of regression coefficients for component \(k\), and \(\sigma_k^2 > 0\) be the corresponding variance given an inverse gamma prior. { Here and throughout, when a quantity indexed by $i$ and/or $k$ is enclosed in braces, for example $\{\boldsymbol{\beta}_k\}$, $\{\sigma_k^2\}$, or $\{Z_{ik}\}$, it denotes the collection of that quantity over the relevant indices, with $i=1,\dots,N$ and $k=1,\dots,K$.} {The parameter \(K\) denotes the unknown number of mixture components. We assume a prior with support \(\{1, \dots, M\}\), where \(M\) is a pre-specified upper bound on the number of components. The hyperparameter \(\rho_1,\dots,\rho_M\) are prior probabilities with \(\rho_m=\Pr(K=m)\) for \(m=1,\dots,M\) and \(\sum_{m=1}^M \rho_m=1\). 



\subsection{Method of Composition}\label{gmm:ss:comp}
Using the method of composition, we can decompose the posterior distribution of \(\{\boldsymbol{\beta}_k\}, \boldsymbol{\pi}, \{\sigma_k^2\}, \{Z_{ik}\}, K \mid \mathbf{y}\) from Equation \eqref{gmm:eq:hm} into a product of conditional distributions that we can generate independent samples from. The following proposition gives the known conditional distributions.

\begin{proposition}
    The joint posterior distribution of \(\{\boldsymbol{\beta}_k\}, \boldsymbol{\pi}, \{\sigma_k^2\}, \{Z_{ik}\}, K \mid \mathbf{y}\) can be decomposed as  
\begin{align}
    f(\{\boldsymbol{\beta}_k\},& \boldsymbol{\pi}, \{\sigma_k^2\}, \{Z_{ik}\}, K \mid \mathbf{y}) = f(\{\boldsymbol{\beta}_k\}\mid \{\sigma_k^2\}, \boldsymbol{\pi}, \{Z_{ik}\}, K,\mathbf{y})f( \{\sigma_k^2\}\mid \boldsymbol{\pi}, \{Z_{ik}\}, K, \mathbf{y})\notag\\ &\times f(\boldsymbol{\pi}\mid \{Z_{ik}\}, K, \mathbf{y})f(\{Z_{ik}\}\mid K, \mathbf{y}) f( K\mid \mathbf{y}).
    \label{gmm:eq:joint}
\end{align}
For \(k = 1, \dots, K\), let \(n_k = \sum_{i = 1}^NZ_{ik}\) be the number of observations assigned to component \(k\), let \(\mathbf{X}_k = (\mathbf{x}_i^{\prime}: Z_{ik} = 1, i=1,\dots, N)\) be the \(n_k \times p\) matrix of covariates, and let \(\mathbf{y}_k = (Y_i: Z_{ik} = 1, i = 1,\dots, N)\) be the \(n_k\)-dimensional response vector. Then the conditional distribution of the regression coefficients is
\begin{align}
    \boldsymbol{\beta}_k \mid \sigma_k^2,\{Z_{ik}\},K,\mathbf{y}
\sim
\begin{cases}
\text{MVN}\!\left(\boldsymbol{\mu}_{P,k},\,\boldsymbol{\Sigma}_{P,k}\right), 
& \text{if } n_k > 0,\\[6pt]
\text{MVN}\!\left(\mathbf{0}_p,\,\sigma_k^2\sigma_{\beta}^2\mathbf{I}_p\right),
& \text{if } n_k = 0~,
\end{cases}
\label{gmm:eq:beta_post}
\end{align}
where \(\mathbf{0}_p\) is a \(p\)-dimensional vector of zeros, \(\mathbf{I}_p\) is a \(p \times p\) identity matrix, \(\boldsymbol{\mu}_{P,k} = \left(\mathbf{X}_k^{\prime}\mathbf{X}_k + \frac{1}{\sigma_{\beta}^2}\mathbf{I}_p\right)^{-1}\mathbf{X}_k^{\prime}\mathbf{y}_k\) and \(\boldsymbol{\Sigma}_{P,k} = \sigma_k^2\left(\mathbf{X}_k^{\prime}\mathbf{X}_k + \frac{1}{\sigma_{\beta}^2}\mathbf{I}_p \right)^{-1}\). The conditional distribution of the component variances is
\begin{align}
\sigma_k^2 \mid \boldsymbol{\pi}, \{Z_{ik}\},K,\mathbf{y}
\sim
\begin{cases}
\text{IG}\!\left(\omega_{P,k},\,\kappa_{P,k}\right),
& \text{if } n_k>0,\\[6pt]
\text{IG}\!\left(\omega,\,\kappa\right),
& \text{if } n_k=0~,
\end{cases}
\label{gmm:eq:sig_post}
\end{align}
where \(\omega_{P,k} = \omega + \frac{n_k}{2}\) and \(\kappa_{P,k} = {\mathbf{y}_k^{\prime}(\sigma_{\beta}^2\mathbf{X}_k\mathbf{X}_k^{\prime} + \mathbf{I}_{n_k})^{-1}\mathbf{y}_k}/{2} + \kappa\). The conditional distribution of the mixture weights is  
\begin{align}
    f(\boldsymbol{\pi}\mid \{Z_{ik}\}, K, \mathbf{y}) \propto  \text{Dirichlet}\left( n_1+  \alpha_1, \dots, n_K + \alpha_K \right).\label{gmm:eq:pi_post}
\end{align}
Now let \(\Delta_K\) be the set of all possible values of \(\{Z_{ik}\}\) for a given \(K\). The joint distribution of \(\{Z_{ik}\}, K, \mathbf{y}\)  is given by
\begin{align}
    f(\{Z_{ik}\}, &K, \mathbf{y}) = f(K) \frac{\Gamma(\sum_{k = 1}^K\alpha_k)}{\Gamma(N + \sum_{k = 1}^K\alpha_k)} \frac{\prod_{k = 1}^K\Gamma(\alpha_k + \sum_{i = 1}^N Z_{ik})}{\prod_{k = 1}^K\Gamma(\alpha_k)} \notag \\
    &\times\prod_{k = 1,\dots, K, n_k > 0}(2\pi)^{-\frac{n_k}{2}} \vert \mathbf{I}_{n_k} + \sigma_{\beta}^2\mathbf{X}_k\mathbf{X}_k^{\prime} \vert^{-\frac{1}{2}}\frac{\kappa^{\omega}}{\Gamma(\omega)} \frac{\Gamma(\omega + \frac{n_k}{2})}{\left(\kappa + \frac{1}{2}\mathbf{y}_k^{\prime}(\sigma_{\beta}^2\mathbf{X}_k\mathbf{X}_k^{\prime} + \mathbf{I}_{n_k})^{-1}\mathbf{y}_k \right)^{\omega + \frac{n_k}{2}}}.
    \label{eq:gmm:zXy}
\end{align}
Then we have, 
\begin{align}
    f(\{Z_{ik}\}\mid K,\mathbf{y})
&= \frac{\sum_{\{Z_{ik}\}\in \Delta_K} f(\{Z_{ik}\},K\mid \mathbf{y})}
            {\sum_{m=1}^M \sum_{\{Z_{ik}\}\in \Delta_m} f(\{Z_{ik}\},m\mid \mathbf{y})}.
    \label{gmm:eq:z_post}
\end{align}
\end{proposition}
\begin{proof}
See Appendix \ref{appen:proof_prop} for the proof.
\end{proof}

The expressions in Equation (\ref{gmm:eq:z_post}) define discrete posterior distributions that can be sampled from directly in low-dimensional unconstrained settings. Under the prior specification \(K \in \{1,\dots,M\}\) with \(\Pr(K=m)=\rho_m\), the posterior distribution \(f(\{Z_{ik}\}\mid K,\mathbf{y})\) is categorical over the set of partitions \(\Delta_K\), and \(f(K\mid \mathbf{y})\) is categorical over \(\{1,\dots,M\}\). In practice, it is difficult to enumerate the set \(\Delta_K\), which is on the order of \(K^N\) number of elements. For this reason, in our simulation study we implement the direct sampler only for \(N \leq 13\) to illustrate the ability to perform MCMC-free posterior inference. 

{Alternatively, it is common to fix the number of available components at a sufficiently large value (e.g., \(K = 25\)) and perform inference conditional on this choice (\citealt{RousseauMengersen:2011}). The overspecification of \(K\) is handled by allowing empty components, so that \(n_k = 0\) for some \(k\). The distinction between the number of mixture components ($K$) and the number of clusters or occupied components (\(\sum_{k = 1}^K\mathbb{I}(n_k >0)\), where \(\mathbb{I}(\cdot)\) is the indicator function) has appeared in the literature (\citealt{argiento:2022, pageetal:2025}).  In this setting, care must be taken to not underspecify \(K\) as doing so negatively affects estimating the number of clusters.  It turns out that fixing \(K=25\) corresponds to a degenerate prior \(\Pr(K=25)=1\) (setting \(M=25\) and \(\rho_{25}=1\)), and side-steps the need of sampling \(K\). 
In the non-direct sampling approaches considered in this paper, we assume a fixed value of \(K\).}

One option to sample from \(f(\{Z_{ik}\}\mid K,\mathbf y)\) is to devise a Gibbs sampler that avoids enumerating \(\Delta_K\). Carrying this out is one of the contributions we aim to make.  Let \(\mathbf{z}_i = (Z_{i1}, \dots, Z_{iK})^{\prime}\) and let \(e_b\) denote the \(b\)-th elemental vector in \(\mathbb{R}^K\). Using the joint distribution \(f(\{Z_{ik}\}, K, \mathbf{y})\) in Equation (\ref{eq:gmm:zXy}), the full conditional for \(\mathbf{z}_i\) is 
\begin{align}
    f(\mathbf{z}_i = e_b \mid \mathbf{z}_{-i}, K, \mathbf{y}) = \frac{f(\mathbf{z}_{-i}, \mathbf{z}_i = e_b, K , \mathbf{y})}{\sum_{h =1}^Kf(\mathbf{z}_{-i}, \mathbf{z}_i = e_h, K, \mathbf{y})},
    \label{gmm:eq:dsmcmc}
\end{align}
where \(\mathbf{z}_{-i}\) denotes all allocations except the \(i\)-th. We refer to this as a method of composition with MCMC where necessary (MC-MCMC).  To implement MC-MCMC first run the Gibbs sampler in Equation (\ref{gmm:eq:dsmcmc}) to obtain draws \(t = 1, \dots, n_{\text{samp}}\) from \(f(\{Z_{ik}\}\mid K, \mathbf{y})\) and then conditional on each \(\mathbf{z}^{[t]}\) we sample the continuous parameters. That is, \(\boldsymbol\pi^{[t]} \sim f(\boldsymbol\pi\mid \mathbf z^{[t]},K,\mathbf y)\) from \eqref{gmm:eq:pi_post},
\(\{\sigma_k^{2\,[t]}\}\sim f(\{\sigma_k^2\}\mid \mathbf z^{[t]},K,\mathbf y)\) from \eqref{gmm:eq:sig_post}, and
\(\{\boldsymbol\beta_k^{[t]}\}\sim f(\{\boldsymbol\beta_k\}\mid \{\sigma_k^{2\,[t]}\},\mathbf z^{[t]},K,\mathbf y)\) from \eqref{gmm:eq:beta_post}. While this avoids summing over $\Delta_K$, it remains an MCMC procedure with computational cost that grows with \(N\) and \(K\) and the number of iterations required for adequate mixing. For this reason, we introduce a constrained model in the next section.

\subsection{Constrained Model}\label{gmm:ss:constrained}
In general, the set \(\Delta_K\) consists of \(K^N\) distinct elements which makes the expression of the marginal posterior distribution of \(Z\) in Equation (\ref{gmm:eq:z_post}) infeasible even with relatively small $N$ and/or $K$. Motivated by the Bayesian selective inference literature (e.g., see \citet{bradley2021empirical, yekutieli2012adjusted}, among others), we constrain \(\Delta_K\) to a smaller subset of plausible candidates to achieve computational feasibility even for large $N$ and/or $K$.

We use a strategy similar to \citet{bradley2021empirical} where the joint model of the parameters \(\boldsymbol{\theta}  =  \{\boldsymbol{\pi}, \{\boldsymbol{\beta}_k\}, \{\sigma_k^2\}\}\), the data \(\mathbf{y}\), and the labels \(\{Z_{ik}\}\) are constrained based on a selection event \(\{(\{Z_{ik}\},K) \in \mathcal{C}(\textbf{y})\}\) where \(\mathcal{C}(\textbf{y}) = \{(\{Z_{ik}^{[1]}\},K^{[1]}), \dots, (\{Z_{ik}^{[W]}\},K^{[W]})\}\) are \(W\) candidates from \({\Delta = \cup_{K = 1}^{M}}\Delta_K\). The set $\mathcal{C}$ can be a function of the data $\mathbf{y}$, and the selective inference strategy addresses the sampling variability of $\mathcal{C}$ by defining $\mathcal{C}$ in the support of the model. To construct \(\mathcal{C}(\textbf{y})\) we use either (i) a Gibbs sampler fit using a subset of the data and a prediction step to produce candidate allocations for all \(N\) observations, or (ii) multiple ML clustering algorithms fit to the full dataset, each producing a candidate allocation (see additional details in Section \ref{gmm:ss:labels}).

Let \(f(\mathbf{y},\boldsymbol{\theta}, \{Z_{ik}\},K)\) denote the joint distribution under the unconstrained mixture model. Following the construction in \citet{bradley2021empirical}, we define the constrained joint model by truncating the joint support of \(\{Z_{ik}\}\) and $K$ to \(\mathcal{C}(\textbf{y})\). That is,
\begin{align}
g(\mathbf{y},\boldsymbol{\theta},\{Z_{ik}\},K \mid \mathcal{C}(\textbf{y}))=
\frac{f(\mathbf{y},\boldsymbol{\theta},\{Z_{ik}\},K)\,\mathbb{I}\{(\{Z_{ik}\},K)\in\mathcal{C}(\textbf{y})\}}
{\mathcal{N}},
\label{eq:gmm:constrained_joint}
\end{align}
where \(\mathcal{N} = \int \int \sum_{(\{Z_{ik}\},K) \in \mathcal{C}(\textbf{y})} f(\mathbf{y}, \boldsymbol{\theta}, \{Z_{ik}\}, K) d\mathbf{y}d\boldsymbol{\theta}\) is a normalizing constant and \(\mathbb{I}(\cdot)\) an indicator function. The implied data generating mechanism for $\mathbf{y}$ is the marginal distribution $g(\mathbf{y}) = \sum_{(\{Z_{ik}\},K) \in \mathcal{C}(\textbf{y})} \int f(\mathbf{y}, \boldsymbol{\theta}, \{Z_{ik}\}, K) d\boldsymbol{\theta}/\mathcal{N}$, which is consistent with the interpretation of the data generating mechanism in the nonparametric Bayesian literature \citep[e.g., see][Chp. 1]{schervish2012theory}.

Now, conditioning on the observed data \(\mathbf{y}\) yields the constrained posterior
\begin{align}
g(\boldsymbol{\theta},\{Z_{ik}\},K \mid \mathbf{y},\mathcal{C}(\textbf{y}))=
\frac{f(\boldsymbol{\theta},\{Z_{ik}\},K\mid \mathbf{y})\,\mathbb{I}\{(\{Z_{ik}\},K)\in\mathcal{C}(\textbf{y})\}}
{\Pr_f((\{Z_{ik}\},K)\in\mathcal{C}(\textbf{y})\mid \mathbf{y})},
\label{eq:gmm:constrained_post}
\end{align}
 where \(\Pr_f((\{Z_{ik}\},K)\in\mathcal{C}(\textbf{y})\mid \mathbf{y})\) is the posterior probability of the selection event computed under the unconstrained model.
That is,  
\begin{align}
\Pr_f((\{Z_{ik}\},K)\in\mathcal{C}(\textbf{y})\mid \mathbf{y})
=\sum_{w=1}^W f(\{Z_{ik}^{[w]}\},K^{[w]}\mid \mathbf{y}).
\label{eq:gmm:constrained_norm}
\end{align}
\noindent
\noindent
{ If one normalizes the likelihood as done in \citet{yekutieli2012adjusted}, instead of the joint distribution stated in (\ref{eq:gmm:constrained_joint}), one can still obtain the posterior distribution in (\ref{eq:gmm:constrained_post}). Namely, consider the hierarchical model whose data model is given by
\begin{equation*}
\begin{aligned}
&g(\mathbf{y}\mid \{\boldsymbol{\beta}_k\}, \{\sigma_k^2\}, \{Z_{ik}\}, K)\\
    &= \frac{1}{\mathcal{L}(\bm{\theta},\{Z_{ik}\},K)}\prod_{i = 1}^N \prod_{k = 1}^K \left\{\frac{(\sigma_k^2)^{-1/2}}{\sqrt{2\pi}}\exp\left(-\frac{(Y_i - \mathbf{x}_i^{\prime}\boldsymbol{\beta}_k)^2}{2\sigma_k^2}\right) \right\}^{Z_{ik}} \mathbb{I}\{(\{Z_{ik}\},K)\in\mathcal{C}(\textbf{y})\},
\end{aligned}
\end{equation*}
where $\mathcal{L}(\bm{\theta},\{Z_{ik}\},K) = \int f(\mathbf{y}\mid \bm{\theta},\{Z_{ik}\},K) \mathbb{I}\{(\{Z_{ik}\},K)\in\mathcal{C}(\textbf{y})\} d\textbf{y}$ being the normalizing constant for a truncated normal distribution. The prior distribution for this hierarchical model is determine by $f(\bm{\theta},\{Z_{ik}\},K)\mathcal{L}(\bm{\theta},\{Z_{ik}\},K)$, which is not necessarily proper, where $f(\bm{\theta},\{Z_{ik}\},K)$ is equal to the product of all the prior distributions listed in (\ref{gmm:eq:hm}). This adjusted prior distribution up-weights values of the parameters that lead to a higher conditional probability of the selection event. The joint density implied by the adjusted likelihood and adjusted prior is proportional to our model in (\ref{eq:gmm:constrained_joint}), which leads to the same expression for the posterior distribution in (\ref{eq:gmm:constrained_post}). There are other ways to augment the marginal distribution of the data $g(\textbf{y})$ to reproduce (\ref{eq:gmm:constrained_post}); see Section 3.4 in \citet{zong2023criterion} for one such alternative expression/interpretation of similar truncated Bayesian hierarchical models.} 

\newpage
The constraint forces \(\{Z_{ik}\}\) to lie in a candidate set of size \(W << K^N\), which is scalable in our setting. The above model simplifies because conditional independence arises in each conditional distribution except for the allocation distribution. In particular, given \(\{Z_{ik}\}\), the continuous parameters \(\{\boldsymbol{\beta}_k\}\), \(\{\sigma_k^2\}\), and \(\boldsymbol{\pi}\) are conditionally independent of the candidate set \(\mathcal{C}(\textbf{y})\). The candidate set only enters through the constrained joint support of \(\{Z_{ik}\}\), $K$, and $\mathbf{y}$.

\color{black}

\noindent

From the results in Section \ref{gmm:ss:comp} { and Equation \ref{eq:gmm:constrained_post}}, the posterior distribution can be decomposed as 
\begin{align*}
    g(\{\boldsymbol{\beta}_k\}, \boldsymbol{\pi}, \{\sigma_k^2\}, \{Z_{ik}\}, K \mid \mathbf{y}, \mathcal{C}(\textbf{y})) = f&(\{\boldsymbol{\beta}_k\}\mid \{\sigma_k^2\}, \boldsymbol{\pi}, \{Z_{ik}\}, K,\mathbf{y})f( \{\sigma_k^2\}\mid \boldsymbol{\pi}, \{Z_{ik}\}, K, \mathbf{y})\notag\\ &\times f(\boldsymbol{\pi}\mid \{Z_{ik}\}, K, \mathbf{y})g(\{Z_{ik}\}, K\mid \mathbf{y}, \mathcal{C}(\textbf{y})),
\end{align*}
where the first three distributions have the same known expressions as in the unconstrained model. The conditional distribution of \(\{Z_{ik}\} \mid K, \mathbf{y},\mathcal{C}(\textbf{y})\) is
\begin{align}
\nonumber
    g(\{Z_{ik}\},K\mid\mathbf{y}, \mathcal{C}(\textbf{y}))= \frac{f(\{Z_{ik}\},K,\mathbf{y})/\mathcal{N}}
    {\sum_{w=1}^W f(\{Z_{ik}^{[w]}\},K,\mathbf{y})/\mathcal{N}}\,
    \mathbb{I}\!\left((\{Z_{ik}\},K)\in\mathcal{C}(\textbf{y})\right)\\
    = \frac{f(\{Z_{ik}\},K,\mathbf{y})}
    {\sum_{w=1}^W f(\{Z_{ik}^{[w]}\},K,\mathbf{y})}\,
    \mathbb{I}\!\left((\{Z_{ik}\},K)\in\mathcal{C}(\textbf{y})\right),
    \label{eq:Z_constrained_cat}
\end{align}
so that for \(w = 1, \dots, W\), $\Pr(\{Z_{ik}\}=\{Z_{ik}^{[w]}\},K = K^{[w]}\mid \mathbf{y}, \mathcal{C}(\textbf{y}))
    = \frac{f(\{Z_{ik}^{[w]}\},K,\mathbf{y})}{\sum_{u=1}^W f(\{Z_{ik}^{[u]}\},K,\mathbf{y})}$, which we can compute for small \(W\).


\subsection{Constructing the Candidate Set}\label{gmm:ss:labels}

\subsubsection{{The DS-Const Approach}}
Our first approach for constructing \(\mathcal{C}(\textbf{y})\){, which we refer to as ``DS-Const,'' implements} the unconstrained {NIG} model with $K$ large (e.g., $K = 25$) with a Gibbs sampler using a random subset of the data and then uses the resulting posterior draws to generate candidate partitions for the full dataset. Let \(\mathbf{y}_{\text{sub}}\)  denote the \(n_{\text{sub}}\)-dimensional vector of observed data obtained by sampling \(n_{\text{sub}} << N\) observations. In our simulation study we consider \(n_{\text{sub}} \in \{0.1N, 0.25N, 0.5N\}\). For details on the full-conditional distributions, and detailed step-by-step instructions for this strategy see Appendix \ref{gmm:subgibbs}.

\subsubsection{Machine Learning Algorithms}\label{ss:gmm:ml}
For our second approach to constructing a candidate set of labels, {which we refer to as ``DS-ML,'' }we use various machine learning clustering algorithms fit over a range of \(\{2,.\dots,K\}\). In the intercept only setting we use $k$-means \citep{hartigan1979algorithm, ahmed2020k}, $k$-medoids \citep{park2009simple}, hierarchical clustering \citep{ward1963hierarchical, murtagh2014ward}, and the expectation maximization (EM) algorithm \citep{dempster1977maximum}. In the regression setting, we use $k$-means, the EM algorithm, and spectral clustering \citep{ng2001spectral}. Here, we note that many other clustering algorithms could be used in either setting. These methods are used primarily as a computationally efficient approach for constructing \(\mathcal{C}(\textbf{y})\), allowing us to quickly obtain a candidate set and obtain comparable or improved inference with improved computational speed. An overview of each algorithm along with detailed pseudocode for each algorithm is provided in Appendix \ref{gmm:appen:ml}. For each ML method, we fit the algorithm for \(k=2,\dots,K\) to obtain \((K-1)\times\) (number of ML algorithms) candidate allocations. We then use \eqref{eq:Z_constrained_cat} to sample from the constrained set of labels $\mathcal{C}(\textbf{y})$.

 {The DS-ML approach uses ML algorithms to constrain the parameter space of the NIG model. The ML algorithms can be implemented in several ways. For example, in the simulations, we use the posterior mode of DS-Const to initialize $k$-means for each $K=2,\ldots,M$ to create the candidate set for direct sampling in the intercept-only setting. We refer to this strategy as DS-Const-MAP. Because some ML clustering algorithms are sensitive to initialization information from DS-Constr may provide a more stable and informed starting point.}

\section{Simulation Study}\label{gmm:s:simstudy}
\subsection{Set Up}\label{gmm:ss:sim_setup}
In our simulation study, we consider two settings. The first is an intercept only Gaussian mixture model and the second is a Gaussian mixture of regressions with one covariate. In both settings, we simulated from a \(K=3\) component mixture with homogeneous variance across components, setting \(\sigma_k^2=0.25\) for \(k=1,2,3\), and  consider an equal mixing proportions setting \(\boldsymbol{\pi} = (\frac{1}{3}, \frac{1}{3}, \frac{1}{3})\).

In the intercept only setting, we generate $\textbf{z}_i \sim \text{Multinomial}(1 \mid \bm{\pi})$ and outcomes are generated from \(Y_i \mid Z_{ik} = 1 \sim \text{Normal}(\mu_k, \sigma_k^2)\). To vary the degree of cluster separation, we consider three specifications for the component means \(\boldsymbol{\mu} = (\mu_1, \mu_2, \mu_3)\). The large overlap cluster separation has \(\boldsymbol{\mu} = (-\frac{1}{3}, 0, \frac{1}{3})\), the moderate overlap has \(\boldsymbol{\mu} = (-1, 0,1)\), and the no overlap has \(\boldsymbol{\mu} = (-2, 0, 2)\).

In the regression setting, we generate covariates \(\mathbf{x}_i = (1, x_{2,i})^{\prime}\) with \(x_{2,i}\sim \text{Normal}(0,1)\), $\textbf{z}_i \sim \text{Multinomial}(1 \mid \bm{\pi})$, and outcomes are simulated according to \(Y_i \mid Z_{ik} = 1 \sim \text{Normal}(\beta_{k,0}+x_{2,i}\beta_{k,1},\sigma_k^2)\). To control separation between regression components, we vary the intercept vector \(\boldsymbol{\beta}_0 = (\beta_{1,0}, \beta_{2,0}, \beta_{3,0})\). For the large overlap scenario \(\boldsymbol{\beta}_0 = (-\tfrac{1}{3},0,\tfrac{1}{3})\), for the moderate overlap scenario \(\boldsymbol{\beta}_0 = (-1,0,1)\), and for the no overlap scenario \(\boldsymbol{\beta}_0 = (-2,0,2)\). Appendix \ref{gmm:sim_data} provides an illustration of the overlap levels for both simulation settings. 

We use a common set of hyperparameters across all methods in this simulation study. For the mixing proportions we use a symmetric Dirichlet prior with \(\alpha_1 = \dots = \alpha_K = 1\). For the regression coefficients, we fix \(\sigma_{\beta}^2 = 100\). For the component variances, we assume an inverse gamma prior with shape \(\omega = 2\) and scale \(\kappa = 1\). For the direct sampler, we place a discrete uniform prior on \(K \in \{1,\dots,M\}\), so that \(\rho_m = \Pr(K=m)=1/M\) for \(m=1,\dots,M\). We take \(M=10\) when \(N=10\) and \(M=25\) in the larger-\(N\) settings. For the non-direct-sampling approaches, we fix the number of available components at \(K=25\).

\subsection{{Illustrating the Constrained Statistical Models}}
In this section, we illustrate the differences in our constrained statistical models using an example from the intercept only simulation. {Part of the goal of this section is to demonstrate why considering ML algorithms to define the constraint is useful.} Figure \ref{gmm:fig:compare_constr} compares the true partition to the top five candidates { as ranked via the marginal posterior probability} in the reduced space \(\mathcal{C}(\textbf{y})\). Under the Gibbs/subset construction, the sampling probabilities are highly sensitive to small changes in the candidate partitions which leads to one (or a few) candidates typically receiving most of the mass, while the remaining candidates have probabilities near zero. {Thus}, posterior draws of the labels concentrate tightly around {a few} candidate{s}, which in many cases collapses into a single partition. Importantly, this behavior is not necessarily bad. The dominant partition is the sample marginal posterior mode obtained from fitting the unconstrained model using a subset of the data, so it typically represents a sensible clustering, even if it reflects less uncertainty than {that of DS-ML}. 
 
When \(\mathcal{C}(\textbf{y})\) is formed using machine learning clustering algorithms, this {pattern that a single candidate dominates the posterior mass is} less pronounced. Multiple algorithms produce partitions that are close to the truth, and these candidates tend to receive similar probabilities (e.g., approximately 0.3 in this example). So, posterior label draws are distributed over a small set of high quality candidates rather than being dominated by a single {candidate}. The ML construction also provides a way of selecting the number of clusters which is typically unknown in practice. We fit each algorithm over \({K} \in \{2, \dots, {25}\}\) {with 5 ML algorithms,} and the candidates with the highest probability correspond to \(K = 3\), which matches the true number of clusters in this example.

\begin{figure}[!h]
    \centering
    \captionsetup{font={stretch=1}}
    \begin{subfigure}[b]{0.45\linewidth}
        \centering
        \includegraphics[width=\linewidth]{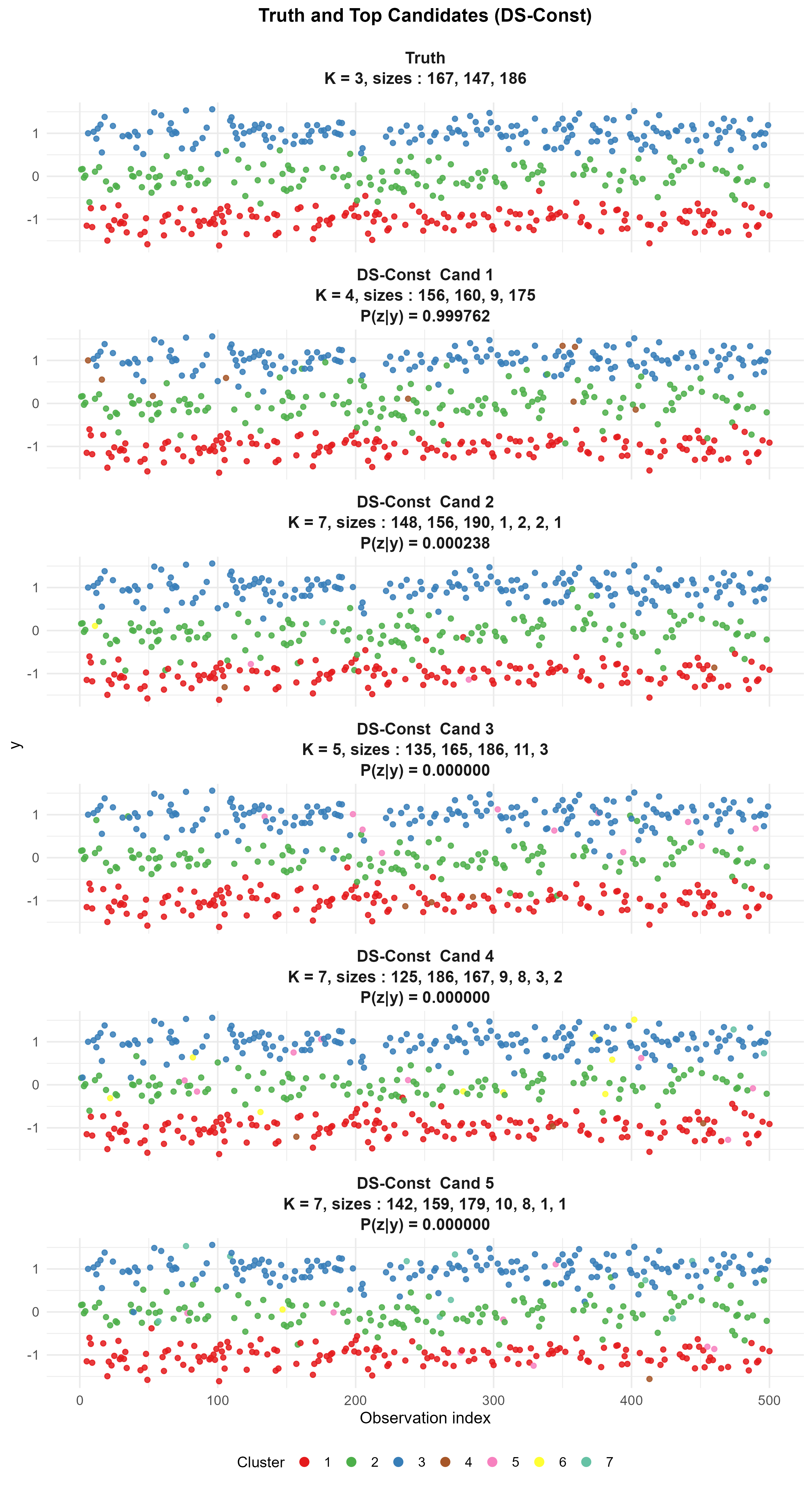}
    \end{subfigure}
    \hfill
    \begin{subfigure}[b]{0.45\linewidth}
        \centering
        \includegraphics[width=\linewidth]{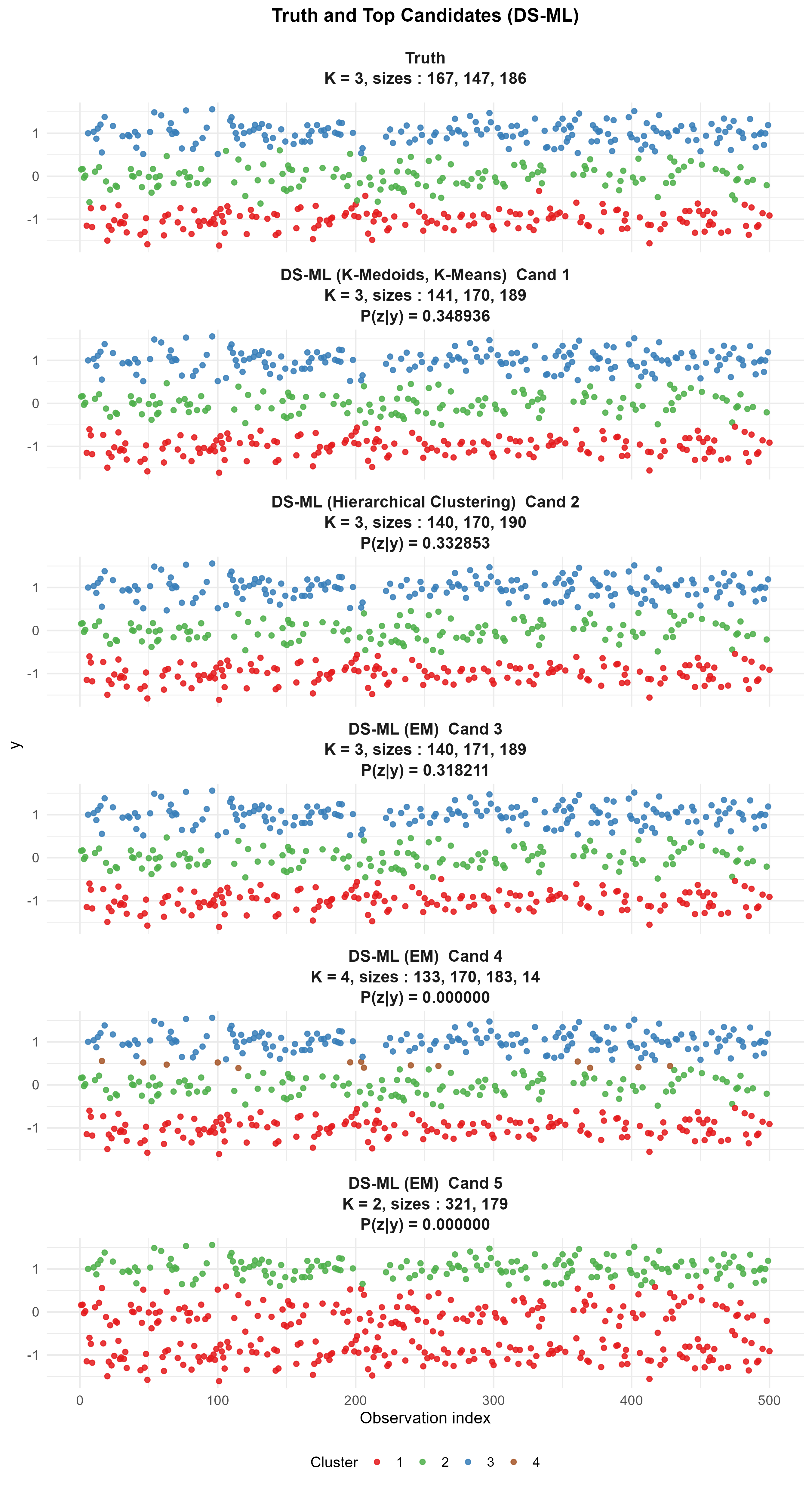}
    \end{subfigure}
    \caption{Comparison of the true partition (top panel) to the top five candidates in the reduced label space \(\mathcal{C}(\textbf{y})\). The left column shows candidates obtained from the Gibbs/subset construction, and the right column is the candidates from the ML construction. For each candidate, the subtitle reports the number of occupied components and corresponding cluster sizes, along with the sampling probability calculated using  Equation (\ref{eq:Z_constrained_cat}).}
    \label{gmm:fig:compare_constr}
\end{figure}

\subsection{Intercept Only Simulations}\label{gmm:ss:y_only}
To evaluate the different approaches in the intercept only scenario, we compare the mean  of the adjusted Rand index (ARI), the effective sample size (ESS) for the first moment of the continuous parameters, the Kolmogorov--Smirnov (KS) statistic, and central processing unit (CPU) time measured in seconds. {All MCMC methodologies have chains of length 1,000.} For each posterior replicate, we compute the ARI \citep{hubert1985comparing} using the \texttt{ARI()} function in the \texttt{salso} R package \citep{dahl2022search}, and summarize it across replicates by its mean. For the KS statistic, assume we have \(j = 1,\dots,n_{\text{samp}}\) posterior draws, \(k = 1,\dots,K\) mixture components, and a uniform grid \(\{t_g\}_{g=1}^{G}\). For draw \(j\), we have replicates of the component specific parameters \(\pi_{k}^{[j]}, \mu_{k}^{[j]}, \sigma_k^{[j]}\). The mixture cumulative distribution function (CDF) for draw \(j\) is \(F^{[j]}(t_g) = \sum_{k = 1}^K \pi^{[j]}_k\Phi\!\left(\frac{t_g - \mu_{k}^{[j]}}{\sigma_{k}^{[j]}} \right)\), \(g=1,\dots,G\), and the posterior mean CDF is \(\widehat{F}(t_g) = \frac{1}{n_{\text{samp}}}\sum_{j=1}^{n_{\text{samp}}}F^{[j]}(t_g)\). We compute the empirical CDF of the data as \( F_{\text{emp}}(t_g) = \frac{1}{N}\sum_{i=1}^N \mathbb{I}(Y_i \le t_g)\), so the KS distance between the fitted and empirical CDFs is \(KS = \max_g\left|\widehat{F}(t_g) - F_{\text{emp}}(t_g)\right|\). 

\begin{table}[h]
\captionsetup{font={stretch=1}}
\caption{Mean over 50 replicates of mean ARI, ESS of the first moment, KS statistic, and CPU time (seconds) for each sampling strategy in the \(N = 10{,}000\) scenario. The DS-Const and DS-Const-MAP approaches fit the initial Gibbs sampler with a subset size \(n_{\text{sub}} = 0.25N\).}
\label{tab:int_only_sim}
\centering
{\renewcommand{\baselinestretch}{1}\selectfont
\begin{tabular}{ccccc}
\toprule
Method & Mean ARI & ESS & KS & CPU Time (s) \\
\midrule

\multicolumn{5}{c}{\textbf{Large overlap}} \\
DS-ML & 0.522 & {1000} & 0.032 & 355.9 \\
DS-Const & 0.425 & {1000} & 0.006 & 232.0 \\
DS-Const-MAP & 0.652 & {1000} & 0.019 & 230.7 \\
MC-MCMC & 0.330 & 1016 & 0.005 & 1087.1 \\
MCMC & 0.328 & 994 & 0.005 & 879.7 \\
\midrule

\multicolumn{5}{c}{\textbf{Moderate overlap}} \\
DS-ML & 0.909 & {1000} & 0.007 & 350.9 \\
DS-Const & 0.824 & {1000} & 0.005 & 232.4 \\
DS-Const-MAP & 0.910 & {1000} & 0.007  & 231.1 \\
MC-MCMC & 0.722 & 990 & 0.004 & 1091.1 \\
MCMC & 0.772 & 1002 & 0.004 & 881.1 \\
\midrule

\multicolumn{5}{c}{\textbf{No overlap}} \\
DS-ML & 1.000 & {1000} & 0.004 & 320.5 \\
DS-Const & 0.993 & {1000} & 0.004 & 235.3 \\
DS-Const-MAP & 1.000 & {1000} & 0.004 & 234.6 \\
MC-MCMC & 0.974 & 1008 & 0.004 & 1104.6 \\
MCMC & 0.971 & 1017 & 0.004 & 880.6 \\

\bottomrule
\end{tabular}
}
\end{table}

Table \ref{tab:int_only_sim} summarizes the evaluation metrics, averaged over \(50\) independent replicates, for the moderate, large, and no overlap settings for sample size \(N =10{,}000\). Simulation results for the smaller \(N\) settings and additional evaluation metrics such as the ESS of the second moment, the ARI standard deviation, and the \(L_2\) norm are presented in Appendix \ref{apppen:y_only}. In all overlap scenarios, the DS-constrained approaches presented in the table use \(n_{\text{sub}} = 0.25N\). Results for other subset sizes can be found in Appendix \ref{apppen:y_only}, and they show that increasing the subsample size used for the initial Gibbs sampler generally improves performance metrics at the cost of slightly increased CPU time. 

Across all methods, ARI values are lowest in the large overlap scenario and then improve as the separation between components increases, which is expected because when the true clusters are closer together, observations near cluster boundaries are more difficult to label. In the moderate overlap scenario, all constrained approaches (DS-ML, DS-Const, and DS-Const-MAP) achieve higher ARI values than the MCMC-based implementations. This is because the constrained approaches restrict the extremely large discrete parameter space to a much smaller set of plausible labelings, whereas the MCMC-based implementations must still search through the full space. The no overlap scenario is the easiest for clustering because the components are well separated, so all methods achieve near perfect clustering performance, with ARI values close to \(1\).

Density estimation performance is also very similar across methods in the moderate and no overlap scenarios. In the large overlap scenario, the constrained methods show slightly worse density estimation performance than the MCMC methods. Since these approaches restrict the label space, it is possible that the candidate set does not always include the best labeling when there is large overlap. One way to address this issue is to expand the candidate set, for example, by incorporating additional ML algorithms or additional initializations to construct a richer set of candidate partitions.

The effective sample size is near \(1000\) for all methods and overlap settings. From the CPU times, we see that this ESS is achieved much faster with the constrained models than with the unconstrained models implemented via MCMC. The MCMC methods required more iterations, including burn in/thinning, whereas for the constrained models, once the candidate set is constructed, only \(1000\) iterations are needed to produce \(1000\) independent replicates from the closed form conditional distributions. Both the DS-constrained approaches (with the Gibbs subsample step) and the DS-ML approach offer clear computational advantages relative to the MCMC-based implementations. 

\subsection{Simulations with Covariates}\label{gmm:ss:xy_sims}
To evaluate the different approaches in the regression scenario, we compare similar metrics as in the intercept only scenario. In this simulation we compute the KS statistic in a different way than the previous section because the grid based computation scales with \(N \times G \times K\times n_{samp}\), becoming very slow for large \(N\). Instead, we compute the KS statistic using an observation based calculation instead of a grid based calculation, similar to \citet{skhosana2024modified}. 

Let \(\{(\mathbf{x}_i,Y_i)\}_{i=1}^{N}\) denote the observed sample. For posterior draw \(j=1,\dots,n_{\text{samp}}\) and component \(k=1,\dots,K\), we have \(\pi_k^{[j]}, \boldsymbol{\beta}_k^{[j]}, \sigma_k^{[j]}\), with component mean
\(\mu_{ik}^{[j]} = \mathbf{x}_i^{\prime}\boldsymbol{\beta}_k^{[j]}\). The conditional mixture CDF for draw \(j\) evaluated at \((\mathbf{x}_i,Y_i)\) is \(F^{[j]}(Y_i\mid \mathbf{x}_i) = \sum_{k=1}^K \pi_k^{[j]}\, \Phi\!\left(\frac{Y_i -\mathbf{x}_i^{\prime}\boldsymbol{\beta}_k^{[j]}}{\sigma_k^{[j]}}\right)\) where \(\Phi\) is the standard Normal CDF. Define the posterior mean CDF \( \widehat F(Y_i\mid \mathbf{x}_i)= \frac{1}{n_{\text{samp}}}\sum_{j=1}^{n_{\text{samp}}} F^{[j]}(Y_i\mid \mathbf{x}_i)\).  The true conditional CDF is \(F_{\text{true}}(Y_i\mid \mathbf{x}_i) = \sum_{k=1}^K \pi_{k,\text{true}}\, \Phi\!\left(\frac{Y_i - \mathbf{x}_i^{\prime}\boldsymbol{\beta}_{k,\text{true}}}{\sigma_{k,\text{true}}}\right).\)   The KS statistic is computed as a maximum over observations \( KS= \max_{i=1,\dots,N}\left| F_{\text{true}}(Y_i\mid \mathbf{x}_i) - \widehat F(Y_i\mid \mathbf{x}_i)\right|.\) 

\begin{table}[!h]
\captionsetup{font={stretch=1}}
\caption{Mean over 50 replicates of mean ARI, ESS of the first moment, KS statistic, and CPU time (seconds) for each sampling strategy in the \(N = 10{,}000\) scenario.}
\label{tab:xy_sim}
\centering
{\renewcommand{\baselinestretch}{1}\selectfont
\begin{tabular}{ccccc}
\toprule
Method & Mean ARI & ESS & KS & CPU Time (s)\\
\midrule

\multicolumn{5}{c}{\textbf{Large overlap}} \\
DS-ML & 0.669 & {1000} & 0.023 & 325.0 \\
DS-Const & 0.458 & {1000} & 0.016 & 353.6 \\
MC-MCMC & 0.4397 & 1013 & 0.009 & 8031.7\\
MCMC & 0.402 & 1031 & 0.009 & 2927.4 \\
\midrule

\multicolumn{5}{c}{\textbf{Moderate overlap}} \\
DS-ML & 0.910 & {1000} & 0.011 & 268.2 \\
DS-Const & 0.781 & {1000} & 0.012 & 373.3 \\
MC-MCMC & 0.770 & 996 & 0.008 & 8102.3 \\
MCMC & 0.767 & 1034 & 0.008 & 2962.9 \\
\midrule

\multicolumn{5}{c}{\textbf{No overlap}} \\
DS-ML & 1.000 & {1000} & 0.007 & 218.1 \\
DS-Const & 0.986 & {1000} & 0.008 & 394.5 \\
MC-MCMC & 0.987 & 1014 & 0.008 & 8103.5 \\
MCMC & 0.988 & 1002 & 0.008 & 2942.5 \\

\bottomrule
\end{tabular}
}
\end{table}

Table \ref{tab:xy_sim} summarizes the evaluation metrics, averaged over 50 independent replicates, for each overlap setting with \(N = 10{,}000\). Simulation results for \(N \in \{10, 500,1000\}\) and additional evaluation metrics can be found in Appendix \ref{appen:xy_sims}. {Recall in the regression case we include an intercept and a single normal generated covariate.} Overall, these results tell a similar story to the intercept-only scenario. In this simulation, we omit the DS-Const-MAP approach because the $k$-means clustering with MAP initialization performs poorly in the regression setting. 

Across all overlap scenarios, DS-ML {(enumerated over $K = 2,\ldots, 25$ and 3 ML algorithms leading to 72 candidates)} achieves the best clustering performance as measured by the ARI. All methods perform well in the no overlap scenario, and clustering performance declines as the amount of overlap increases. The MCMC-based implementations consistently yield the lowest ARI values. The constrained methods restrict the posterior label draws to a smaller candidate set. When this candidate set contains partitions close to the truth, posterior draws from the reduced space tend to yield consistently high ARI values. The unconstrained MCMC sampler must explore a much larger label space, so some posterior draws correspond to good partitions (high ARI) while others are farther from the truth (low ARI). This leads to a lower mean ARI for MCMC relative to the constrained samplers.

Density estimation is similar across the constrained and unconstrained methods in the moderate and no overlap settings. Similar to the intercept only scenario, the constrained approaches have slightly worse density estimation performance in the large overlap setting. This may be due to the restriction of the label space.  However, the constrained approaches perform better in terms of clustering, so in the large overlap setting there is a trade-off depending on whether clustering performance or density estimation is of greater interest. 

The constrained approaches also offer a clear computational advantage over both MCMC-based implementations, which require more iterations to achieve a comparable ESS for the first and second moments. MC-MCMC is particularly slow relative to the standard Gibbs sampler because it updates labels one observation at a time using the full conditional in \eqref{gmm:eq:dsmcmc}, which requires repeated evaluation of the joint distribution in \eqref{eq:gmm:zXy} for each possible reassignment. Appendix \ref{gmm:crab_analysis} presents a real data illustration for the regression setting with \(N = 200\) that closely resembles this simulation setup.

\section{Birth Weight Data Application}\label{gmm:s:birthweight}
In this application, we analyze 2024 U.S. natality data consisting of \(N = 3{,}408{,}680\) observations obtained from \url{https://www.cdc.gov/nchs/data_access/vitalstatsonline.htm}. Gaussian FMMs have been used previously in the birth weight literature. Several studies use a {G}aussian mixture model with two components to identify implausible birth weights or gestational ages, or to correct gestational age miss classification in preterm birth data \citep{tentoni2004birthweight, platt2001detecting, urquia2012mixture}. Other work has used FMMs to study the joint distribution of gestational age and birth weight \citep{schwartz2010joint} or to determine the number of mixture components from the data \citep{charnigo2010thinking}. The scale of the existing birth weight mixture modeling analyses varies, but the largest samples sizes considered are only on the scale of hundreds of thousands \citep{schwartz2010joint}.  As mentioned, the data set we consider is much larger. 

We first consider an intercept only analysis in which we use birth weight, measured in grams, as the response variable, with the goal of determining whether there are multiple latent subgroups. We also consider an analysis in which we include gestational age, measured in weeks, as a covariate. Prior studies have shown that birth weight distributions are heterogeneous and can be modeled using Gaussian FMMs \citep{oja1991fitting, charnigo2010thinking, gage2002modeling, slaughter2009bayesian}. We fit the NIG mixture model using the DS-ML approach and compare these results with those from a traditional MCMC implementation. 

To summarize clustering performance, we report the ARI using preterm versus non-preterm status as a clinically motivated reference classification, since prior work has shown that gestational age is closely tied to heterogeneity in birth weight distributions and that lower weight components are associated with preterm births \citep{wilcox1983birthweight}. We emphasize that preterm status is used here as a clinically meaningful proxy for comparison, rather than as the literal true latent class structure. We define preterm births as those occurring before 37 completed weeks of gestation, and non-preterm births as those occurring at 37 or more completed weeks \citep{goldenberg2008epidemiology}. For  DS-ML we set the maximum number of components \(M = 25\) and for MCMC we consider the maximum number of components \(M \in \{2, 5, 25\}\). 

\begin{table}[H]
\caption{ARI, ESS, KS, and CPU time for birth weight data.}
\label{table:data_anal_int_only}
\centering
{\renewcommand{\baselinestretch}{1}\selectfont
\begin{tabular}{cccccc}
\toprule
Method & Mean ARI & ESS & KS & M & CPU Time \\ \midrule
DS-ML & 0.429 & 1000 & 0.016 & 25 & 1.29 hrs \\
MCMC & 0.0001 & 904.5 & 0.004 & 25 & 4.21 days \\
MCMC & 0.021 & 1000 & 0.005 & 5 & 3.41 days \\
MCMC & 0.256 & 1000 & 0.008 & 2 & 3.27 days \\
\bottomrule
\end{tabular}
}
\end{table}

The results in Table \ref{table:data_anal_int_only} show that MCMC yields smaller KS values (using preterm versus non-preterm as the reference classification) than the DS-ML, but worse agreement with the clinically motivated preterm/non-preterm reference and significantly longer computation times. The smaller KS values indicate that the overfitted MCMC model provides a closer approximation to the empirical CDF of the birth weight data. The corresponding ARI values suggest that these fits align poorly with the clinical reference, particularly when many components are allowed, because observations in the non-preterm group are spread across several overlapping mixture components. The DS-ML approach produces more clinically interpretable partitions and has a significantly faster computation time than the MCMC implementation. 
\begin{figure}[!b]
    \centering
    \captionsetup{font={stretch=1}}
    
    \begin{subfigure}[b]{0.48\textwidth}
        \centering
        \includegraphics[width=\textwidth]{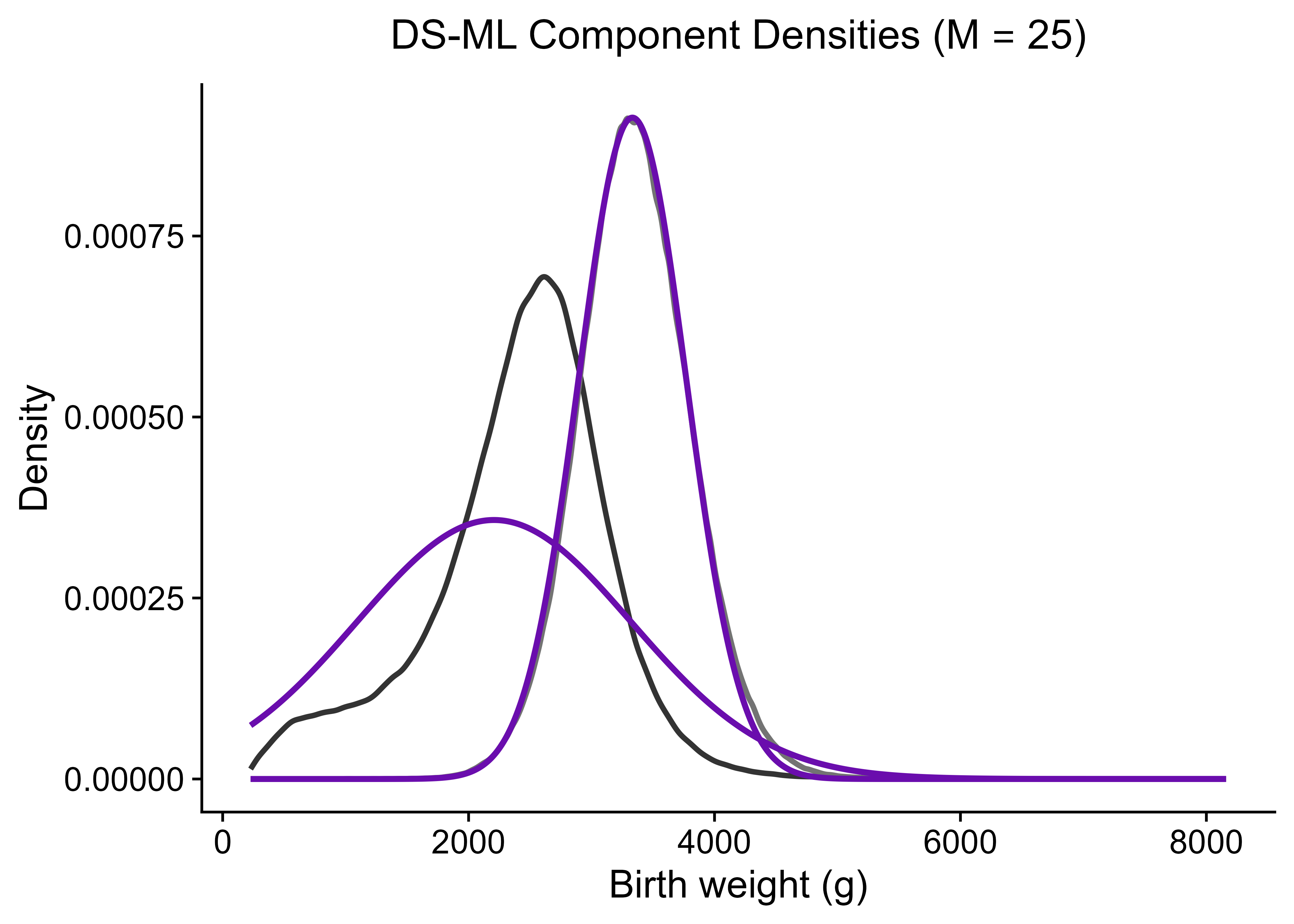}
    \end{subfigure}
    \hfill
    \begin{subfigure}[b]{0.48\textwidth}
        \centering
        \includegraphics[width=\textwidth]{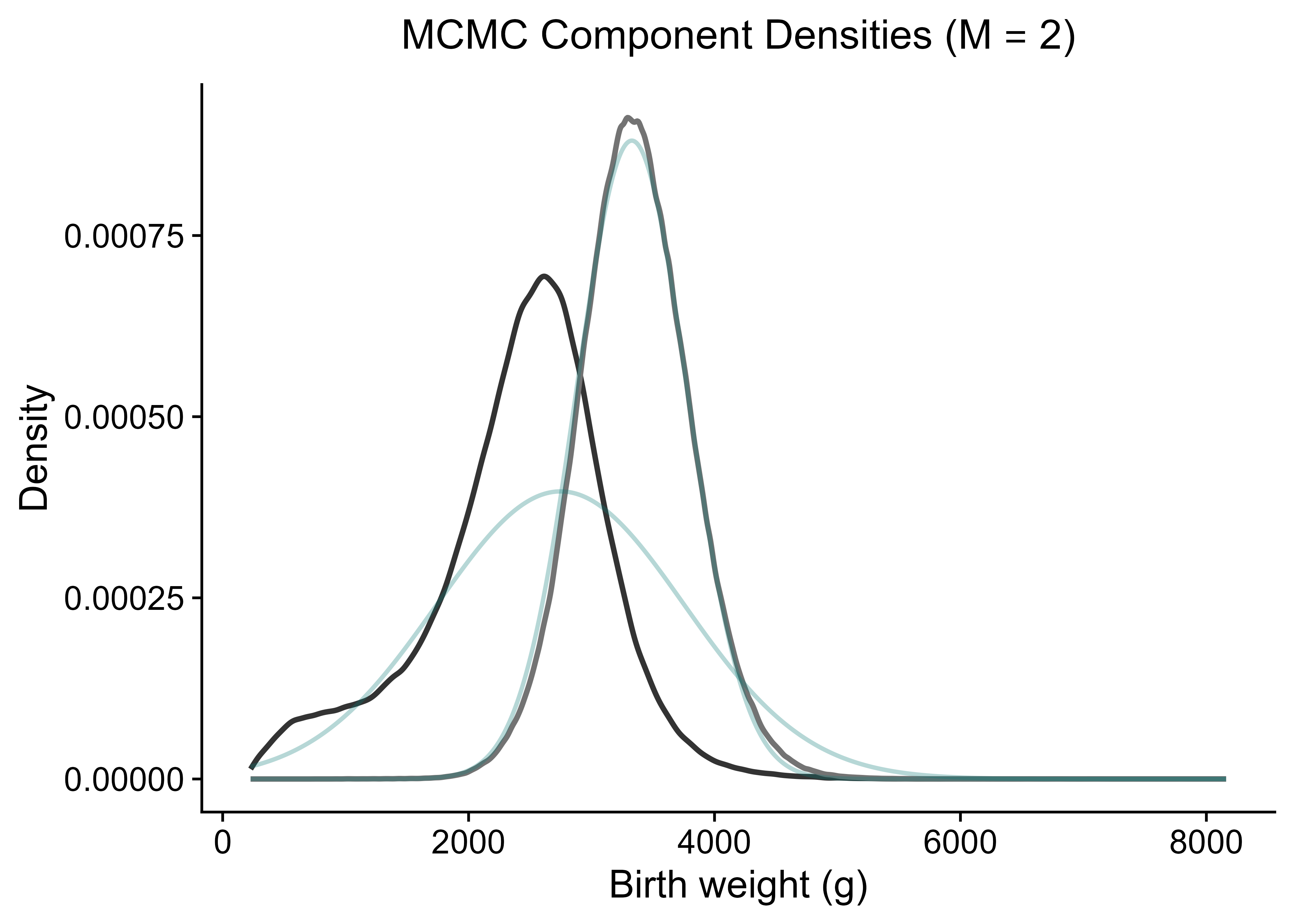}
    \end{subfigure}
    
    \vspace{0.5cm}
    
    \begin{subfigure}[b]{0.48\textwidth}
        \centering
        \includegraphics[width=\textwidth]{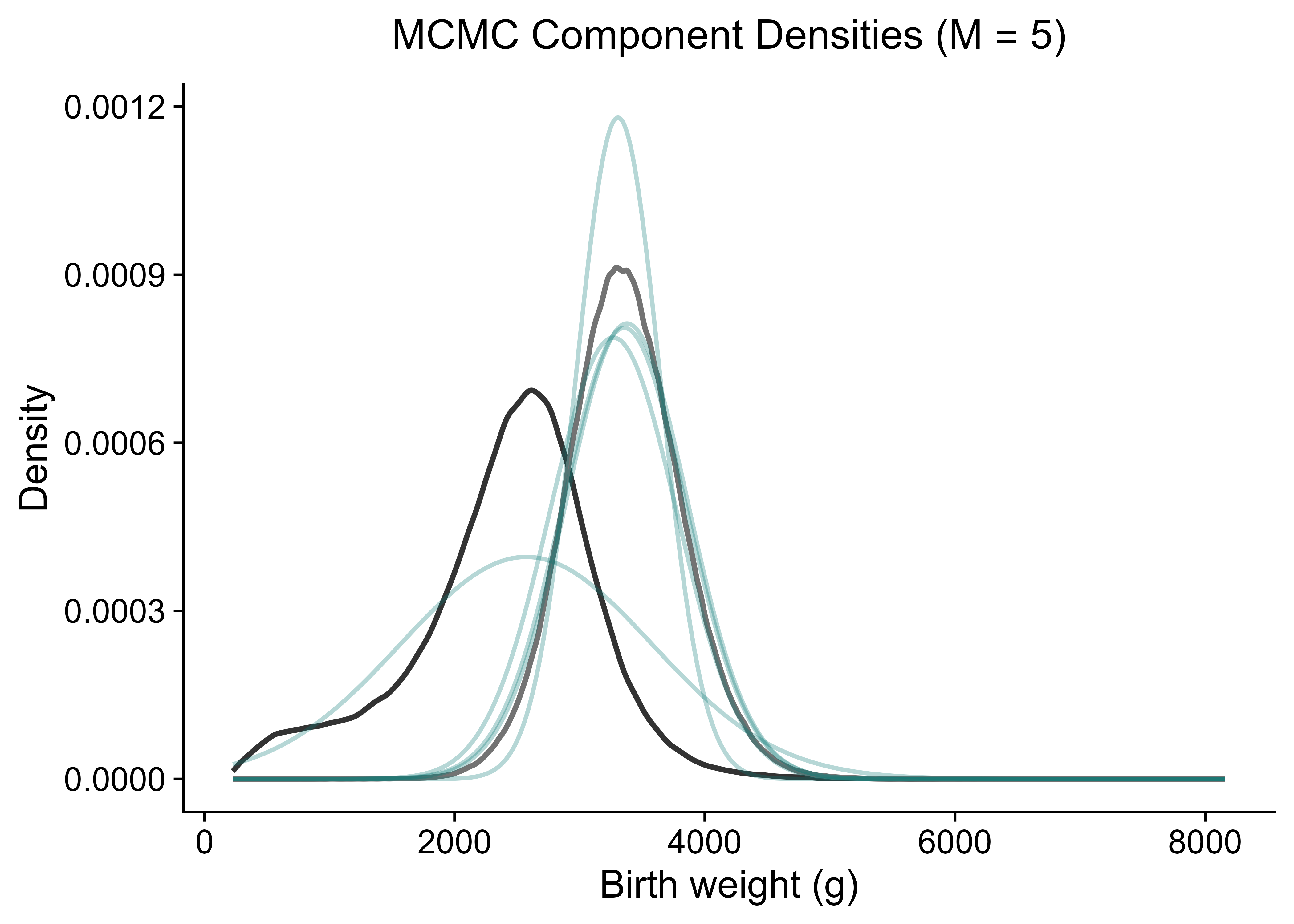}
    \end{subfigure}
    \hfill
    \begin{subfigure}[b]{0.45\textwidth}
        \centering
        \includegraphics[width=\textwidth]{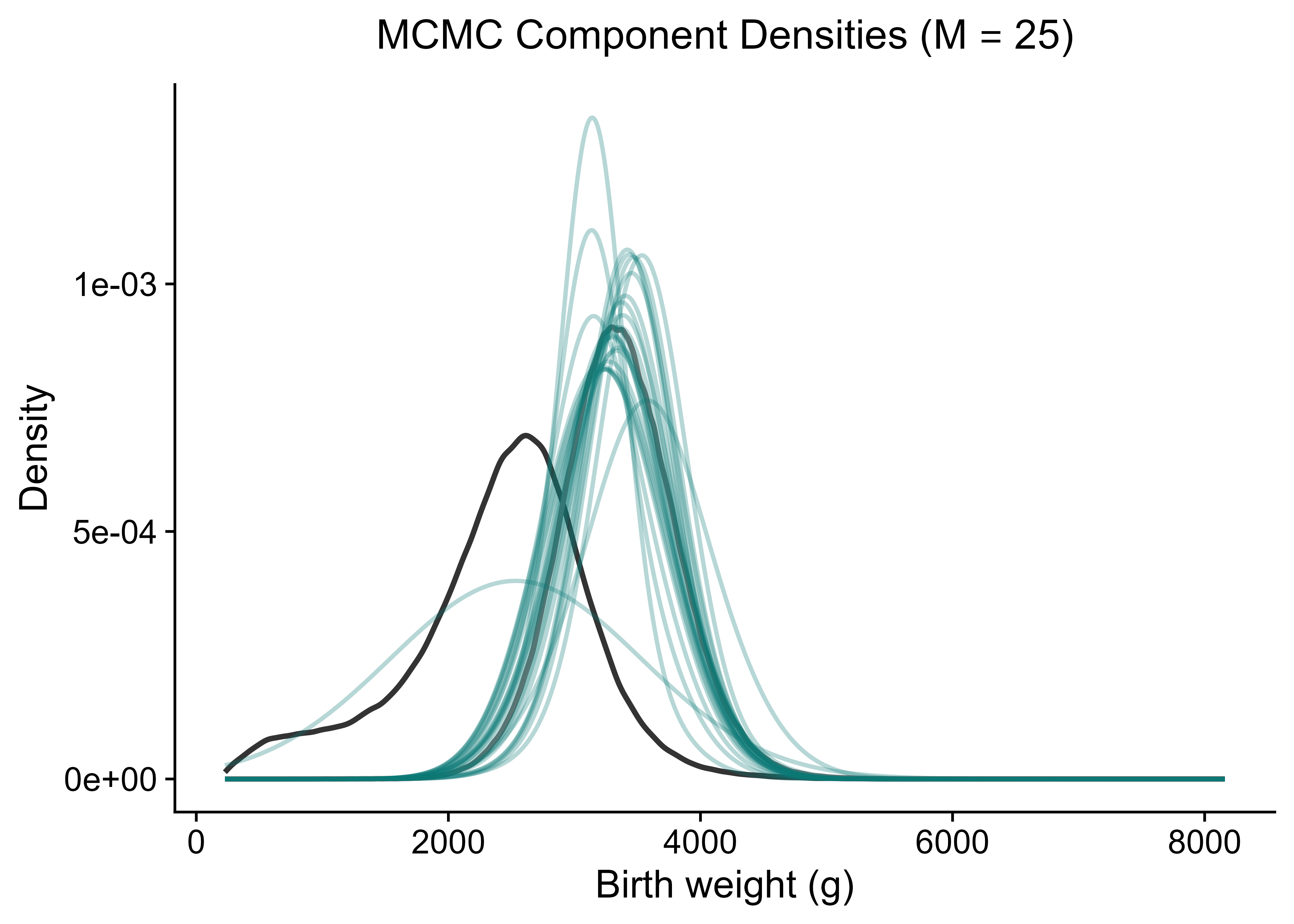}
    \end{subfigure}
    
    \caption{Comparison of with the empirical birth weight densities for the preterm and non-preterm groups. At each posterior draw, the mixture components with nonzero weights are selected, ordered by their means, and their unweighted  densities are computed. The plotted component curves are then obtained by averaging these densities across posterior draws at each grid point.}
    \label{fig:data_anal_int_only}
\end{figure}

Figure \ref{fig:data_anal_int_only} compares the posterior mean unweighted component densities with the empirical birth weight densities for the preterm and non-preterm groups. Across all fitted models, the dominant structure appears to consist of roughly two main components. There is a lower birth weight component associated with the preterm group and a higher birth weight component associated with the non-preterm group. This two group pattern is also visible in the overfitted MCMC models with larger values of \(M\). The overfitted model appears to divide the higher weight group into multiple overlapping components, which is consistent with the lower ARI reported in Table \ref{table:data_anal_int_only}. 

For the regression analysis, we fit DS-ML using various combinations of covariates. Full results, including comparison with MCMC, are provided in Appendix \ref{appen:sensitivity}. Among the different set of covariates we considered, using gestational age alone as a covariate resulted in the strongest recovery of the preterm/non-preterm clusters, so we focus on that specification here. With $M=25$, DS-ML achieved a mean ARI of 0.776 with a CPU time of approximately 4.56 hours. The MCMC sampler in the regression scenario again took several days to run (see Appendix \ref{appen:sensitivity} for the full comparison and discussion). These results further illustrate the computational benefits of using DS-ML in this application. As in the intercept-only scenario, DS-ML recovers interpretable clinical groups in hours rather than days, providing a favorable balance of computational efficiency and clustering performance.

\section{Discussion}\label{gmm:s:discuss}
In this paper, we introduced new sampling strategies for Gaussian FMMs. First, we proposed a method of composition with MCMC where necessary (MC-MCMC) strategy, which samples the mixture labels via a Gibbs sampler and then samples the continuous parameters from their conditional distributions derived under the NIG model. Second, we proposed a constrained approach that constructs a candidate set of plausible label allocations and  performs posterior inference conditional on \(\{Z_{ik}\}\) lying in this reduced set. We compared both methods (and the direct sampler for small \(N\)) to a traditional Gibbs sampler implementation. Across many scenarios, the constrained approach had improved clustering performance and significant computational advantages, while the MC-MCMC performed similarly to the standard MCMC. We further illustrated our constrained sampler using 2024 U.S. natality data from the CDC. In this application with an extremely large sample size, we show the DS-ML approach was computationally efficient compared to a Gibbs sampler, fitting the mixture models in hours instead of days while still being able to recover scientifically meaningful groups in the birth weight data.

\section*{Appendix A: Proof of Proposition 1}
From Proposition 1 in the Main Text, recall the joint posterior distribution of \(\{\boldsymbol{\beta}_k\}, \boldsymbol{\pi}, \{\sigma_k^2\}, \{Z_{ik}\}, K \vert \mathbf{y}\) can be decomposed as follows:   
\begin{align}
	f(\{\boldsymbol{\beta}_k\},& \boldsymbol{\pi}, \{\sigma_k^2\}, \{Z_{ik}\}, K \vert \mathbf{y}) = f(\{\boldsymbol{\beta}_k\}\vert \{\sigma_k^2\}, \boldsymbol{\pi}, \{Z_{ik}\}, K,\mathbf{y})f( \{\sigma_k^2\}\vert \boldsymbol{\pi}, \{Z_{ik}\}, K, \mathbf{y})\notag\\ &\times f(\boldsymbol{\pi}\vert \{Z_{ik}\}, K, \mathbf{y})f(\{Z_{ik}\}\vert K, \mathbf{y}) f( K\vert \mathbf{y}).
	\label{appen:eq:joint}
\end{align}
We can derive each of the conditional distributions on the right hand side of Equation (\ref{appen:eq:joint}). Each of these pdf/pmfs will have a known expression that can be simulated from. Thus, to sample from (\ref{appen:eq:joint}) one simply sequentially draws from each of the conditional distributions on the right-hand-side of (\ref{appen:eq:joint}). First, for simplicity, we re-write the data model from Equation (1) in the Main Text using vectors and matrices.  
\begin{align*}
	f(\mathbf{y}\vert \{\boldsymbol{\beta}_k\}&, \{\sigma_k^2\}, \{Z_{ik}\}, K) \propto \prod_{k = 1, \dots,K, n_k >0}\frac{(\sigma_k^2)^{-n_k/2}}{(2\pi)^{n_k/2}} \exp \left\{- \frac{(\mathbf{y}_k - \mathbf{X}_k\boldsymbol{\beta}_k)^{\prime}(\mathbf{y}_k - \mathbf{X}_k\boldsymbol{\beta}_k)}{2\sigma_k^2} \right\},
\end{align*}
where \(n_k = \sum_{i = 1}^NZ_{ik}\), the \(n_k \times p\) matrix \(\mathbf{X}_k = (\mathbf{x}_i^{\prime}: Z_{ik} = 1, i=1,\dots, N)\), and the \(n_k\)-dimensional vector \(\mathbf{y}_k = (Y_i: Z_{ik} = 1, i = 1,\dots, N)\). Using a proportionality argument and a complete the squares argument, we can derive \(f(\{\boldsymbol{\beta}_k\}\vert \{\sigma_k^2\}, \{Z_{ik}\}, K,\mathbf{y})\) as follows: 
\begin{align*}
	f&(\{\boldsymbol{\beta}_k\}\vert \{\sigma_k^2\}, \{Z_{ik}\}, K,\mathbf{y}) \propto f(\mathbf{y}\vert \{\boldsymbol{\beta}_k\}, \{\sigma_k^2\}, \{Z_{ik}\}, K)f(\{\boldsymbol{\beta}_k\} \vert \{\sigma_k^2\}, \sigma_{\beta}^2) \\ 
	&\propto \prod_{k = 1,\dots, K, n_k = 0}\frac{(\sigma_k^2\sigma_{\beta}^2)^{-p/2}}{(2\pi)^{p/2}}\exp\left(-\frac{\boldsymbol{\beta}_k^{\prime}\boldsymbol{\beta}_k}{2\sigma_k^2\sigma_{\beta}^2}\right) \\ 
	&\times \prod_{k = 1, \dots,K, n_k >0}\frac{(\sigma_k^2)^{-n_k/2}}{(2\pi)^{n_k/2}} \frac{(\sigma_k^2\sigma_{\beta}^2)^{-p/2}}{(2\pi)^{p/2}} \exp \left\{- \frac{(\mathbf{y}_k - \mathbf{X}_k\boldsymbol{\beta}_k)^{\prime}(\mathbf{y}_k - \mathbf{X}_k\boldsymbol{\beta}_k)}{2\sigma_k^2} -\frac{\boldsymbol{\beta}_k^{\prime}\boldsymbol{\beta}_k}{2\sigma_k^2\sigma_{\beta}^2} \right\} \\
	&\propto \prod_{k = 1,\dots, K, n_k = 0}\frac{(\sigma_k^2\sigma_{\beta}^2)^{-p/2}}{(2\pi)^{p/2}}\exp\left(-\frac{\boldsymbol{\beta}_k^{\prime}\boldsymbol{\beta}_k}{2\sigma_k^2\sigma_{\beta}^2}\right)\prod_{k = 1, \dots,K, n_k >0}\frac{(\sigma_k^2)^{-n_k/2}}{(2\pi)^{n_k/2}} \frac{(\sigma_k^2\sigma_{\beta}^2)^{-p/2}}{(2\pi)^{p/2}} \\ 
	&\times\exp \left\{ - \frac{\mathbf{y}^{\prime}_k\mathbf{y}_k}{2\sigma_k^2} - \frac{\boldsymbol{\beta}_k^{\prime}\mathbf{X}_k^{\prime}\mathbf{X}_{k}\boldsymbol{\beta}_k}{2\sigma_k^2} + \frac{\boldsymbol{\beta}_k^{\prime}\mathbf{X}_k^{\prime}\mathbf{y}_k}{\sigma_k^2} - \frac{\boldsymbol{\beta}_k^{\prime}\boldsymbol{\beta}_k}{2\sigma_k^2\sigma_{\beta}^2}\right\} 
\end{align*}
and then collecting the terms in \(\boldsymbol{\beta}_k\) gives
\begin{align*}
	&\propto \prod_{k = 1,\dots, K, n_k = 0}\frac{(\sigma_k^2\sigma_{\beta}^2)^{-p/2}}{(2\pi)^{p/2}}\exp\left(-\frac{\boldsymbol{\beta}_k^{\prime}\boldsymbol{\beta}_k}{2\sigma_k^2\sigma_{\beta}^2}\right)\prod_{k = 1, \dots,K, n_k >0}\frac{(\sigma_k^2)^{-n_k/2}}{(2\pi)^{n_k/2}} \frac{(\sigma_k^2\sigma_{\beta}^2)^{-p/2}}{(2\pi)^{p/2}} \\ 
	&\times\exp \left\{ - \frac{\mathbf{y}^{\prime}_k\mathbf{y}_k}{2\sigma_k^2}  - \frac{1}{2}\boldsymbol{\beta}_k^{\prime} \frac{1}{\sigma_k^2} \left(\mathbf{X}_k^{\prime}\mathbf{X}_k + \frac{1}{\sigma_{\beta}^2}\mathbf{I}_p \right)\boldsymbol{\beta}_k + \frac{\boldsymbol{\beta}_k^{\prime}\mathbf{X}_k^{\prime}\mathbf{y}_k}{\sigma_k^2} \right\} \\
	&\propto \prod_{k = 1,\dots, K, n_k = 0}\frac{(\sigma_k^2\sigma_{\beta}^2)^{-p/2}}{(2\pi)^{p/2}}\exp\left(-\frac{\boldsymbol{\beta}_k^{\prime}\boldsymbol{\beta}_k}{2\sigma_k^2\sigma_{\beta}^2}\right)\prod_{k = 1, \dots,K, n_k >0}\frac{(\sigma_k^2)^{-n_k/2}}{(2\pi)^{n_k/2}} \frac{(\sigma_k^2\sigma_{\beta}^2)^{-p/2}}{(2\pi)^{p/2}} \\ 
	&\times\exp \left\{ - \frac{\mathbf{y}^{\prime}_k\mathbf{y}_k}{2\sigma_k^2}  - \frac{1}{2}\boldsymbol{\beta}_k^{\prime}\boldsymbol{\Sigma}_{P,k}^{-1}\boldsymbol{\beta}_k + \frac{\boldsymbol{\beta}_k^{\prime}\mathbf{X}_k^{\prime}\mathbf{y}_k}{\sigma_k^2} \right\} \\
	&\propto \prod_{k = 1,\dots, K, n_k = 0}\frac{(\sigma_k^2\sigma_{\beta}^2)^{-p/2}}{(2\pi)^{p/2}}\exp\left(-\frac{\boldsymbol{\beta}_k^{\prime}\boldsymbol{\beta}_k}{2\sigma_k^2\sigma_{\beta}^2}\right)\prod_{k = 1, \dots,K, n_k >0}\frac{(\sigma_k^2)^{-n_k/2}}{(2\pi)^{n_k/2}} \frac{(\sigma_k^2\sigma_{\beta}^2)^{-p/2}}{(2\pi)^{p/2}} \\ 
	&\times\exp \left\{ - \frac{\mathbf{y}^{\prime}_k\mathbf{y}_k}{2\sigma_k^2}  - \frac{1}{2}\boldsymbol{\beta}_k^{\prime}\boldsymbol{\Sigma}_{P,k}^{-1}\boldsymbol{\beta}_k +\boldsymbol{\beta}_k^{\prime}\frac{1}{\sigma_k^2}\left(\mathbf{X}_k^{\prime}\mathbf{X}_k + \frac{1}{\sigma_{\beta}^2}\mathbf{I}_p \right)\left(\mathbf{X}_k^{\prime}\mathbf{X}_k + \frac{1}{\sigma_{\beta}^2}\mathbf{I}_p \right)^{-1}\mathbf{X}_k^{\prime}\mathbf{y}_k \right\} \\
	&\propto \prod_{k = 1,\dots, K, n_k = 0}\frac{(\sigma_k^2\sigma_{\beta}^2)^{-p/2}}{(2\pi)^{p/2}}\exp\left(-\frac{\boldsymbol{\beta}_k^{\prime}\boldsymbol{\beta}_k}{2\sigma_k^2\sigma_{\beta}^2}\right)\prod_{k = 1, \dots,K, n_k >0}\frac{(\sigma_k^2)^{-n_k/2}}{(2\pi)^{n_k/2}} \frac{(\sigma_k^2\sigma_{\beta}^2)^{-p/2}}{(2\pi)^{p/2}} \\ 
	&\times\exp \left\{ - \frac{\mathbf{y}^{\prime}_k\mathbf{y}_k}{2\sigma_k^2}  - \frac{1}{2}\boldsymbol{\beta}_k^{\prime}\boldsymbol{\Sigma}_{P,k}^{-1}\boldsymbol{\beta}_k + \boldsymbol{\beta}_k^{\prime}\boldsymbol{\Sigma}_{P,k}^{-1}\boldsymbol{\mu}_{P,k} \right\} 
\end{align*}
Completing the square yields 
\begin{align*}
	&\propto \prod_{k = 1,\dots, K, n_k = 0}\frac{(\sigma_k^2\sigma_{\beta}^2)^{-p/2}}{(2\pi)^{p/2}}\exp\left(-\frac{\boldsymbol{\beta}_k^{\prime}\boldsymbol{\beta}_k}{2\sigma_k^2\sigma_{\beta}^2}\right)\prod_{k = 1, \dots,K, n_k >0}\frac{(\sigma_k^2)^{-n_k/2}}{(2\pi)^{n_k/2}} \frac{(\sigma_k^2\sigma_{\beta}^2)^{-p/2}}{(2\pi)^{p/2}} \\ 
	&\times\exp \left\{ -\frac{1}{2}(\boldsymbol{\beta}_k - \boldsymbol{\mu}_{P,k})^{\prime}\boldsymbol{\Sigma}_{P,k}^{-1}(\boldsymbol{\beta}_k - \boldsymbol{\mu}_{P,k})\right\} 
\end{align*}
where \(\boldsymbol{\mu}_{P,k} = \left(\mathbf{X}_k^{\prime}\mathbf{X}_k + \frac{1}{\sigma_{\beta}^2}\mathbf{I}_p\right)^{-1}\mathbf{X}_k^{\prime}\mathbf{y}_k\) and \(\boldsymbol{\Sigma}_{P,k} = \sigma_k^2\left(\mathbf{X}_k^{\prime}\mathbf{X}_k + \frac{1}{\sigma_{\beta}^2}\mathbf{I}_p \right)^{-1}\). So the known distribution of \(\{\boldsymbol{\beta}_k\}\vert \{\sigma_k^2\}, \boldsymbol{\pi}, \{Z_{ik}\}, K,\mathbf{y}\) is 
\begin{align}
	\boldsymbol{\beta}_k \mid \sigma_k^2,\{Z_{ik}\},K,\mathbf{y}
	\sim
	\begin{cases}
		\text{MVN}\!\left(\boldsymbol{\mu}_{P,k},\,\boldsymbol{\Sigma}_{P,k}\right), 
		& \text{if } n_k > 0,\\[6pt]
		\text{MVN}\!\left(\mathbf{0}_p,\,\sigma_k^2\sigma_{\beta}^2\mathbf{I}_p\right),
		& \text{if } n_k = 0~.
	\end{cases}
	\label{appen:eq:beta_post}
\end{align}
When marginalizing the full likelihood across \(\boldsymbol{\beta}_k\) we obtain the \(f(\mathbf{y}_k\vert\ \{\sigma_k^2\}, \{Z_{ik}\}, K)\). That is for \(n_k > 0\) we have 
\begin{align*}
	\int f&(\mathbf{y}_k \vert \mathbf{X}_{k}\boldsymbol{\beta}_k, \sigma_k^2I_{n_k}) f(\boldsymbol{\beta}_k\vert \mathbf{0}_p, \sigma_k^2\sigma_{\beta}^2\mathbf{I}_p) d\boldsymbol{\beta}_k \\
	&= (2\pi\sigma_k^2)^{-\frac{n_k}{2}}(2\pi\sigma_k^2\sigma_{\beta}^2)^{-\frac{p}{2}}\int \exp \left\{- \frac{1}{2\sigma_k^2}\left(\mathbf{y}_k^{\prime}\mathbf{y}_k  - 2\boldsymbol{\beta}_k^{\prime}\mathbf{X}_k^{\prime}\mathbf{y}_k + \boldsymbol{\beta}_k^{\prime} \left[ \sigma_{\beta}^{-2}(\mathbf{I}_p + \sigma_{\beta}^2\mathbf{X}_k^{\prime}\mathbf{X}_k)\right]\boldsymbol{\beta}_k\right) \right\} d\boldsymbol{\beta}_k \\
	&= (2\pi\sigma_k^2)^{-\frac{n_k}{2}}(2\pi\sigma_k^2\sigma_{\beta}^2)^{-\frac{p}{2}}\int \exp \left\{- \frac{1}{2\sigma_k^2}\left(\mathbf{y}_k^{\prime}\mathbf{y}_k  - 2\boldsymbol{\beta}_k^{\prime}\mathbf{U}_k + \boldsymbol{\beta}_k^{\prime} \mathbf{C}_{k}\boldsymbol{\beta}_k\right) \right\} d\boldsymbol{\beta}_k
\end{align*}
where \(\mathbf{U}_k = \mathbf{X}_{k}^{\prime}\mathbf{y}_k\) and \(\mathbf{C}_k = \mathbf{X}_k^{\prime}\mathbf{X}_k  + \sigma_{\beta}^{-2}\mathbf{I}_p =\sigma_{\beta}^{-2}(\mathbf{I}_p + \sigma_{\beta}^2\mathbf{X}_k^{\prime}\mathbf{X}_{k})\) and \(\mathbf{C}_k^{-1} = \sigma_{\beta}^2(\mathbf{I}_p + \sigma_{\beta}^2\mathbf{X}_k^{\prime}\mathbf{X}_{k})^{-1}\).
\begin{align*}
	&\int f(\mathbf{y}_k \vert \mathbf{X}_{k}\boldsymbol{\beta}_k, \sigma_k^2I_{n_k}) f(\boldsymbol{\beta}_k\vert \mathbf{0}_p, \sigma_k^2\sigma_{\beta}^2\mathbf{I}_p) d\boldsymbol{\beta}_k = (2\pi\sigma_k^2)^{-\frac{n_k}{2}}(2\pi\sigma_k^2\sigma_{\beta}^2)^{-\frac{p}{2}} \\& \times\int \exp \left\{-\frac{1}{2\sigma_k^2} \left[ \mathbf{y}_k^{\prime}\mathbf{y}_k - \mathbf{U}_k^{\prime}\mathbf{C}_k^{-1}\mathbf{U}_k + (\boldsymbol{\beta}_k - \mathbf{C}_k^{-1}\mathbf{U}_k)^{\prime}\mathbf{C}_k(\boldsymbol{\beta}_k - \mathbf{C}_k^{-1}\mathbf{U}_k) \right]\right\} d\boldsymbol{\beta}_k \\
	&= (2\pi\sigma_k^2)^{-\frac{n_k}{2}}(2\pi\sigma_k^2\sigma_{\beta}^2)^{-\frac{p}{2}}\exp\left\{-\frac{1}{2\sigma_k^2}\left(\mathbf{y}_k^{\prime}\mathbf{y}_k -\mathbf{U}_k^{\prime}\mathbf{C}_k^{-1}\mathbf{U}_k\right)\right\} \\
	&\times\int \exp \left\{-\frac{1}{2\sigma_k^2} \left[(\boldsymbol{\beta}_k - \mathbf{C}_k^{-1}\mathbf{U}_k)^{\prime}\mathbf{C}_k(\boldsymbol{\beta}_k - \mathbf{C}_k^{-1}\mathbf{U}_k) \right]\right\} d\boldsymbol{\beta}_k \\
	&= (2\pi\sigma_k^2)^{-\frac{n_k}{2}}(2\pi\sigma_k^2\sigma_{\beta}^2)^{-\frac{p}{2}}\exp\left\{-\frac{1}{2\sigma_k^2}\left(\mathbf{y}_k^{\prime}\mathbf{y}_k -\mathbf{U}_k^{\prime}\mathbf{C}_k^{-1}\mathbf{U}_k\right)\right\} (2\pi)^{\frac{p}{2}}\sigma_k^p\vert\mathbf{C}_k\vert^{-1/2} \\
	&= (2\pi\sigma_k^2)^{-\frac{n_k}{2}}(\sigma_{\beta}^2)^{-\frac{p}{2}}\vert \mathbf{C}_k\vert^{-\frac{1}{2}}\exp\left\{-\frac{1}{2\sigma_k^2}\left(\mathbf{y}_k^{\prime}\mathbf{y}_k -\mathbf{U}_k^{\prime}\mathbf{C}_k^{-1}\mathbf{U}_k\right)\right\}
\end{align*}
Now plugging back in for \(\mathbf{C}_k\) and \(\mathbf{U}_k\) we have 
\begin{align*}
	&\int f(\mathbf{y}_k \vert \mathbf{X}_{k}\boldsymbol{\beta}_k, \sigma_k^2I_{n_k}) f(\boldsymbol{\beta}_k\vert \mathbf{0}_p, \sigma_k^2\sigma_{\beta}^2\mathbf{I}_p) d\boldsymbol{\beta}_k \\
	&= (2\pi\sigma_k^2)^{-\frac{n_k}{2}}(\sigma_{\beta}^2)^{-\frac{p}{2}}\vert \sigma_{\beta}^{-2}(\mathbf{I}_p + \sigma_{\beta}^2\mathbf{X}_k^{\prime}\mathbf{X}_k)\vert^{-\frac{1}{2}} \\
	&\times \exp \left\{-\frac{1}{2\sigma_k^2}\left[\mathbf{y}_k^{\prime}\mathbf{y}_k - \sigma_{\beta}^2\mathbf{y}_k^{\prime}\mathbf{X}_{k}(\mathbf{I}_p + \sigma_{\beta}^2\mathbf{X}_k^{\prime}\mathbf{X}_k)^{-1}\mathbf{X}_k^{\prime}\mathbf{y}_k \right] \right\} \\
	&= (2\pi\sigma_k^2)^{-\frac{n_k}{2}}\vert\mathbf{I}_{n_k} + \sigma_{\beta}^2 \mathbf{X}_k \mathbf{X}_k^{\prime} \vert^{-\frac{1}{2}} \exp \left\{- \frac{1}{2\sigma_k^2}\left[\mathbf{y}_k^{\prime}\left(\mathbf{I}_{n_k} - \sigma_{\beta}^2\mathbf{X}_k(\mathbf{I}_p + \sigma_{\beta}^2\mathbf{X}_k^{\prime}\mathbf{X}_k)^{-1}\mathbf{X}_k^{\prime} \right)\mathbf{y}_k \right] \right\}
\end{align*}
Now, applying the Woodbury identity gives 
\begin{align*}
	&\int f(\mathbf{y}_k \vert \mathbf{X}_{k}\boldsymbol{\beta}_k, \sigma_k^2I_{n_k}) f(\boldsymbol{\beta}_k\vert \mathbf{0}_p, \sigma_k^2\sigma_{\beta}^2\mathbf{I}_p) d\boldsymbol{\beta}_k \\
	&= (2\pi\sigma_k^2)^{- \frac{n_k}{2}} \vert\mathbf{I}_{n_k} + \sigma_{\beta}^2\mathbf{X}_k\mathbf{X}_k^{\prime} \vert^{-\frac{1}{2}} \exp \left\{-\frac{1}{2\sigma_k^2}\left[\mathbf{y}_k^{\prime}(\sigma_{\beta}^2\mathbf{X}_k\mathbf{X}_k^{\prime} + \mathbf{I}_{n_k})^{-1}\mathbf{y}_k\right] \right\} \\
	&= MVN(\mathbf{0}_{n_k}, \sigma_k^2(\mathbf{I}_{n_k} + \sigma_{\beta}^2\mathbf{X}_k\mathbf{X}_k^{\prime}))
\end{align*}
Thus,
\[
\mathbf{y}_k \mid \sigma_k^2,\mathbf{X}_k,\{Z_{ik}\},K
\sim \text{MVN}\!\left(\mathbf{0}_{n_k},\, \sigma_k^2\left(\mathbf{I}_{n_k}+\sigma_\beta^2\mathbf{X}_k\mathbf{X}_k^\top\right)\right),
\]
and the integral above equals the corresponding multivariate normal density evaluated at \(\mathbf{y}_k\). Using this marginal likelihood and a proportionality argument, we have that 
\begin{align*}
	f&(\{\sigma_k^2\}\vert \{Z_{ik}\}, K,\mathbf{y}) \propto f(\mathbf{y}\vert \{\sigma_k^2\}, \{Z_{ik}\}, K)f( \{\sigma_k^2\}\vert \omega, \kappa) \\
	&\propto \prod_{k = 1,\dots, K, n_k = 0} f(\sigma_k^2\vert\omega, \kappa) \prod_{k = 1,\dots, K,n_k > 0} f(\mathbf{y}_k \vert \sigma_k^2, \{Z_{ik}\}, K)f(\sigma_k^2 \vert \omega, \kappa).
\end{align*}
For \(n_k > 0\), we have 
\begin{align*}
	&\prod_{k = 1,\dots, K,n_k > 0} f(\mathbf{y}_k \vert \sigma_k^2, \{Z_{ik}\}, K)f(\sigma_k^2 \vert \omega, \kappa) \\
	&\propto \prod_{k = 1,\dots, K,n_k > 0}\frac{(\sigma_k^2)^{-\frac{n_k}{2}}}{(2\pi)^{\frac{n_k}{2}}} \exp \left\{ - \frac{\mathbf{y}_k^{\prime}(\sigma_{\beta}^2\mathbf{X}_k\mathbf{X}_k^{\prime} + \mathbf{I}_{n_k})^{-1}\mathbf{y}_k}{2\sigma_k^2}\right\} \frac{\kappa^{\omega}}{\Gamma(\omega)}(\sigma_k^2)^{-(\omega + 1)}\exp\left\{ - \frac{\kappa}{\sigma_k^2}\right\} \\
	&\propto \prod_{k = 1,\dots, K,n_k > 0}(\sigma_k^2)^{-(\omega + 1 + \frac{n_k}{2})} \exp \left\{ -\frac{1}{\sigma_k^2}\left(\frac{\mathbf{y}_k^{\prime}(\sigma_{\beta}^2\mathbf{X}_k\mathbf{X}_k^{\prime} + \mathbf{I}_{n_k})^{-1}\mathbf{y}_k}{2} + \kappa \right)\right\} \\
	&\propto IG\left(\omega + \frac{n_k}{2},  \frac{\mathbf{y}_k^{\prime}(\sigma_{\beta}^2\mathbf{X}_k\mathbf{X}_k^{\prime} + \mathbf{I}_{n_k})^{-1}\mathbf{y}_k}{2} + \kappa\right) \\
	&\propto IG(\omega_{P,k}, \kappa_{P,k})
\end{align*}
So, we have the the following known distribution: 
\begin{align}
	\sigma_k^2 \mid \boldsymbol{\pi}, \{Z_{ik}\},K,\mathbf{y}
	\sim
	\begin{cases}
		\text{IG}\!\left(\omega_{P,k},\,\kappa_{P,k}\right),
		& \text{if } n_k>0,\\[6pt]
		\text{IG}\!\left(\omega,\,\kappa\right),
		& \text{if } n_k=0~,
	\end{cases}
	\label{appen:eq:sig_post}
\end{align}
where \(\omega_{P,k} = \omega + \frac{n_k}{2}\) and \(\kappa_{P,k} = {\mathbf{y}_k^{\prime}(\sigma_{\beta}^2\mathbf{X}_k\mathbf{X}_k^{\prime} + \mathbf{I}_{n_k})^{-1}\mathbf{y}_k}/{2} + \kappa\).

Now we have the joint distribution \(f(\boldsymbol{\pi}, \{Z_{ik}\}, K, \mathbf{y}) \) given by 
\begin{align*}
	f(\boldsymbol{\pi}&, \{Z_{ik}\}, K, \mathbf{y}) = \frac{f(\{\sigma_k^2\}, \boldsymbol{\pi}, \{Z_{ik}\}, K, \mathbf{y})}{f(\{\sigma_k^2\}\vert \boldsymbol{\pi}, \{Z_{ik}\}, K, \mathbf{y})} \\
	&= \frac{f(\mathbf{y}\vert \{\sigma_k^2\}, \boldsymbol{\pi}, \{Z_{ik}\}, K)\prod_{i = 1}^Nf(Z_i \vert \boldsymbol{\pi}, K)f(\boldsymbol{\pi}\vert \{\alpha_k\}, K)f(K) \prod_{k = 1}^Kf(\sigma_k^2\vert\omega, \kappa)}{f(\{\sigma_k^2\}\vert\boldsymbol{\pi}, \{Z_{ik}\}, K, \mathbf{y})}.
\end{align*}
Substituting the known expressions into the above equation gives the following: 
\begin{align*}
	&f(\boldsymbol{\pi}, \{Z_{ik}\}, K, \mathbf{y}) = f(\boldsymbol{\pi}\vert \{\alpha_k\}, K)f(K) \prod_{i = 1}^Nf(Z_i\vert\boldsymbol{\pi}, K) \\
	&\times\prod_{k = 1,\dots, K, n_k > 0}(2\pi)^{-\frac{n_k}{2}} \vert \mathbf{I}_{n_k} + \sigma_{\beta}^2\mathbf{X}_k\mathbf{X}_k^{\prime} \vert^{-\frac{1}{2}}\frac{\kappa^{\omega}}{\Gamma(\omega)} \frac{\Gamma(\omega + \frac{n_k}{2})}{\left(\kappa + \frac{1}{2}\mathbf{y}_k^{\prime}(\sigma_{\beta}^2\mathbf{X}_k\mathbf{X}_k^{\prime} + \mathbf{I}_{n_k})^{-1}\mathbf{y}_k \right)^{\omega + \frac{n_k}{2}}}.
\end{align*}
Using a proportionality argument, we obtain
\begin{align*}
	f(\boldsymbol{\pi}\vert \{Z_{ik}\}, K, \mathbf{y}) &\propto f(\boldsymbol{\pi}, \{Z_{ik}\}, K, \mathbf{y}) \\
	&\propto f(\boldsymbol{\pi}\vert \{\alpha_k\}, K)\prod_{i = 1}^Nf(\mathbf{z}_i\vert\boldsymbol{\pi}, K) \\
	&\propto \frac{\Gamma(\sum_{k =1}^K\alpha_k)}{\prod_{k =1}^K\Gamma(\alpha_k)}\prod_{k = 1}^K \pi_k^{\alpha_k -1} \prod_{i = 1}^N \prod_{k = 1}^K\pi_k^{Z_{ik}} \\
	&\propto \prod_{k = 1}^K \pi_k^{\alpha_k - 1}\prod_{k = 1}^K \pi_k^{\sum_{i = 1}^N Z_{ik}} \\
	&\propto \prod_{k = 1}^K \pi_k^{\alpha_k - 1 + \sum_{i = 1}^N Z_{ik}} \\
	&\propto \frac{\Gamma(N + \sum_{k = 1}^K\alpha_k)}{\prod_{k = 1}^K \Gamma(\alpha_k + \sum_{i = 1}^NZ_{ik})}\prod_{k = 1}^K \pi_k^{\sum_{i = 1}^N Z_{ik} + \alpha_k - 1} \\
	&\propto \text{Dirichlet}(\alpha_1+n_1,\dots,\alpha_K+n_K),
\end{align*}
where \(n_k=\sum_{i=1}^N Z_{ik}\) denotes the number of observations allocated to component \(k\).

Now we have the joint distribution \(f(\mathbf{z}_1, \dots, \mathbf{z}_N, K, \mathbf{y})\) is given by 
\begin{align*}
	f(\mathbf{z}_1, \dots, \mathbf{z}_N, K, \mathbf{y}) = \frac{f(\boldsymbol{\pi}, \{Z_{ik}\}, K, \mathbf{y})}{f(\boldsymbol{\pi}\vert \{Z_{ik}\}, K, \mathbf{y})}.
\end{align*}
Substituting in the known expressions into the above equation gives the following: 
\begin{align}
	f(\{Z_{ik}\}, &K, \mathbf{y}) = f(K) \frac{\Gamma(\sum_{k = 1}^K\alpha_k)}{\Gamma(N + \sum_{k = 1}^K\alpha_k)} \frac{\prod_{k = 1}^K\Gamma(\alpha_k + \sum_{i = 1}^N Z_{ik})}{\prod_{k = 1}^K\Gamma(\alpha_k)} \notag \\
	&\times\prod_{k = 1,\dots, K, n_k > 0}(2\pi)^{-\frac{n_k}{2}} \vert \mathbf{I}_{n_k} + \sigma_{\beta}^2\mathbf{X}_k\mathbf{X}_k^{\prime} \vert^{-\frac{1}{2}}\frac{\kappa^{\omega}}{\Gamma(\omega)} \frac{\Gamma(\omega + \frac{n_k}{2})}{\left(\kappa + \frac{1}{2}\mathbf{y}_k^{\prime}(\sigma_{\beta}^2\mathbf{X}_k\mathbf{X}_k^{\prime} + \mathbf{I}_{n_k})^{-1}\mathbf{y}_k \right)^{\omega + \frac{n_k}{2}}}.
	\label{eq:appen:zXy}
\end{align}
Let \(\Delta_K\) be the set of all possible values of \(\{Z_{ik}\}\) for a given \(K\). Then 
\begin{align*}
	f(K,\mathbf{y}) &= \sum_{\{Z_{ik}\}\in \Delta_K} f(K) \frac{\Gamma(\sum_{k = 1}^K\alpha_k)}{\Gamma(N + \sum_{k = 1}^K\alpha_k)} \frac{\prod_{k = 1}^K\Gamma(\alpha_k + \sum_{i = 1}^N Z_{ik})}{\prod_{k = 1}^K\Gamma(\alpha_k)} \\
	&\times\prod_{k = 1,\dots, K, n_k > 0}(2\pi)^{-\frac{n_k}{2}} \vert \mathbf{I}_{n_k} + \sigma_{\beta}^2\mathbf{X}_k\mathbf{X}_k^{\prime} \vert^{-\frac{1}{2}}\frac{\kappa^{\omega}}{\Gamma(\omega)} \frac{\Gamma(\omega + \frac{n_k}{2})}{\left(\kappa + \frac{1}{2}\mathbf{y}_k^{\prime}(\sigma_{\beta}^2\mathbf{X}_k\mathbf{X}_k^{\prime} + \mathbf{I}_{n_k})^{-1}\mathbf{y}_k \right)^{\omega + \frac{n_k}{2}}} \\
	&= f(K) \frac{\Gamma(\sum_{k = 1}^K\alpha_k)}{\Gamma(N + \sum_{k = 1}^K\alpha_k)}\sum_{\{Z_{ik}\}\in \Delta_K} \frac{\prod_{k = 1}^K\Gamma(\alpha_k + \sum_{i = 1}^N Z_{ik})}{\prod_{k = 1}^K\Gamma(\alpha_k)} \\
	&\times\prod_{k = 1,\dots, K, n_k > 0}(2\pi)^{-\frac{n_k}{2}} \vert \mathbf{I}_{n_k} + \sigma_{\beta}^2\mathbf{X}_k\mathbf{X}_k^{\prime} \vert^{-\frac{1}{2}}\frac{\kappa^{\omega}}{\Gamma(\omega)} \frac{\Gamma(\omega + \frac{n_k}{2})}{\left(\kappa + \frac{1}{2}\mathbf{y}_k^{\prime}(\sigma_{\beta}^2\mathbf{X}_k\mathbf{X}_k^{\prime} + \mathbf{I}_{n_k})^{-1}\mathbf{y}_k \right)^{\omega + \frac{n_k}{2}}}, 
\end{align*}
and
\begin{align*}
	f(\mathbf{y}&) = \sum_{K = 1}^M\sum_{\{Z_{ik}\}\in \Delta_K} f(K) \frac{\Gamma(\sum_{k = 1}^K\alpha_k)}{\Gamma(N + \sum_{k = 1}^K\alpha_k)} \frac{\prod_{k = 1}^K\Gamma(\alpha_k + \sum_{i = 1}^N Z_{ik})}{\prod_{k = 1}^K\Gamma(\alpha_k)} \\
	&\times\prod_{k = 1,\dots, K, n_k > 0}(2\pi)^{-\frac{n_k}{2}} \vert \mathbf{I}_{n_k} + \sigma_{\beta}^2\mathbf{X}_k\mathbf{X}_k^{\prime} \vert^{-\frac{1}{2}}\frac{\kappa^{\omega}}{\Gamma(\omega)} \frac{\Gamma(\omega + \frac{n_k}{2})}{\left(\kappa + \frac{1}{2}\mathbf{y}_k^{\prime}(\sigma_{\beta}^2\mathbf{X}_k\mathbf{X}_k^{\prime} + \mathbf{I}_{n_k})^{-1}\mathbf{y}_k \right)^{\omega + \frac{n_k}{2}}} \\
	&= \sum_{K =1}^M f(K) \frac{\Gamma(\sum_{k = 1}^K\alpha_k)}{\Gamma(N + \sum_{k = 1}^K\alpha_k)}\sum_{\{Z_{ik}\}\in \Delta_K} \frac{\prod_{k = 1}^K\Gamma(\alpha_k + \sum_{i = 1}^N Z_{ik})}{\prod_{k = 1}^K\Gamma(\alpha_k)} \\
	&\times\prod_{k = 1,\dots, K, n_k > 0}(2\pi)^{-\frac{n_k}{2}} \vert \mathbf{I}_{n_k} + \sigma_{\beta}^2\mathbf{X}_k\mathbf{X}_k^{\prime} \vert^{-\frac{1}{2}}\frac{\kappa^{\omega}}{\Gamma(\omega)} \frac{\Gamma(\omega + \frac{n_k}{2})}{\left(\kappa + \frac{1}{2}\mathbf{y}_k^{\prime}(\sigma_{\beta}^2\mathbf{X}_k\mathbf{X}_k^{\prime} + \mathbf{I}_{n_k})^{-1}\mathbf{y}_k \right)^{\omega + \frac{n_k}{2}}}.
\end{align*}
Using the expressions above we have 
\begin{align}
	f(\{Z_{ik}\}\vert K, \mathbf{y}) &= \frac{f(\{Z_{ik}\}, K, \mathbf{y})}{f(K, \mathbf{y})} \notag\\
	f(K\vert\mathbf{y}) &= \frac{f(K, \mathbf{y})}{f(\mathbf{y})}.
	\label{appen:eq:z_post}
\end{align}
These expressions define discrete posterior distributions that can be evaluated directly and can be sampled from directly in low-dimensional unconstrained settings.

\section*{Appendix B: Gibbs Sampler with Subset of Data}
\begin{algorithm}[H]
	\footnotesize
	\caption{Construct $\mathcal{C}$ with Gibbs Sampling using subset of Data}
	\begin{algorithmic}[1]
		\State Select subset fraction \(p_{\text{sub}}\in\{0.1,0.25,0.5\}\) and set \(n_{\text{sub}}=p_{\text{sub}}N\), specify number of iterations \(W\), and hyperparameters \(\{\alpha_k\}, \sigma_{\beta}^2, \omega, \kappa\). Select $K = K^{[1]}=\ldots=K^{[W]}$ large.
		\State Sample indices \(S\subset\{1,\ldots,N\}\) uniformly without replacement and order \(S\), \(|S|=n_{\text{sub}}\) and define \((\mathbf y_{\text{sub}},\mathbf X_{\text{sub}})=\{(Y_i,\mathbf x_i): i\in S\}\).
		\Statex \textbf{Notation:} At iteration \(t\), let \(\mathbf{z}_{\text{sub}}^{[t]} = (z_{\text{sub},1}^{[t]}, \dots, z_{\text{sub}, n_{\text{sub}}}^{[t]})^{\prime}\) denote the vector of labels for the subset, where \(z_{\text{sub}, i}^{[t]} \in \{1, \dots, K\}\). Let \(S_k^{[t]} = \{i \in \{1,\dots,n_{\text{sub}}\} : z_{\text{sub}, i}^{[t]} = k\}\), \(\mathbf{X}_{\text{sub},k} = \{\mathbf{x}_{i}^{\prime}: i \in S_k^{[t]}\} \), and \(\mathbf{y}_{\text{sub},k} = \{Y_i: i \in S_k^{[t]}\}\).
		\State Initialize: \(z_{\text{sub},i}^{[0]}\sim \text{Categorical}(1/K,\ldots,1/K)\),
		\(\boldsymbol{\beta}_k^{[0]}=\mathbf{0}_p\), \(\sigma_k^{2[0]}=1\), \(\pi_k^{[0]}=1/K\).
		\For{$t=1,\ldots,W$}
		\State Update mixture weights:
		\(\boldsymbol{\pi}^{[t]}\mid \mathbf{z}_{\text{sub}}^{[t-1]}
		\sim \text{Dirichlet}\!\left(\alpha_1+ n_1^{[t-1]},\dots,\alpha_K+n_K^{[t-1]}\right)\),  \(n_k^{[t-1]} = \sum_{i = 1}^{n_{\text{sub}}}\mathbb{I}(z_{\text{sub},i}^{[t-1]} = k).\)
		
		\State Update component means:
		\[
		\boldsymbol{\beta}_k^{[t]} \vert \sigma_k^{2[t-1]}, \mathbf{z}_{\text{sub}}^{[t-1]}, \mathbf{y}_{\text{sub}}
		\sim MVN\!\left(
		\left(\mathbf{X}_{\text{sub},k}^{\prime}\mathbf{X}_{\text{sub},k} + \frac{1}{\sigma_{\beta}^{2}}\mathbf{I}_p\right)^{-1}\mathbf{X}^{\prime}_{\text{sub},k}\mathbf{y}_{\text{sub},k},
		\ \sigma_k^{2[t-1]}\left(\mathbf{X}_{\text{sub},k}^{\prime}\mathbf{X}_{\text{sub},k} + \frac{1}{\sigma_{\beta}^{2}}\mathbf{I}_p\right)^{-1}
		\right).
		\]
		\State Update the component variances:
		\[
		\sigma_k^{2[t]} \vert \boldsymbol{\beta}_k^{[t]}, \sigma_{\beta}^{2}, \mathbf{y}_{\text{sub}}, \mathbf{z}_{\text{sub}}^{[t-1]}
		\sim IG\!\left(
		\omega + \frac{n_k^{[t-1]} + p}{2},\
		\kappa + \frac{1}{2}(\mathbf{y}_{\text{sub},k}- \mathbf{X}_{\text{sub},k}\boldsymbol{\beta}_k^{[t]})^{\prime}(\mathbf{y}_{\text{sub},k}- \mathbf{X}_{\text{sub},k}\boldsymbol{\beta}_k^{[t]})
		+ \frac{1}{2\sigma_{\beta}^2}\boldsymbol{\beta}^{[t]\prime}_k\boldsymbol{\beta}_k^{[t]}
		\right),
		\]
		\State Update subset labels:
		\(z_{\text{sub},i}^{[t]} \sim \text{Categorical}(p_{i1}, \dots, p_{iK}),\quad i=1,\dots,n_{\text{sub}},\) with
		\[
		p_{ik}=\frac{\pi_{k}^{[t]}\phi(Y_{\text{sub},i}\vert \mathbf{x}_{\text{sub},i}^{\prime}\boldsymbol{\beta}_k^{[t]}, \sigma_k^{2[t]})}{
			\sum_{j =1}^K \pi_j^{[t]}\phi(Y_{\text{sub},i}\vert \mathbf{x}_{\text{sub},i}^{\prime}\boldsymbol{\beta}_j^{[t]}, \sigma_j^{2[t]})},
		\]
		and for numerical stability let \(l_{ik}=\log\pi_k^{[t]}+\log\phi(Y_{\text{sub},i}\vert\mathbf{x}_{\text{sub},i}^\prime\boldsymbol{\beta}_k^{[t]},\sigma_k^{2[t]})\),
		\(m_i=\max_{1\le h\le K}l_{ih}\), and
		\[
		p_{ik} = \frac{\exp(l_{ik}-m_i)}{\sum_{h=1}^K \exp(l_{ih}-m_i)}.
		\]
		\State Store \((\boldsymbol{\pi}^{[t]}, \{\boldsymbol{\beta}^{[t]}_k\}_{k = 1}^K, \{\sigma_k^{2[t]}\}_{k = 1}^K)\)
		\EndFor
		\Statex \textbf{Construct candidate partitions from Gibbs output (full-data prediction):}
		\State Let \(W\) be the number of stored draws and denote \(\{(\boldsymbol{\pi}^{[w]},\boldsymbol{\beta}^{[w]},\boldsymbol{\sigma}^{2[w]})\}_{w=1}^W\).
		\For{\(w=1,\ldots,W\)}
		\State Compute \(l_{ik}^{[w]}=\log \pi_k^{[w]}+\log\phi(Y_i\vert \mathbf{x}_i^\prime\boldsymbol{\beta}_k^{[w]},\sigma_k^{2[w]})\) for \(i=1,\ldots,N\), \(k=1,\ldots,K\).
		\State Let \(m_i^{[w]}=\max_{1\le h\le K} l_{ih}^{[w]}\) and
		\(
		p_{ik}^{[w]}={\exp(l_{ik}^{[w]}-m_i^{[w]})}/{\sum_{h=1}^K \exp(l_{ih}^{[w]}-m_i^{[w]})}.
		\)
		\State Sample full data labels \(\widetilde{z}_i^{[w]}\sim \text{Categorical}(p_{i1}^{[w]},\dots,p_{iK}^{[w]})\) for \(i=1,\ldots,N\).
		\State Form \(\widetilde{\mathbf{z}}^{[w]}=(\widetilde{z}_1^{[w]},\ldots,\widetilde {z}_N^{[w]})^{\prime}\), relabel \(\widetilde{\mathbf{z}}^{[w]}\)
		\State If unique, convert into an \(N \times K\) allocation matrix \(\widetilde{\mathbf{Z}}^{[w]}\) with \(\widetilde{Z}_{ik}^{[w]}= \mathbb{I} \{\widetilde{z}_i^{[w]}=k\}\) and store \(\widetilde{\mathbf{Z}}^{[w]}\) in \(\mathcal{C}\) 
		\EndFor
		\State \Return \(\mathcal{C}\)
	\end{algorithmic}
\end{algorithm}

\section*{Appendix C: Machine Learning Algorithms}
In this appendix, we describe the machine learning clustering algorithms used to construct the candidate set \(\mathcal{C}(\mathbf{y})\) in the constrained model. Specifically, we consider \(k\)-means, \(k\)-medoids, spectral clustering, hierarchical clustering, and the EM algorithm. For each method, we fit the algorithm over a grid of values of \(K\) and record the resulting partitions. Collecting these partitions across methods yields the candidate set \(\mathcal{C}(\mathbf{y})\) used in the main text. We briefly summarize each algorithm below before giving the detailed pseudocode.

The \(k\)-means algorithm clusters the observations into \(k\) groups by repeatedly assigning each observation to the nearest cluster center and then updating centers to minimize within cluster variation, iterating until convergence or a maximum number of iterations is reached \citep{hartigan1979algorithm, ahmed2020k}. The \(k\)-medoids algorithm is similar, but it represents each cluster with an actual data point (a medoid) and iteratively swaps medoids to reduce the total within cluster dissimilarity  \citep{park2009simple}. Spectral clustering starts by measuring how similar each pair of observations is using a similarity graph, then uses an eigenvector based transformation to map the data into a space where the clusters are more easily separated, and finally applies \(k\)-means to that transformed representation \citep{ng2001spectral}. Hierarchical clustering builds clusters step by step repeatedly merging the two most similar groups of observations to form a tree (dendrogram), and a partition with \(k\) clusters is obtained by cutting the tree at the level that yields \(k\) groups \citep{ward1963hierarchical, murtagh2014ward}. The EM algorithm fits a mixture model by iterating between computing component membership probabilities for each component (E-step) and updating the component parameters to maximize the expected log-likelihood of the complete data (M-step), repeating until the fit stabilizes \citep{dempster1977maximum}. 

\begin{algorithm}[H]
	\caption{\(k\)-means candidate partitions for \(\mathcal{C}\)}
	\begin{algorithmic}[1]
		\Require Responses \(\{Y_i\}_{i=1}^N\), covariates \(\{\mathbf x_i\}_{i=1}^N\) (optional), maximum clusters \(K\), random starts \(S\), maximum iterations \(T_{\max}\)
		\State Form feature vectors \(\mathbf v_i\) for \(i=1,\ldots,N\):
		\State \hspace{0.5cm} Intercept-only: \(\mathbf v_i=Y_i \in \mathbb{R}\)
		\State \hspace{0.5cm} Regression mixture: \(\mathbf v_i=(Y_i,\mathbf x_i^{\prime})^{\prime}\in\mathbb R^{p+1}\)
		\For{\(w=2,\ldots,K\)}
		\State Initialize best within cluster sum of squares \(\varepsilon^{*}\leftarrow +\infty\)
		\For{\(s=1,\ldots,S\)}
		\State Initialize centers \(\{\mathbf{m}_k^{(0)}\}_{k=1}^w\) by sampling \(w\) feature vectors from \(\{\mathbf v_i\}_{i=1}^N\)
		\For{\(t=1,\ldots,T_{\max}\)}
		\Statex \textbf{Assignment step}
		\For{\(i=1,\ldots,N\)}
		\State \(z_i^{(t)} \leftarrow \arg\min_{k\in\{1,\ldots,w\}} \|\mathbf v_i-\mathbf m_k^{(t-1)}\|_2^2\)
		\EndFor
		\State Define \(\mathcal I_k^{(t)}=\{i: z_i^{(t)}=k\}\) and \(n_k^{(t)}=|\mathcal I_k^{(t)}|\), for \(k=1,\ldots,w\)
		
		\Statex \textbf{Update step}
		\For{\(k=1,\ldots,w\)}
		\If{\(n_k^{(t)}=0\)}
		\State Reinitialize \(\mathbf{m}_k^{(t)}\) (e.g., set to a random \(\mathbf v_i\))
		\Else
		\State \(\mathbf{m}_k^{(t)} \leftarrow \frac{1}{n_k^{(t)}}\sum_{i\in\mathcal{I}_k^{(t)}} \mathbf{v}_i\)
		\EndIf
		\EndFor
		
		\State Compute \(\varepsilon^{(t)} \leftarrow \sum_{k=1}^w\sum_{i\in \mathcal I_k^{(t)}} \|\mathbf v_i-\mathbf m_k^{(t)}\|_2^2\)
		\If{\(\varepsilon^{(t)} < \varepsilon^{*}\)}
		\State Store \(\mathbf{z}^{*}\leftarrow (z_1^{(t)},\ldots,z_N^{(t)})^{\prime}\) and set \(\varepsilon^{*}\leftarrow \varepsilon^{(t)}\)
		\EndIf
		\If{\(t>1\) \textbf{and} \(\mathbf {z}^{(t)}=\mathbf{z}^{(t-1)}\)}
		\State \textbf{break}
		\EndIf
		\EndFor
		\EndFor
		
		\State Set \(\mathbf{z}^{[w]}\leftarrow \mathbf{z}^{*}\) 
		\State Convert into a \(N \times w\) allocation matrix \(\mathbf{Z}^{[w]}\) where \(Z^{[w]}_{ik}=\mathbb I\{z_i^{[w]}=k\}\)
		\EndFor
		
		\State \Return  \(\mathcal{C} = \{(\mathbf{Z}^{[1]},K), \dots, (\mathbf{Z}^{[W]},K)\}\)
	\end{algorithmic}
\end{algorithm}

\begin{algorithm}[H]
	\caption{Spectral Clustering candidate partitions for \(\mathcal{C}\)}
	\small
	\begin{algorithmic}[1]
		\Require Responses \(\{Y_i\}_{i=1}^N\), covariates \(\{\mathbf x_i\}_{i=1}^N\), number of clusters \(K\), number of scale candidates \(G\) (e.g., \(G=10\)), quantiles $q_{\min},q_{\max}$ (e.g., \(0.10,0.90\))
		\State Form feature vectors \(\mathbf v_i\) for \(i=1,\ldots,N\):
		\State \hspace{0.5cm} Regression mixture: \(\mathbf v_i=(Y_i,\mathbf x_i^{\prime})^{\prime}\in\mathbb R^{p+1}\)
		\State Stack rows into \(\mathbf{V}=[\mathbf{v}_1^\prime;\dots;\mathbf{v}_N^\prime]\in\mathbb{R}^{N\times (p+1)}\). Standardize columns of \(\mathbf{V}\) so \(\mathbf{y}\) and \(\mathbf{X}\) are comparable in scale.
		\State Compute squared distance matrix \(\mathbf D\in\mathbb R^{N\times N}\) with entries \(D_{ij}=\|\mathbf v_i-\mathbf v_j\|_2^2\).
		\State Construct a scale grid \(\{h^{(1)},\ldots,h^{(G)}\}\):
		\State \hspace{0.5cm} Let \(\boldsymbol\phi=\{\sqrt{D_{ij}}:1\le i<j\le N\}\). Set \(h_{\min}=\mathrm{Quantile}(\boldsymbol\phi,q_{\min})\), \(h_{\max}=\mathrm{Quantile}(\boldsymbol\phi,q_{\max})\). Let \(\{h^{(g)}\}\) be \(G\) log-spaced values from \(h_{\min}\) to \(h_{\max}\).
		
		\For{\(w=2,\ldots,K\)} 
		\State Initialize best loss \(\varepsilon^{*}\leftarrow +\infty\)
		\For{\(g=1,\ldots,G\)} 
		\State Set \(h\leftarrow h^{(g)}\)
		\State Compute affinity matrix \(\mathbf A\) with \(A_{ij}=\exp\{-D_{ij}/(2h^2)\}\) and set \(A_{ii}=0\).
		\State Compute matrix \(\mathbf B=\mathrm{diag}(\mathbf A\mathbf 1)\) and normalized affinity \(\mathbf L=\mathbf B^{-1/2}\mathbf A\mathbf B^{-1/2}\).
		\State Compute the eigenvectors of \(\mathbf L\) with the \(w\) largest eigenvalues and form \(\mathbf U\in\mathbb R^{N\times w}\).
		\State Row-normalize \(\mathbf U\) to obtain \(\mathbf U^*\), i.e., \(\mathbf U^*_{i\cdot}\leftarrow \mathbf U_{i\cdot}/\|\mathbf U_{i\cdot}\|_2\) for \(i=1,\ldots,N\).
		\State Run \(w\)-means on the rows of \(\mathbf U^*\) to obtain labels \(z_1^{(g)},\ldots,z_N^{(g)}\).
		\State Let \(\mathcal I_k^{(g)}=\{i:z_i^{(g)}=k\}\) and define centroids
		\(\mathbf m_k=\frac{1}{|\mathcal I_k^{(g)}|}\sum_{i\in\mathcal I_k^{(g)}} \mathbf U^*_{i\cdot}\) for \(k=1,\ldots,w\).
		\State Compute loss
		\( \varepsilon^{(g)}\leftarrow \sum_{k=1}^{w}\sum_{i\in\mathcal I_k^{(g)}} \|\mathbf U^*_{i\cdot}-\mathbf m_k\|_2^2.\)
		\If{\(\varepsilon^{(g)}<\varepsilon^{*}\)}
		\State Store \(\mathbf z^{*}\leftarrow (z_1^{(g)},\ldots,z_N^{(g)})^\top\) and set \(\varepsilon^{*}\leftarrow \varepsilon^{(g)}\).
		\EndIf
		\EndFor
		\State Set \(\mathbf{z}^{[w]}\leftarrow \mathbf{z}^{*}\) 
		\State Convert into a \(N \times w\) allocation matrix \(\mathbf{Z}^{[w]}\) where \(Z^{[w]}_{ik}=\mathbb I\{z_i^{[w]}=k\}\)
		\EndFor
		
		\State \Return  \(\mathcal{C} = \{(\mathbf{Z}^{[1]},K), \dots, (\mathbf{Z}^{[W]},K)\}\)
	\end{algorithmic}
\end{algorithm}

\begin{algorithm}[H]
	\caption{Hierarchical clustering allocation for \(\mathcal{C}\)}
	\begin{algorithmic}[1]
		\Require Responses \(\{Y_i\}_{i=1}^N\), covariates \(\{\mathbf x_i\}_{i=1}^N\), number of clusters \(K\)
		\State Form feature vectors \(\mathbf v_i\) for \(i=1,\ldots,N\):
		\State \hspace{0.5cm} Intercept-only: \(\mathbf v_i=Y_i \in \mathbb{R}\)
		\State \hspace{0.5cm} Regression mixture: \(\mathbf v_i=(Y_i,\mathbf x_i^{\prime})^{\prime}\in\mathbb R^{p+1}\)
		\State \textbf{Step A: Compute initial dissimilarities}
		\State Compute Euclidean Distances:
		\(d(\{i\},\{j\}) \gets \|\mathbf{v}_i-\mathbf{v}_j\|_2 \quad (1\le i<j\le N)\)
		and store them in a symmetric ``dissimilarity matrix'' \(\mathbf{D}\).
		\State \textbf{Step B: Build dendrogram} 
		\State Initialize cluster set \(\mathcal{S}\gets\big\{\{1\},\{2\},\dots,\{N\}\big\}\) and weights \(w_{\{i\}}\gets 1\) for all singletons.
		\State Initialize  \(\mathcal{H}\) which is used to store the merged clusters and heights \(h\). 
		\While{\(|\mathcal{S}|>1\)}
		\State Choose the pair of clusters \((A^*,B^*)\) with smallest Ward dissimilarity \(d(A, B)\):
		\[
		(A^*,B^*)= \underset{A\neq B \in \mathcal{S}}{\text{argmin}} \ d(A,B).
		\]
		\State Create merged cluster \(C^* \gets A^*\cup B^*\) and set weight \(w_{C^*}\gets w_{A^*}+w_{B^*}\)
		\State Record merge in \(\mathcal{H}\): store \((A^*,B^*)\) and height \(h\gets d(A^*,B^*)\)
		\State Update dissimilarities from the new cluster \(C^*\) to every other cluster \(Q\in \mathcal{S}\setminus\{A^*,B^*\}\) using the Ward update:
		\footnotesize{
			\[
			d(C^*,Q)
			\;\gets\;
			\left[
			\frac{w_{A^*}+w_Q}{w_{A^*}+w_{B^*}+w_Q}\,\delta(A^*,Q)^2
			+
			\frac{w_{B^*}+w_Q}{w_{A^*}+w_{B^*}+w_Q}\,\delta(B^*,Q)^2
			-
			\frac{w_Q}{w_{A^*}+w_{B^*}+w_Q}\,\delta(A^*,B^*)^2
			\right]^{1/2}.
			\]}
		\State Remove \(A^*\) and \(B^*\) from \(\mathcal{S}\); add \(C^*\) to \(\mathcal{S}\) 
		\EndWhile
		
		\State \textbf{Step C: Cut the dendrogram for each \(w\)}
		\For{each \(w \in \{2,\dots,K\}\)}
		\State Re-initialize clusters \(\mathcal{S}_{w}\gets\big\{\{1\},\dots,\{N\}\big\}\)
		\State \(\mathcal{H}\) from Step 2 lists the merges in chronological order \((A_1,B_1),\dots,(A_{N-1},B_{N-1})\)
		\For{\(t=1\) to \(N-k\)}
		\State Apply merge \((A_t,B_t)\): replace \(A_t\) and \(B_t\) in \(\mathcal{S}_{w}\) by \(A_t\cup B_t\)
		\EndFor
		\State Now \(|\mathcal{S}_{w}|=w\) clusters remain
		
		\State Enumerate \(\mathcal{S}_{w}=\{C_{1},\dots,C_{w}\}\) and assign labels:
		\For{\(j=1\) to \(w\)}
		\For{each \(i\in C_j\)}
		\State \(z^{(w)}_i \gets j\)
		\EndFor
		\EndFor
		
		\State Store \(\mathbf{z}^{(w)}=(z^{(w)}_1,\dots,z^{(w)}_N)\)
		
		\State Convert \(\mathbf{z}^{(w)}\) into a \(N \times w\) allocation matrix \(\mathbf{Z}^{[w]}\) where \(Z^{[w]}_{iw}=\mathbb I\{z_i^{[w]}=w\}\)
		\EndFor
		
		\State \Return  \(\mathcal{C} = \{(\mathbf{Z}^{[1]},K), \dots, (\mathbf{Z}^{[W]},K)\}\)
	\end{algorithmic}
\end{algorithm}

\begin{algorithm}[H]
	\footnotesize
	\caption{$K$-medoids candidate allocation for $\mathcal{C}$}
	\begin{algorithmic}[1]
		\Require Responses \(\{Y_i\}_{i=1}^N\), covariates \(\{\mathbf x_i\}_{i=1}^N\) (optional), maximum clusters \(K\), random starts \(S\), maximum iterations \(T_{\max}\)
		\State Form feature vectors \(\mathbf v_i\) for \(i=1,\ldots,N\):
		\State \hspace{0.5cm} Intercept-only: \(\mathbf v_i=Y_i \in \mathbb{R}\)
		\State \hspace{0.5cm} Regression mixture: \(\mathbf v_i=(Y_i,\mathbf x_i^{\prime})^{\prime}\in\mathbb R^{p+1}\)
		\State Compute dissimilarities \(d_{ij}=\|\mathbf{v}_i-\mathbf{v}_j\|_2\) for \(1\le i<j\le N\).
		
		\For{\(w=2,\ldots,K\)} 
		\State Set best loss \(\varepsilon \leftarrow +\infty\)
		\For{\(s=1,\ldots,S\)} 
		\State Initialize medoid indices \(\mathcal M^{(0)}=\{m_1^{(0)},\ldots,m_w^{(0)}\}\) by sampling \(w\) distinct indices from \(\{1,\ldots,N\}\)
		\For{\(t=1,\ldots,T_{\max}\)}
		\Statex \textbf{Assignment step}
		\For{\(i=1,\ldots,N\)}
		\State \(z_i^{(t)} \leftarrow \arg\min_{k\in\{1,\ldots,w\}} d_{i\,m_k^{(t-1)}}\)
		\EndFor
		\State Define \(\mathcal I_k^{(t)}=\{i: z_i^{(t)}=k\}\) and \(n_k^{(t)}=|\mathcal I_k^{(t)}|\), for \(k=1,\ldots,w\)
		\Statex \textbf{Medoid update step}
		\For{\(k=1,\ldots,w\)}
		\If{\(n_k^{(t)}=0\)}
		\State Reinitialize \(m_k^{(t)}\) (e.g., set to a random index in \(\{1,\ldots,N\}\))
		\Else
		\State \(m_k^{(t)} \leftarrow \arg\min_{j\in \mathcal I_k^{(t)}} \sum_{i\in \mathcal I_k^{(t)}} d_{ij}\)
		\EndIf
		\EndFor
		\State Set \(\mathcal M^{(t)}=\{m_1^{(t)},\ldots,m_w^{(t)}\}\)
		\State Compute current loss \(\varepsilon^{(t)} \leftarrow \sum_{k=1}^w \sum_{i\in\mathcal I_k^{(t)}} d_{i\,m_k^{(t)}}\)
		\If{\(t>1\) and \(\Big( \mathcal M^{(t)}=\mathcal M^{(t-1)} \ \textbf{or}\  \mathbf z^{(t)}=\mathbf z^{(t-1)} \Big)\)}
		\State \textbf{break}
		\EndIf
		\EndFor
		
		\If{\(\varepsilon^{(t)} < \varepsilon^{*}\)}
		\State Store \(\mathbf z^{*}\leftarrow (z_1^{(t)},\ldots,z_N^{(t)})^{\prime}\) and set \(\varepsilon^{*}\leftarrow \varepsilon^{(t)}\)
		\EndIf
		\EndFor
		
		\State Set $\mathbf z^{[w]}\leftarrow \mathbf z^{*}$ 
		
		\State Convert \(\mathbf z^{[w]}\) into an \(N \times K\) allocation matrix \({\mathbf{Z}}^{[w]}\) with \({Z}_{ik}^{[w]}= \mathbb{I} \{{z}_i^{[w]}=k\}\) 
		\EndFor
		\State \Return  \(\mathcal{C} = \{(\mathbf{Z}^{[1]},K), \dots, (\mathbf{Z}^{[W]},K)\}\)
	\end{algorithmic}
\end{algorithm}

\begin{algorithm}[H]
	\caption{EM candidate partitions for \(\mathcal{C}\)}
	\begin{algorithmic}[1]
		\Require Responses \(\{Y_i\}_{i=1}^N\), covariates \(\{\mathbf{x}_i\}_{i=1}^N\) (optional), maximum clusters \(K\), tolerance \(\tau\), maximum iterations \(T_{\max}\)
		\State \hspace{0.5cm} Intercept-only: \(\mathbf{v}_i=Y_i \in \mathbb{R}\)
		\State \hspace{0.5cm} Regression mixture: \(\mathbf{v}_i=(Y_i,\mathbf x_i^{\prime})^{\prime}\in\mathbb R^{p+1}\)
		\State Let \(\phi_d(\mathbf v\mid\boldsymbol\mu,\boldsymbol\Sigma)\) denote the \(d\)-dimensional multivariate normal density.
		
		\For{\(w=2,\ldots,K\)}
		\State Initialize parameters \(\{\pi_k^{(0)},\boldsymbol{\mu}_k^{(0)},\boldsymbol{\Sigma}_k^{(0)}\}_{k=1}^w\).
		\For{\(t=1,\ldots,T_{\max}\)}
		\Statex \textbf{E-step:}
		\[
		r_{ik}^{(t)} \leftarrow
		\frac{\pi_k^{(t-1)}\,\phi_d(\mathbf v_i\mid\boldsymbol\mu_k^{(t-1)},\boldsymbol\Sigma_k^{(t-1)})}
		{\sum_{\ell=1}^w \pi_\ell^{(t-1)}\,\phi_d(\mathbf v_i\mid\boldsymbol\mu_\ell^{(t-1)},\boldsymbol\Sigma_\ell^{(t-1)})},
		\quad i=1,\ldots,N,\ k=1,\ldots,w.
		\]
		
		\Statex \textbf{M-step:}
		\[
		n_k^{(t)}\leftarrow\sum_{i=1}^N r_{ik}^{(t)},\qquad
		\pi_k^{(t)}\leftarrow n_k^{(t)}/N,\qquad
		\boldsymbol\mu_k^{(t)}\leftarrow \frac{1}{n_k^{(t)}}\sum_{i=1}^N r_{ik}^{(t)}\mathbf v_i.
		\]
		\State Update covariances
		\[
		\boldsymbol\Sigma_k^{(t)}\leftarrow \frac{1}{n_k^{(t)}}\sum_{i=1}^N r_{ik}^{(t)}(\mathbf v_i-\boldsymbol\mu_k^{(t)})(\mathbf v_i-\boldsymbol\mu_k^{(t)})^{\prime}.
		\]
		\State Update log-likelihood
		\[
		\ell^{(t)}\leftarrow \sum_{i=1}^N \log\!\left(\sum_{k=1}^w \pi_k^{(t)}\,\phi_d(\mathbf v_i\mid\boldsymbol\mu_k^{(t)},\boldsymbol\Sigma_k^{(t)})\right).
		\]
		\If{$t>1$ \textbf{and} $|\ell^{(t)}-\ell^{(t-1)}|<\tau$}
		\State \textbf{break}
		\EndIf
		\EndFor
		
		\State Convert probabilities to a hard partition:
		\[
		z_i^{(w)} \leftarrow \underset{k\in\{1,\ldots,w\}}{\text{argmax}} r_{ik}^{(t)},\quad i=1,\ldots,N.
		\]
		\State Form an \(N \times w\) allocation matrix \(\mathbf{Z}^{(w)}\) where
		\[
		Z^{(w)}_{ik}=\mathbb I\{z_i^{(w)}=k\},\quad i=1,\ldots,N,\ k=1,\ldots,w.
		\]
		\EndFor
		
		\State \Return  \(\mathcal{C} = \{\mathbf{Z}^{(2)}, \dots, \mathbf{Z}^{(K)}\}\)
	\end{algorithmic}
\end{algorithm}

\section*{Appendix D: Simulated Data}

In this appendix, we illustrate the different levels of cluster separation considered in the simulation studies. Figure~\ref{fig:gmm:data} displays the separation settings for both the intercept-only and regression scenarios.

\begin{figure}[H]
	\centering
	
	\begin{subfigure}{0.7\textwidth}
		\centering
		\includegraphics[width=\linewidth]{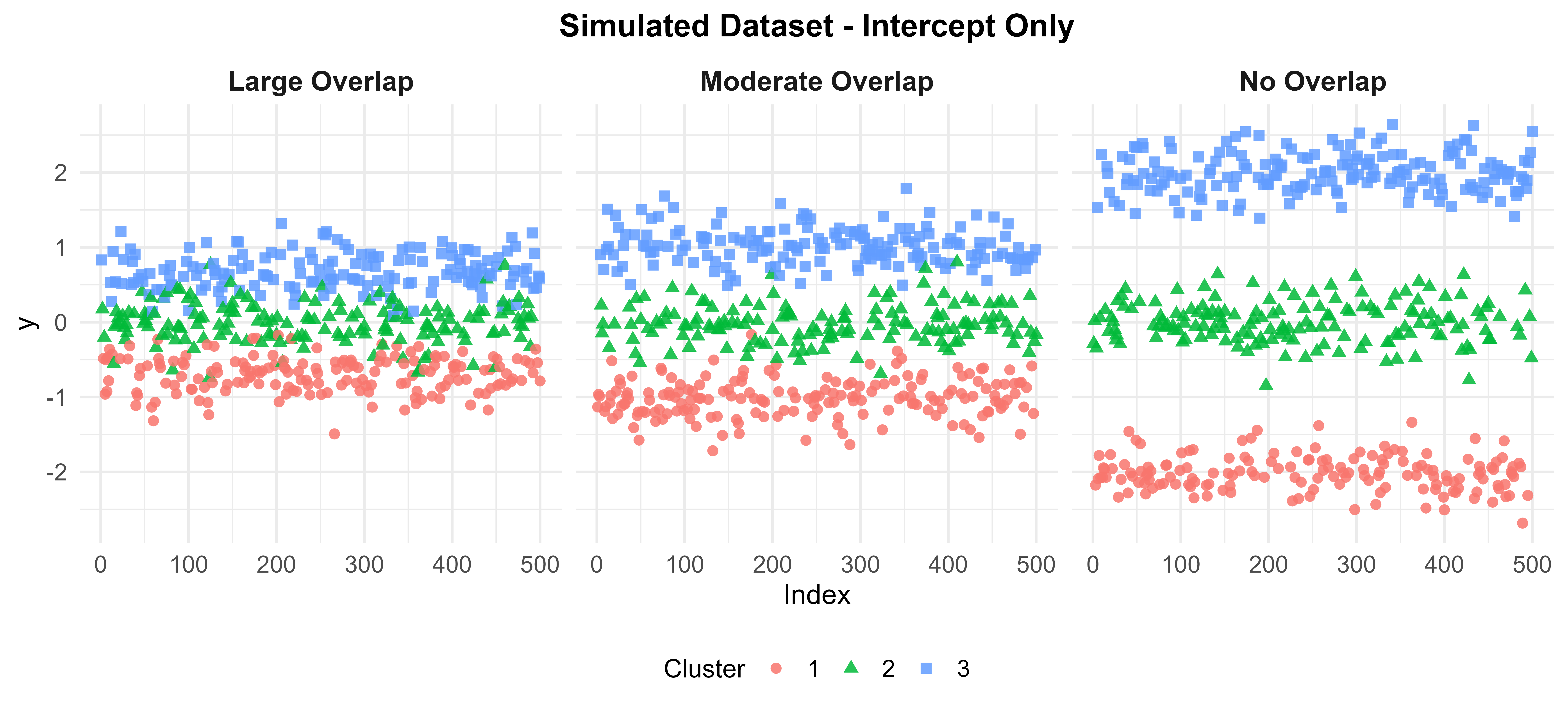}
	\end{subfigure}
	
	\begin{subfigure}{0.7\textwidth}
		\centering
		\includegraphics[width=\linewidth]{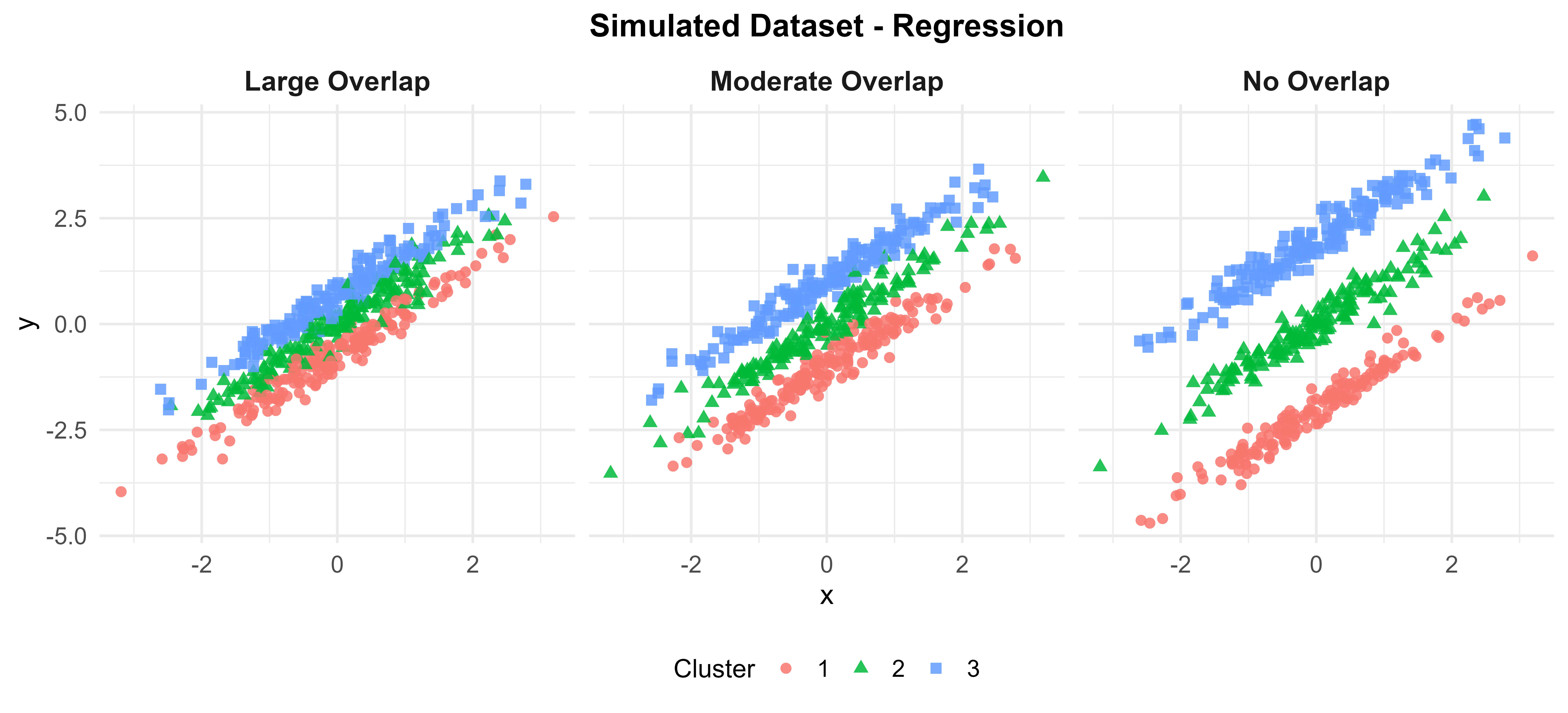}
	\end{subfigure}
	
	\caption{Illustration of the cluster separation settings considered in the simulation studies. The top panel corresponds to the intercept-only scenario, and the bottom panel corresponds to the regression scenario. The data shown were simulated with \(N=500\), although the same separation settings were used for all values of \(N\). From left to right, the panels show large overlap, moderate overlap, and no overlap among \(K=3\) clusters.}
	\label{fig:gmm:data}
\end{figure}

\section*{Appendix E: Detailed Simulation Studies}
\subsection*{Appendix E.1: Intercept Only Simulations}
To evaluate the different approaches in the intercept only scenario, we compare the mean and standard deviation of the adjusted Rand index (ARI), the effective sample size (ESS) for the first and second moments, the Kolmogorov--Smirnov (KS) statistic, the \(L_2\) norm, and central processing unit (CPU) time measured in seconds. For each posterior replicate, we compute the ARI \citep{hubert1985comparing} using the \texttt{ARI()} function in the \texttt{salso} R package \citep{dahl2022search}, and summarize it across replicates by its mean and standard deviation. 

For the KS statistic and \(L_2\) norm, assume we have \(j = 1,\dots,n_{\text{samp}}\) posterior draws, \(k = 1,\dots,K\) mixture components, and a uniform grid \(\{t_g\}_{g=1}^{G}\). For draw \(j\), we have replicates of the component specific parameters \(\pi_{k}^{[j]}, \mu_{k}^{[j]}, \sigma_k^{[j]}\). The mixture cumulative distribution function (CDF) for draw \(j\) is \(F^{[j]}(t_g) = \sum_{k = 1}^K \pi^{[j]}_k\Phi\!\left(\frac{t_g - \mu_{k}^{[j]}}{\sigma_{k}^{[j]}} \right)\), \(g=1,\dots,G\), and the posterior mean CDF is \(\widehat{F}(t_g) = \frac{1}{n_{\text{samp}}}\sum_{j=1}^{n_{\text{samp}}}F^{[j]}(t_g)\). We compute the empirical CDF of the data as \( F_{\text{emp}}(t_g) = \frac{1}{N}\sum_{i=1}^N \mathbb{I}(Y_i \le t_g)\), so the KS distance between the fitted and empirical CDFs is \(KS = \max_g\left|\widehat{F}(t_g) - F_{\text{emp}}(t_g)\right|\). For the \(L_2\) norm, the mixture PDF is \( f^{[j]}(t_g) = \sum_{k=1}^K \pi_k^{[j]}\,\phi\!\left(t_g;\mu_k^{[j]},\sigma_k^{2[j]}\right)\), \(g=1,\dots,G\), where \(\phi(\cdot;\mu,\sigma^2)\) is the Normal density. The posterior mean PDF on the grid is \(\widehat f(t_g)=\frac{1}{n_{\text{samp}}}\sum_{j=1}^{n_{\text{samp}}} f^{[j]}(t_g)\), and the true mixture density is \(f_{\text{true}}(t_g)=\sum_{k=1}^K \pi_{k,\text{true}}\,\phi\!\left(t_g;\mu_{k,\text{true}},\sigma_{k,\text{true}}^{2}\right)\), where the true parameter values are defined in Section \ref{gmm:ss:sim_setup}. The continuous \(L_2\) distance between the fitted and true densities is evaluated on the same uniform grid, with spacing \(\Delta t = t_g - t_{g-1}\), as \( \| \widehat f - f_{\text{true}} \|_2
= \left\{\int \left[\widehat f(t) - f_{\text{true}}(t)\right]^2 dt \right\}^{1/2}
\approx
\left\{\sum_{g=1}^G \left[\widehat f(t_g) - f_{\text{true}}(t_g)\right]^2 \Delta t \right\}^{1/2}.\)

\begin{figure}[H]
	\centering
	\captionsetup{font={stretch=1}}
	\includegraphics[width=0.99\textwidth]{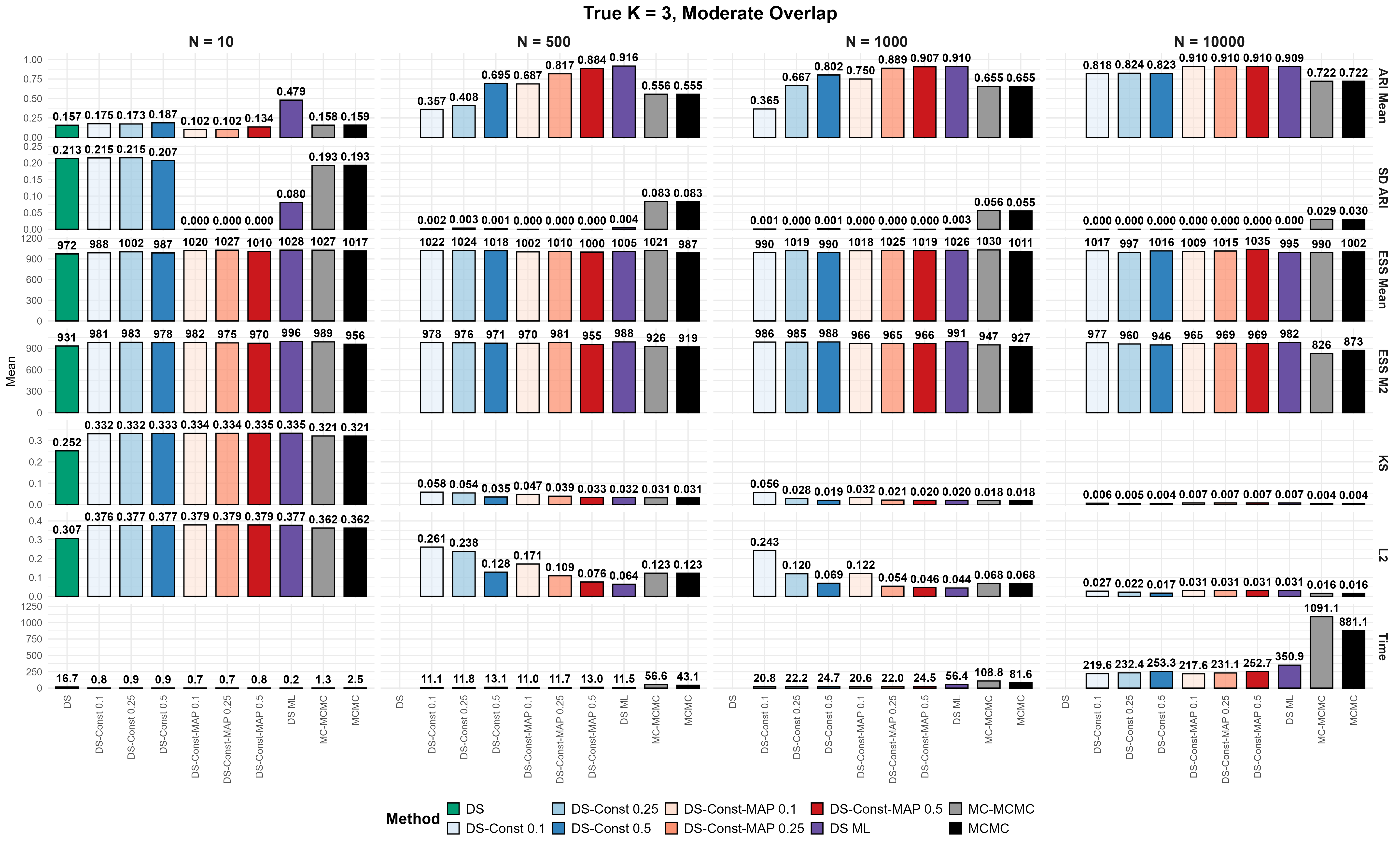}
	\caption{Mean (over 50 replicates) of ARI mean, ARI standard deviation, ESS of the first moment, ESS of the second moment, KS statistic, \(L_2\) norm, and CPU time (seconds) for each approach.}
	\label{fig:gmm:moderate_y}
\end{figure}

Figure \ref{fig:gmm:moderate_y} summarizes the evaluation metrics, averaged over \(50\) independent replicates, for the moderate overlap setting. For the \(N =10\) scenario, we use the direct sampler that was derived in Section \ref{gmm:ss:comp}. Because \(N\) is small, it is computationally feasible to enumerate the label space and sample directly from the posterior. For the DS-constrained approaches using the Gibbs sampler with the subset of data, increasing the subsample size generally improves performance metrics, at the cost of slightly increased CPU time. The DS-ML approach achieves higher ARI than the traditional MCMC while providing comparable density estimation performance. Overall, both the DS-constrained approaches (with the Gibbs subsample step), and the DS-ML approach offer clear computational advantages relative to MCMC.  

\begin{figure}[H]
	\centering
	\captionsetup{font={stretch=1}}
	\includegraphics[width=0.99\textwidth]{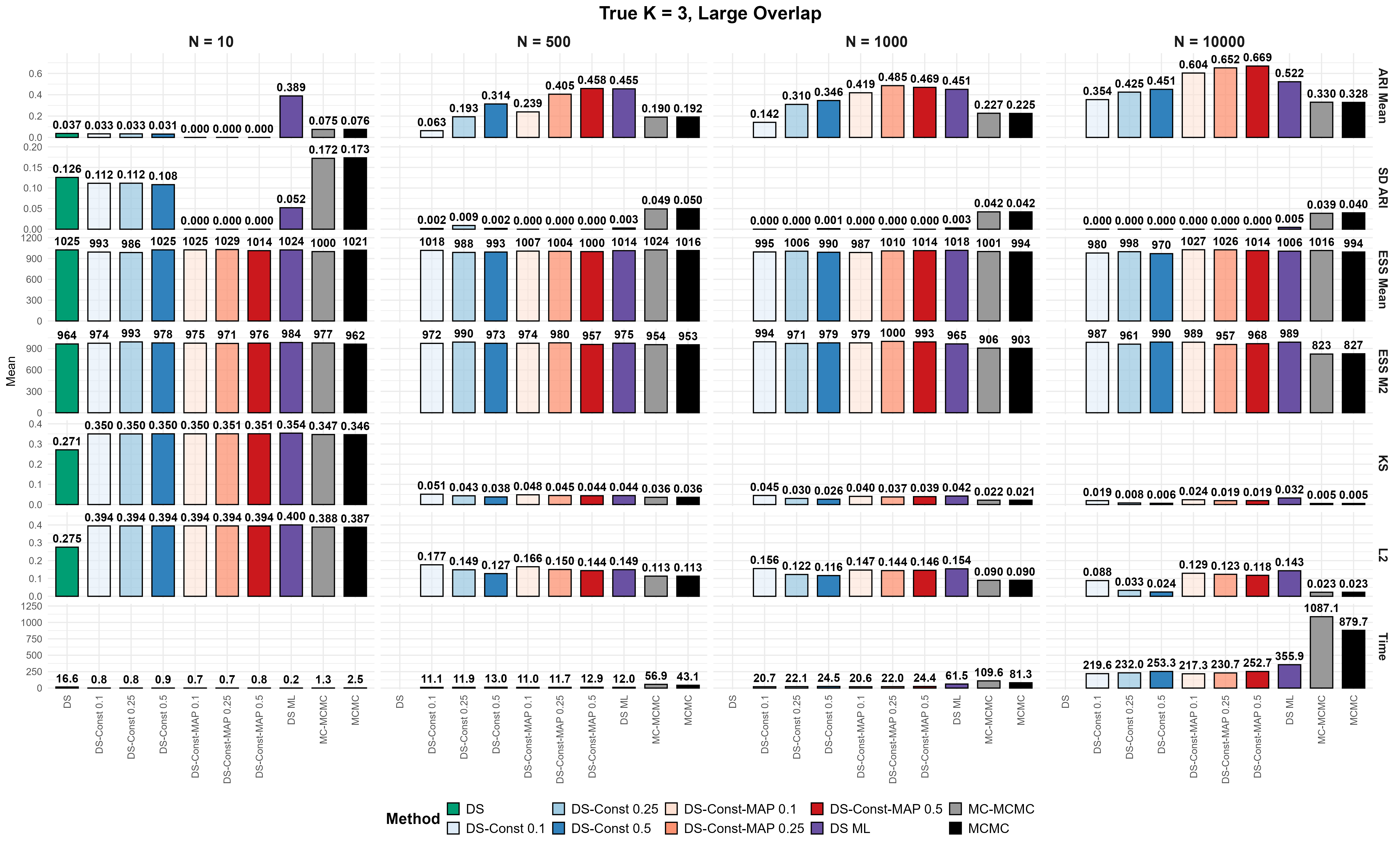}
	\caption{Mean (over 50 replicates) of ARI mean, ARI standard deviation, ESS of the first moment, ESS of the second moment, KS statistic, \(L_2\) norm, and CPU time (seconds) for each approach.}
	\label{fig:gmm:large_y}
\end{figure}

Figure \ref{fig:gmm:large_y} summarizes the evaluation metrics for the large overlap setting. Overall, these results lead to similar conclusions as the moderate overlap setting. The constrained approaches outperform MCMC and MC-MCMC in terms of ARI. Across all methods, ARI values are lower than in the moderate overlap scenario, which is expected because the true clusters are closer together and observations near cluster boundaries are more difficult to label. As before, the constrained approaches provide computational advantages relative to traditional MCMC. 

One difference in this setting is that the constrained methods show slightly worse density estimation performance than the MCMC methods. Since these approaches restrict the label space, it is possible that the candidate set does not always include the best labeling when there is large overlap. This issue was not seen in the moderate and no overlap settings. One way to address this issue is to expand the candidate set, for example by incorporating additional ML algorithms or additional initializations to construct a richer set of candidate partitions.

\begin{figure}[H]
	\centering
	\captionsetup{font={stretch=1}}
	\includegraphics[width=0.99\textwidth]{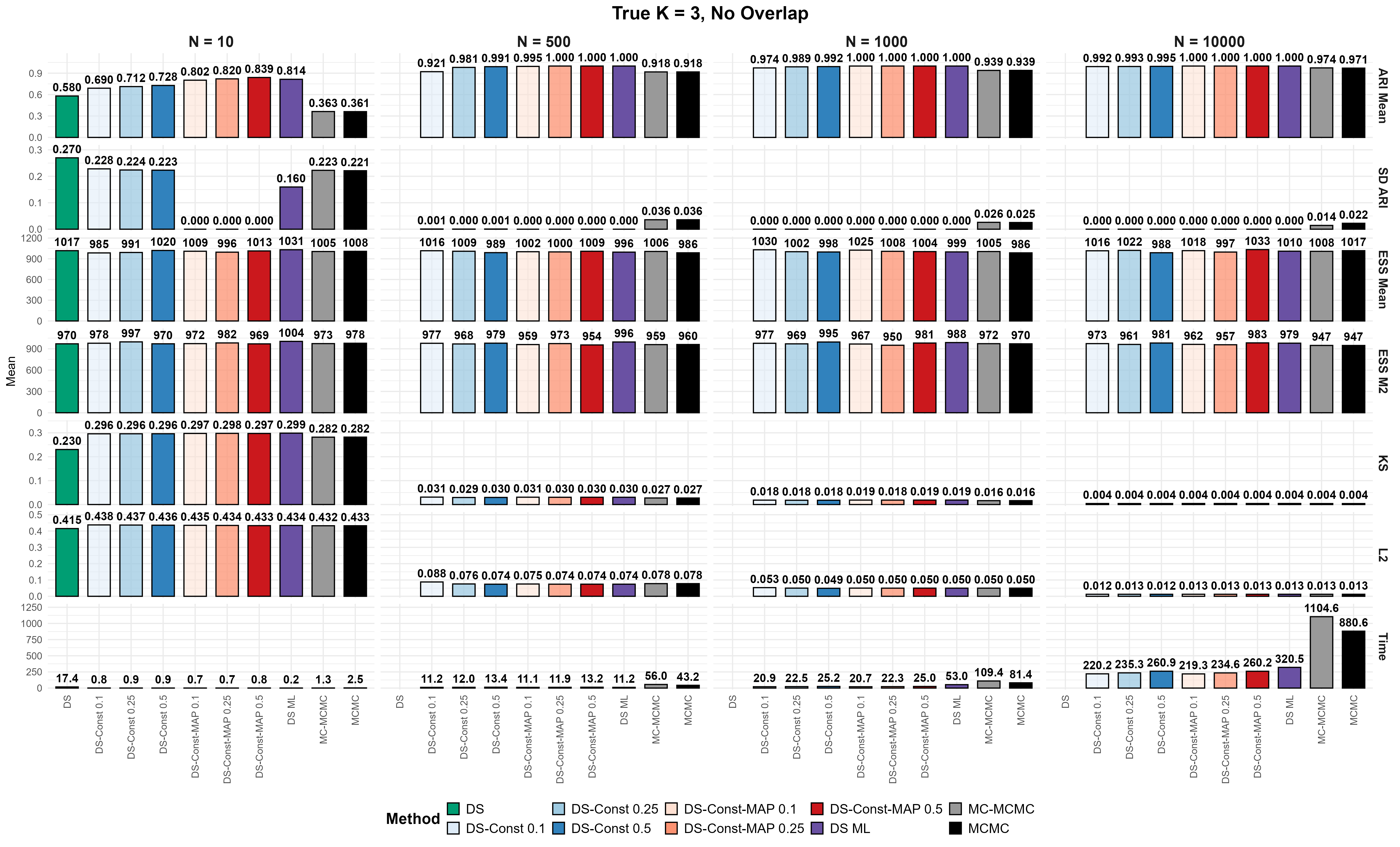}
	\caption{Mean (over 50 replicates) of ARI mean, ARI standard deviation, ESS of the first moment, ESS of the second moment, KS statistic, \(L_2\) norm, and CPU time (seconds) for each approach.}
	\label{fig:gmm:far_y}
\end{figure}

The third scenario, with no overlap between clusters is presented in Figure \ref{fig:gmm:far_y}. This setting is the easiest for clustering because the components are well separated. For larger \(N\), all methods achieve near perfect clustering performance, with ARI values close to \(1\). Density estimation performance is also very similar across methods. Even in this favorable setting, the constrained approaches have computational advantages compared to MCMC especially as \(N\) increases. 

Across all three separation scenarios, we obtain ESS values for the first and second moments close to \(1000\). An advantage of the DS and DS-constrained approaches is that this ESS is achieved using only \(1000\) posterior draws, whereas the MCMC approach typically requires additional iterations (and burn-in/thinning) to have an ESS of similar size. 

\subsection*{Appendix E.2: Simulations with Covariates}
To evaluate the different approaches in the regression scenario, we compare similar metrics as in the intercept only scenario. In this simulation we compute the KS and \(L_2\) in a different way than the previous section because the grid based computation scales with \(N \times G \times K\times n_{samp}\), becoming very slow for large \(N\). Instead, we compute KS and \(L_2\) using an observation based calculation instead of a grid based calculation, similar to \citet{skhosana2024modified}.

Let \(\{(\mathbf{x}_i,Y_i)\}_{i=1}^{N}\) denote the observed sample. For posterior draw \(j=1,\dots,n_{\text{samp}}\) and component \(k=1,\dots,K\), we have \(\pi_k^{[j]}, \boldsymbol{\beta}_k^{[j]}, \sigma_k^{[j]}\), with component mean
\(\mu_{ik}^{[j]} = \mathbf{x}_i^{\prime}\boldsymbol{\beta}_k^{[j]}\). The conditional mixture CDF for draw \(j\) evaluated at \((\mathbf{x}_i,Y_i)\) is \(F^{[j]}(Y_i\mid \mathbf{x}_i) = \sum_{k=1}^K \pi_k^{[j]}\, \Phi\!\left(\frac{Y_i -\mathbf{x}_i^{\prime}\boldsymbol{\beta}_k^{[j]}}{\sigma_k^{[j]}}\right)\) where \(\Phi\) is the standard Normal CDF. Define the posterior mean CDF \( \widehat F(Y_i\mid \mathbf{x}_i)= \frac{1}{n_{\text{samp}}}\sum_{j=1}^{n_{\text{samp}}} F^{[j]}(Y_i\mid \mathbf{x}_i)\).  The true conditional CDF is \(F_{\text{true}}(Y_i\mid \mathbf{x}_i) = \sum_{k=1}^K \pi_{k,\text{true}}\, \Phi\!\left(\frac{Y_i - \mathbf{x}_i^{\prime}\boldsymbol{\beta}_{k,\text{true}}}{\sigma_{k,\text{true}}}\right).\)   The KS statistic is computed as a maximum over observations 
\[KS= \max_{i=1,\dots,N}\left| F_{\text{true}}(Y_i\mid \mathbf{x}_i) - \widehat F(Y_i\mid \mathbf{x}_i)\right|. \]
The conditional mixture pdf is \(
f^{[j]}(Y_i\mid \mathbf{x}_i)= \sum_{k=1}^K \pi_k^{[j]}\,\frac{1}{\sigma_k^{[j]}}\,
\phi\!\left(\frac{Y_i - \mathbf{x}_i^{\prime}\boldsymbol{\beta}_k^{[j]}}{\sigma_k^{[j]}}\right),\)
where \(\phi\) is the standard Normal pdf. The posterior mean pdf is \(\widehat f(Y_i\mid \mathbf{x}_i)= \frac{1}{n_{\text{samp}}}\sum_{j=1}^{n_{\text{samp}}} f^{[j]}(Y_i\mid \mathbf{x}_i)\).
The true conditional pdf is
\(f_{\text{true}}(Y_i\mid \mathbf{x}_i)= \sum_{k=1}^K \pi_{k,\text{true}}\,\frac{1}{\sigma_{k,\text{true}}}\,
\phi\!\left(\frac{Y_i - \mathbf{x}_i^{\prime}\boldsymbol{\beta}_{k,\text{true}}}{\sigma_{k,\text{true}}}\right).\)
The empirical \(L_2\) distance between the fitted and true conditional pdfs is
\(L_2= \left\{\frac{1}{N}\sum_{i=1}^{N}\left[ f_{\text{true}}(Y_i\mid \mathbf{x}_i) - \widehat f(Y_i\mid \mathbf{x}_i)\right]^2\right\}^{1/2}\).

\begin{figure}[H]
	\centering
	\captionsetup{font={stretch=1}}
	\includegraphics[width=0.89\textwidth]{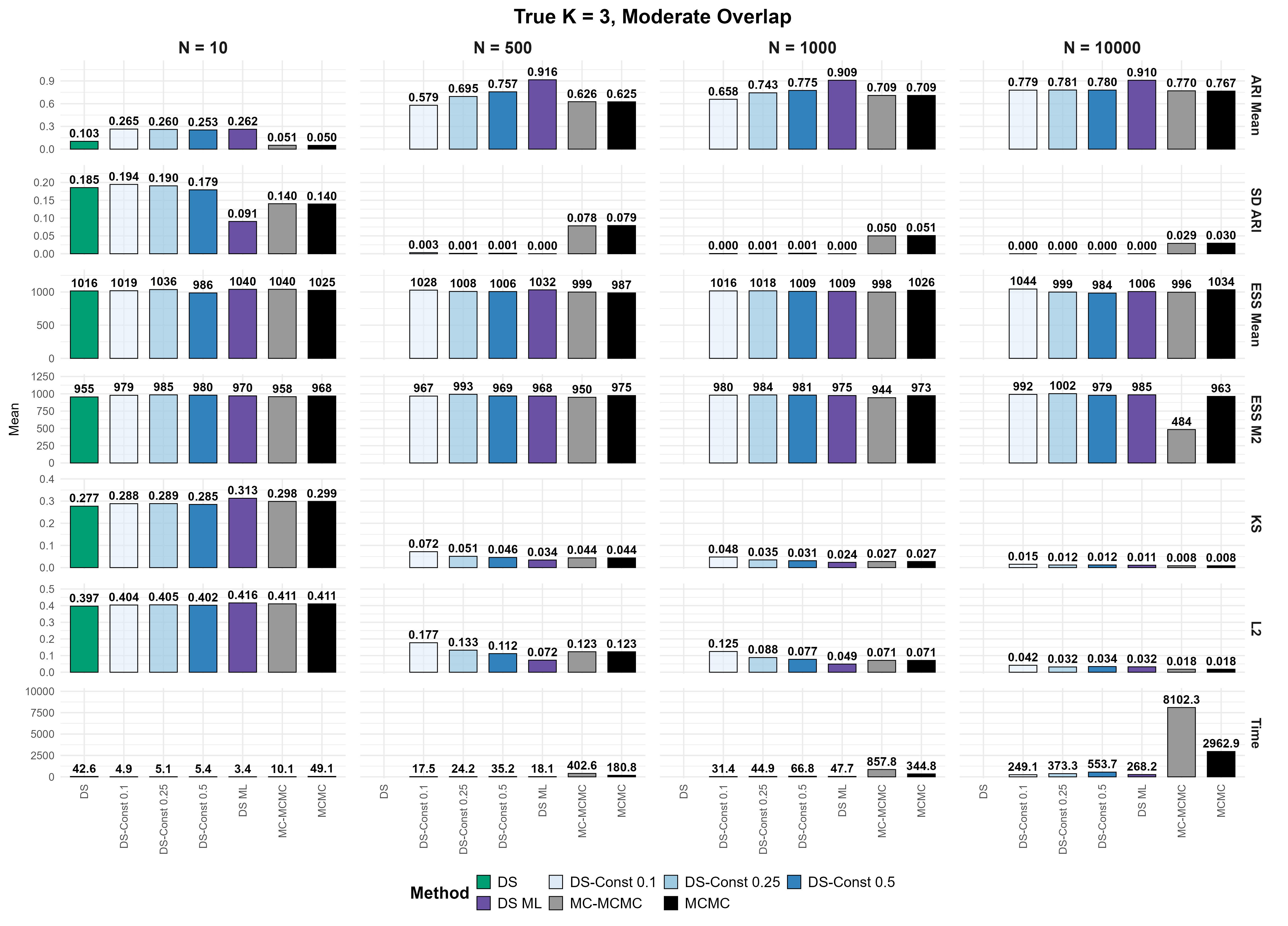}
	\caption{Mean (over 50 replicates) of ARI mean, ARI standard deviation, ESS of the first moment, ESS of the second moment, KS statistic, \(L_2\) norm, and CPU time (seconds) for each approach.}
	\label{fig:gmm:moderate_xy}
\end{figure}

Figure \ref{fig:gmm:moderate_xy} summarizes the evaluation metrics for the moderate overlap setting in the regression simulation. These results tell a similar story to the intercept only scenario. In this simulation, we omit the DS-Const-MAP approach because the \(k\)-means clustering with the MAP initialization performs poorly in the regression scenario. Overall, the constrained approaches, both the Gibbs sampler with the subset of data and the ML algorithms, perform well in the moderate overlap setting. They yield a higher ARI than MCMC while a also requiring significantly less computation time. The density estimation metrics are comparable across all methods in this scenario. 

The constrained approaches have a much smaller variability in clustering performance, as reflected by lower standard deviations of the ARI compared to MCMC. This reduced variability is expected because the constrained methods restrict the posterior label draws to a smaller candidate set. When this candidate set contains partitions close/equal to the truth and all the posterior draws come from this reduced space, then there will be consistently high ARI values across draws. The unconstrained MCMC sampler must explore a much larger label space. So some posterior draws are good partitions (high ARI) while others are further from the truth (low ARI). This is why we see a high variability and lower mean ARI for MCMC compared to the constrained approaches. 

\begin{figure}[H]
	\centering
	\captionsetup{font={stretch=1}}
	\includegraphics[width=0.89\textwidth]{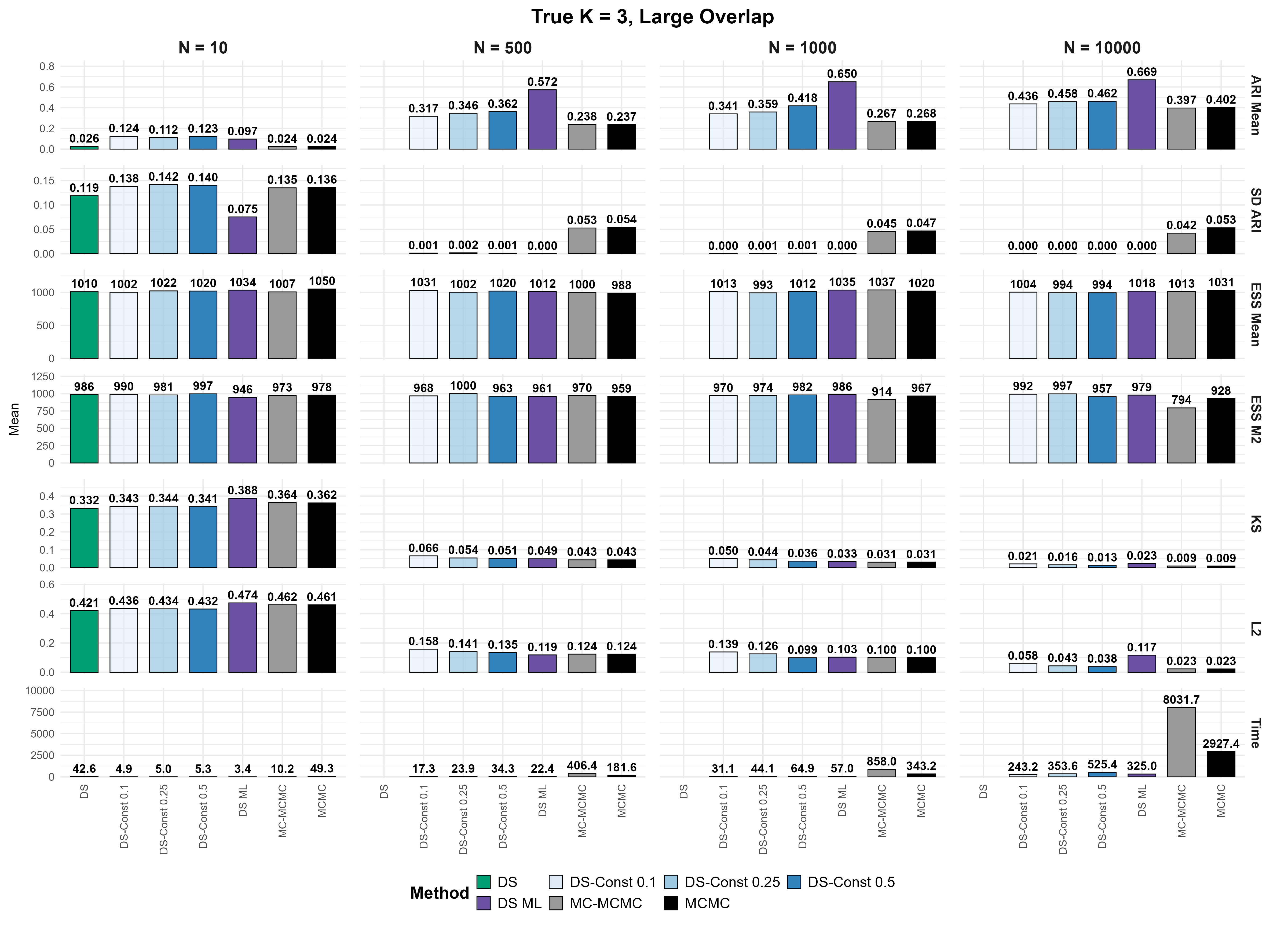}
	\caption{Mean (over 50 replicates) of ARI mean, ARI standard deviation, ESS of the first moment, ESS of the second moment, KS statistic, \(L_2\) norm, and CPU time (seconds) for each approach.}
	\label{fig:gmm:close_xy}
\end{figure}

Figure \ref{fig:gmm:close_xy} summarizes the evaluation metrics for the moderate overlap setting in the regression simulation. Similar to the intercept only simulations, ARI values are lower across all methods than in the moderate overlap setting since more observations fall in the region where the three components overlap, making labeling those observations more challenging. In this setting the constrained approaches still have lower variability in ARI than MCMC, indicating more stable clustering performance across posterior draws. The constrained approaches have density estimation performance similar to the MCMC methods. The constrained approaches also have a clear computational advantage over MCMC, which requires more iterations to achieve a comparable ESS for the first and second moment.

\begin{figure}[H]
	\centering
	\captionsetup{font={stretch=1}}
	\includegraphics[width=0.89\textwidth]{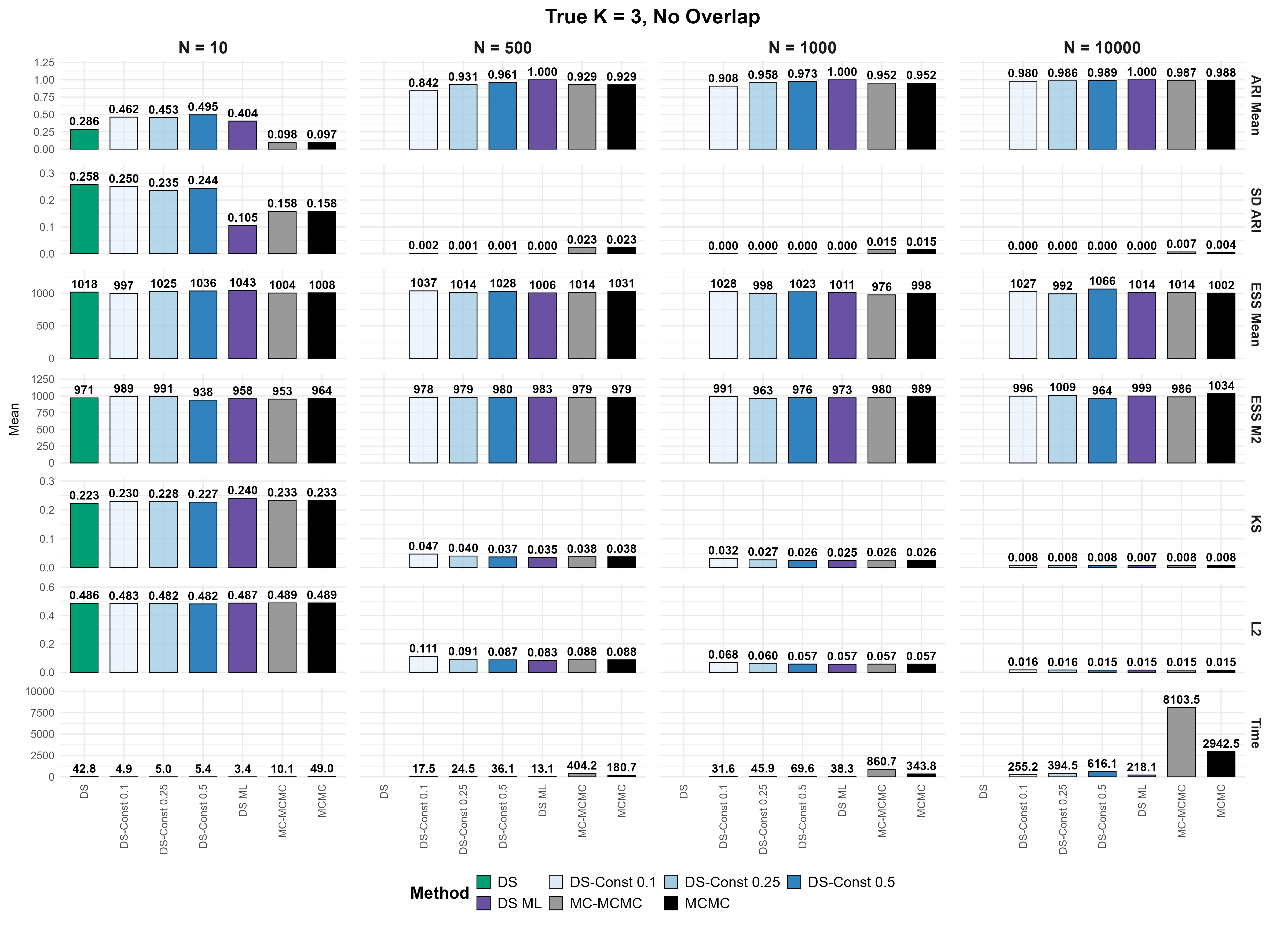}
	\caption{Mean (over 50 replicates) of ARI mean, ARI standard deviation, ESS of the first moment, ESS of the second moment, KS statistic, \(L_2\) norm, and CPU time (seconds) for each approach.}
	\label{fig:gmm:far_xy}
\end{figure}

Figure \ref{fig:gmm:far_xy} summarizes the evaluation metrics for the moderate overlap setting in the regression simulation. All methods achieve near perfect clustering performance with a mean ARI close to \(1\), and density estimation is comparable for all approaches. Even in this favorable setting, the constrained methods still have a computational advantage over MCMC.

\section*{Appendix F: Data Illustration}
We illustrate the proposed methodology using measurements on two rock crab species from the genus Leptograpsus \citep{campbell1974multivariate}, which is widely used as a benchmark example in the finite mixture of regressions literature. The dataset contains \(N = 200\) crabs with variables including frontal lip width (FL), maximum width of the carapace (CW), rear width (RW), length along the midline of the carapace (CL), and body depth (BD). In this illustration, we use FL as the covariate and CW as the response. Each crab is also assigned a color label, either blue or orange, which is treated as the ground truth cluster membership for evaluating the clustering performance via the ARI defined in Section \ref{gmm:ss:y_only}. Table \ref{fig:gmm:crab.table} shows that DS-ML and DS-Const-0.5 achieve comparable ARI values close to one, indicating strong clustering performance, while having a much faster computation time compared the the implementations with MCMC (MCMC and MC-MCMC). 

\begin{table}[H]
	\centering
	\caption{Clustering performance (mean ARI), effective sample size (ESS), and computation time measured in seconds for the crab data illustration.}
	\begin{tabular}{cccc}
		\toprule
		Method & Mean ARI & ESS & CPU Time (s) \\ \midrule
		DS-ML & 0.979 & 1000 & 9.266 \\ 
		DS-Constr-0.1 & 0.831 & 910.912 & 8.676 \\ 
		DS-Constr-0.25 & 0.935 & 1000  & 10.531 \\
		DS-Constr-0.5 & 0.971 & 1000 & 13.771 \\ 
		MC-MCMC & 0.826 & 1032.703 & 81.865\\
		MCMC  & 0.817  & 1000 & 91.052  \\ 
		\bottomrule
	\end{tabular}\label{fig:gmm:crab.table}
\end{table}

\section*{Appendix G: Sensitivity Analysis of Risk Factors}
In this Appendix, we investigate how the choice of covariates affects clustering performance in the regression analysis. Table~\ref{tab:xi_summary} summarizes the candidate covariates, which consist of gestational age and several binary risk factors.

\begin{table}[H]
	\centering
	\caption{Summary of covariates used in analysis. Binary variables are summarized by their sample proportion, and gestational age is summarized by its mean (sd).}
	\label{tab:xi_summary}
	\small
	\begin{tabular}{lll}
		\toprule
		Variable & Type & Summary \\
		\midrule
		No prenatal care & Binary & 0.021 \\
		Smoker & Binary & 0.023 \\
		Previous preterm birth & Binary & 0.039 \\
		Any diabetes & Binary & 0.098 \\
		Any hypertension & Binary & 0.135 \\
		Gestational age (weeks) & Continuous & 38.484 (1.664) \\
		\bottomrule
	\end{tabular}
\end{table}

Because clustering algorithms that rely on Euclidean distance can perform poorly when covariates are binary \citep{huang1998extensions, linzer2011polca, krzanowski1993location, hennig2013find}, we implement the DS-ML approach with algorithms designed for mixed-type (e.g., binary and continuous) data \citep{huang1998extensions, kim2017fast, foss2018kamila, mcparland2016model}. In particular, we use $k$-prototypes \citep{huang1998extensions} and KAMILA \citep{foss2018kamila}. When only the continuous covariate (gestational age) is included, we instead use the EM algorithm and $k$-means. Table~\ref{tab:test_covariates} reports the ARI under several covariate combinations. 

\begin{table}[h!]
	\centering
	\caption{ARI results for the birth weight data under different covariate sets. The unconstrained MCMC implementation ran for 4.402 days, whereas DS-ML with \(K_{\max}=5\) ran for approximately 1--2 hours across all scenarios. All scenarios had an ESS of \(1000\).}
	\label{tab:birth_weight_results_xy3}
	\small
	\setlength{\tabcolsep}{5pt}
	\begin{tabular}{lp{6.2cm}lccc}
		\toprule
		Method & Variables & Algorithms & Mean ARI & $K_{\max}$ & $K_{\text{post}}$ \\ 
		\midrule
		MCMC  & Gestational age, previous preterm, prenatal care & - & -0.024  & 2 & 2 \\
		DS-ML & Gestational age, previous preterm 
		& $k$-prototypes, kamila 
		& 0.379 & 5 & 2 \\
		
		DS-ML & Gestational age, prenatal care, smoking status, previous preterm, diabetes, hypertension 
		& $k$-prototypes, kamila 
		& 0.132  & 5 & 2 \\
		
		DS-ML &  Gestational age, prenatal care, previous preterm 
		& $k$-prototypes, kamila 
		& 0.707  & 5 & 2 \\
		
		DS-ML &  Gestational age, prenatal care 
		& $k$-prototypes, kamila 
		& 0.720  & 5 & 2 \\
		
		DS-ML &  Gestational age, prenatal care, previous preterm, smoking status 
		& $k$-prototypes, kamila 
		& 0.381 & 5 & 2 \\
		
		DS-ML & Gestational age 
		& EM, $k$-means 
		& 0.729  & 5 & 2 \\
		
		DS-ML & Gestational age, diabetes, hypertension 
		& $k$-prototypes, kamila 
		& 0.679  & 5 & 2 \\

		DS-ML & Gestational age, prenatal care, diabetes, hypertension 
		& $k$-prototypes, kamila 
		& 0.674   & 5 & 2 \\

		DS-ML & Gestational age, diabetes
		& $k$-prototypes, kamila 
		& 0.103  & 5 & 2 \\
		
		DS-ML & Gestational age, hypertension
		& $k$-prototypes, kamila 
		& 0.132  & 5 & 2 \\
		
		DS-ML & Gestational age, prenatal care, previous preterm, diabetes, hypertension 
		& $k$-prototypes, kamila 
		& 0.133 & 5 & 2 \\
		\bottomrule
		\label{tab:test_covariates}
	\end{tabular}
\end{table}

The strongest recovery of the preterm/non-preterm groups is achieved by the models including gestational age alone (ARI $=0.729$), gestational age with prenatal care (ARI $=0.720$), and gestational age with prenatal care and previous preterm birth (ARI $=0.707$). These results suggest that gestational age, together with at most one or two well-chosen covariates, is sufficient to recover the clinically relevant classes. We use the term ``well-chosen'' rather than ``informative'' to emphasize that a covariate's value is being judged here against a known reference partition. In a fully unsupervised setting, this benchmark would not be available, and formal variable selection methods for clustering \citep{witten2010framework, dy2004feature} would be more appropriate.

The impact of adding binary covariates varies across covariate sets. Adding prenatal care to gestational age preserves strong recovery (ARI $=0.720$), whereas adding diabetes alone or hypertension alone reduces ARI to $0.103$ and $0.132$, respectively. Smoking status shows a similar pattern where we see by adding it to the otherwise strong combination of gestational age, prenatal care, and previous preterm birth reduces the ARI from $0.707$ to $0.381$. Among the binary covariates considered, prenatal care is the only one that consistently preserves or improves recovery of the preterm partition. The others tend to introduce variation that pulls the clustering toward alternative latent structures. These results align with the clustering literature on variable selection that suggest natural groupings often depend on a small subset of available features, and irrelevant variables can worsen clustering performance \citep{dy2004feature, friedman2004clustering, witten2010framework, wang2008variable}. For this reason, in the main text we focus on the model that includes only gestational age as a covariate.

In Table~\ref{tab:test_covariates}, the ARI of the unconstrained MCMC implementation is near zero. To interpret this result, it is important to note that DS-ML and the MCMC sampler are clustering on different inputs. DS-ML applies clustering algorithms directly to the selected feature vectors. The Gibbs sampler fits a mixture of regressions for birth weight conditional on gestational age and the selected risk factors, and so its components reflect residual heterogeneity in birth weight after accounting for gestational age. Poor agreement with the preterm/non-preterm partition does not necessarily imply poor model fit of the Gibbs sampler. It suggests that the posterior clustering induced by the conditional mixture model targets a different latent structure.

This interpretation is consistent with \citet{schwartz2010joint}, who emphasize that birth weight and gestational age are closely related outcomes and that modeling birth weight conditional on gestational age can change the interpretation of covariate effects. They propose a Bayesian finite mixture model for the joint distribution of birth weight and gestational age instead of treating gestational age only as a covariate in a conditional birth weight model. Motivated by \citet{schwartz2010joint}, we also fit an intercept-only bivariate Gaussian mixture model for the joint distribution of birth weight and gestational age using a Gibbs sampler, with $K=2$. This model produced an ARI of $0.572$, which is much higher than the conditional MCMC result.

Table~\ref{tab:posterior_comparison} compares posterior summaries of the regression coefficients and other mixture components between DS-ML and the conditional MCMC. The two methods produce quite different parameter estimates. For example, the gestational age coefficient in component $1$ is $108.41$ under DS-ML versus $197.78$ under MCMC, the intercepts for component $2$ differ in sign, and the variance of component $2$ is approximately twice as large under MCMC compared to DS-ML. These differences support the interpretation that the two approaches are identifying different mixture components.

\begin{table}[H]
	\centering
	\caption{Posterior means and standard deviations of mixture components for DS-ML vs MCMC.}
	\label{tab:posterior_comparison}
	\begin{tabular}{lrrrr}
		\hline
		Parameter & DS-ML Mean & DS-ML SD & MCMC Mean & MCMC SD \\
		\hline
		$\pi_1$ & 0.8867 & 0.0002 & 0.8702 & 0.0014 \\
		$\pi_2$  & 0.1133 & 0.0002 & 0.1298 & 0.0014 \\
		$\sigma^2_1$ & 170,222.9 & 141.3 & 144,512 & 234.9 \\
		$\sigma^2_2$ & 153,094.3 & 360.8 & 290,337.3 & 1,197.2 \\
		$\beta_{1,1}$ (Intercept) & -824.56 & 8.83 & -4,374.39 & 6.22 \\
		$\beta_{1,2}$ (Gestational age) & 108.41 & 0.23 & 197.78 & 0.16 \\
		$\beta_{1,3}$ (No prenatal care) & -41.82 & 1.9 & -55.85 & 2.08 \\
		$\beta_{1,4}$ (Previous preterm) & 83.88 & 1.65 & -3.73 & 1.39 \\
		$\beta_{2,1}$ (Intercept) & -2,553.74 & 10.64 & 1,284.89 & 52.16 \\
		$\beta_{2,2}$ (Gestational age) & 140.48 & 0.3 & 61.04 & 1.32 \\
		$\beta_{2,3}$ (No prenatal care) & 79.21 & 2.93 & -188.76 & 10.88 \\
		$\beta_{2,4}$ (Previous preterm) & 175.22 & 1.67 & -4.36 & 6.7 \\
		\hline
	\end{tabular}
\end{table}

Overall, the DS-ML approach constructs a constrained candidate set using clustering algorithms applied directly to the selected features. When gestational age is included among these features, the resulting candidate partitions can more directly reflect the preterm/non-preterm structure. The higher ARI values for DS-ML therefore should be interpreted as evidence that the constrained candidate set contains partitions aligned with the clinical reference grouping, rather than as evidence that the unconstrained Gibbs sampler has failed. 

{\renewcommand{\baselinestretch}{1}\selectfont
	\bibliography{bibliography.bib}

@article{pageetal:2025,
  author  = {Page, Garritt L. and Ventrucci, Massimo and Franco-Villoria, Maria and Seeley, Matthew K.},
  title   = {Informed Bayesian Finite Mixture Models via Asymmetric Dirichlet Priors},
  journal = {The Annals of Applied Statistics},
  year    = {2025},
  volume  = {19},
  number  = {3},
  pages   = {2412--2435},
}

@ARTICLE{argiento:2022,
  author = {{Argiento}, Raffaele and {De Iorio}, Maria},
  title = {Is Infinity that Far? A Bayesian Nonparametric Perspective of Finite Mixture Models},
  journal = {Annals of Statistics},
  year = {2022},
  pages = {2641-2663},
  volume= {50},
  number = {5}
}

@Article {RousseauMengersen:2011,
 author = {Rousseau, Judith and Mengersen, Kerrie},
 title = {Asymptotic behaviour of the posterior distribution in overfitted mixture models},
 journal = {Journal of the Royal Statistical Society: Series B (Statistical Methodology)},
 volume = {73},
 number = {5},
 pages = {689-710},
 year = {2011}
}

@book{schervish2012theory,
  title={Theory of statistics},
  author={Schervish, Mark J},
  year={2012},
  publisher={Springer Science \& Business Media}
}

@book{press2009subjective,
	title={Subjective and objective Bayesian statistics: principles, models, and applications},
	author={Press, S James},
	year={2009},
	publisher={John Wiley \& Sons}
}

@article{celeux2000computational,
  title={Computational and inferential difficulties with mixture posterior distributions},
  author={Celeux, Gilles and Hurn, Merrilee and Robert, Christian P},
  journal={Journal of the American Statistical Association},
  volume={95},
  number={451},
  pages={957--970},
  year={2000},
  publisher={Taylor \& Francis}
}

@article{jasra2005markov,
  author    = {Jasra, Ajay and Holmes, Christopher C. and Stephens, David A.},
  title     = {Markov Chain Monte Carlo Methods and the Label Switching Problem in Bayesian Mixture Modeling},
  journal   = {Statistical Science},
  volume    = {20},
  number    = {1},
  pages     = {50--67},
  year      = {2005},
  month     = {February},
}

@article{stephens2000bayesian,
  title={Bayesian analysis of mixture models with an unknown number of components-an alternative to reversible jump methods},
  author={Stephens, Matthew},
  journal={Annals of statistics},
  pages={40--74},
  year={2000},
  publisher={JSTOR}
}

@article{diebolt1994estimation,
  title={Estimation of finite mixture distributions through Bayesian sampling},
  author={Diebolt, Jean and Robert, Christian P},
  journal={Journal of the Royal Statistical Society: Series B (Methodological)},
  volume={56},
  number={2},
  pages={363--375},
  year={1994},
  publisher={Wiley Online Library}
}

@article{richardson1997bayesian,
  title={On Bayesian analysis of mixtures with an unknown number of components (with discussion)},
  author={Richardson, Sylvia and Green, Peter J},
  journal={Journal of the Royal Statistical Society Series B: Statistical Methodology},
  volume={59},
  number={4},
  pages={731--792},
  year={1997},
  publisher={Oxford University Press}
}

@article{stephens2000dealing,
  title={Dealing with label switching in mixture models},
  author={Stephens, Matthew},
  journal={Journal of the Royal Statistical Society: Series B (Statistical Methodology)},
  volume={62},
  number={4},
  pages={795--809},
  year={2000},
  publisher={Wiley Online Library}
}

@article{jain2004split,
  title={A split-merge Markov chain Monte Carlo procedure for the Dirichlet process mixture model},
  author={Jain, Sonia and Neal, Radford M},
  journal={Journal of computational and Graphical Statistics},
  volume={13},
  number={1},
  pages={158--182},
  year={2004},
  publisher={Taylor \& Francis}
}

@article{nobile2007bayesian,
  title={Bayesian finite mixtures with an unknown number of components: The allocation sampler},
  author={Nobile, Agostino and Fearnside, Alastair T},
  journal={Statistics and Computing},
  volume={17},
  number={2},
  pages={147--162},
  year={2007},
  publisher={Springer}
}

@article{viele2002modeling,
  title={Modeling with mixtures of linear regressions},
  author={Viele, Kert and Tong, Barbara},
  journal={Statistics and Computing},
  volume={12},
  number={4},
  pages={315--330},
  year={2002},
  publisher={Springer}
}

@article{banerjee2020modeling,
  title={Modeling massive spatial datasets using a conjugate Bayesian linear modeling framework},
  author={Banerjee, Sudipto},
  journal={Spatial statistics},
  volume={37},
  pages={100417},
  year={2020},
  publisher={Elsevier}
}

@article{finley2019efficient,
  title={Efficient algorithms for Bayesian nearest neighbor Gaussian processes},
  author={Finley, Andrew O and Datta, Abhirup and Cook, Bruce D and Morton, Douglas C and Andersen, Hans E and Banerjee, Sudipto},
  journal={Journal of Computational and Graphical Statistics},
  volume={28},
  number={2},
  pages={401--414},
  year={2019},
  publisher={Taylor \& Francis}
}

@article{zhang2019practical,
  title={Practical Bayesian modeling and inference for massive spatial data sets on modest computing environments},
  author={Zhang, Lu and Datta, Abhirup and Banerjee, Sudipto},
  journal={Statistical Analysis and Data Mining: The ASA Data Science Journal},
  volume={12},
  number={3},
  pages={197--209},
  year={2019},
  publisher={Wiley Online Library}
}

@article{taylor2015statistical,
  title={Statistical learning and selective inference},
  author={Taylor, Jonathan and Tibshirani, Robert J},
  journal={Proceedings of the National Academy of Sciences},
  volume={112},
  number={25},
  pages={7629--7634},
  year={2015},
  publisher={National Academy of Sciences}
}

@article{fithian2014optimal,
  title={Optimal inference after model selection},
  author={Fithian, William and Sun, Dennis and Taylor, Jonathan},
  journal={arXiv preprint arXiv:1410.2597},
  year={2014}
}

@article{gao2024selective,
  title={Selective inference for hierarchical clustering},
  author={Gao, Lucy L and Bien, Jacob and Witten, Daniela},
  journal={Journal of the American Statistical Association},
  volume={119},
  number={545},
  pages={332--342},
  year={2024},
  publisher={Taylor \& Francis}
}

@article{chen2023selective,
  title={Selective inference for k-means clustering},
  author={Chen, Yiqun T and Witten, Daniela M},
  journal={Journal of Machine Learning Research},
  volume={24},
  number={152},
  pages={1--41},
  year={2023}
}

@article{yekutieli2012adjusted,
  title={Adjusted Bayesian inference for selected parameters},
  author={Yekutieli, Daniel},
  journal={Journal of the Royal Statistical Society Series B: Statistical Methodology},
  volume={74},
  number={3},
  pages={515--541},
  year={2012},
  publisher={Oxford University Press}
}

@article{park2020function,
  title={A function emulation approach for doubly intractable distributions},
  author={Park, Jaewoo and Haran, Murali},
  journal={Journal of Computational and Graphical Statistics},
  volume={29},
  number={1},
  pages={66--77},
  year={2020},
  publisher={Taylor \& Francis}
}

@article{zong2023criterion,
  title={Criterion constrained Bayesian hierarchical models},
  author={Zong, Qingying and Bradley, Jonathan R},
  journal={TEST},
  volume={32},
  number={1},
  pages={294--320},
  year={2023},
  publisher={Springer}
}

@article{ng2001spectral,
  title={On spectral clustering: Analysis and an algorithm},
  author={Ng, Andrew and Jordan, Michael and Weiss, Yair},
  journal={Advances in neural information processing systems},
  volume={14},
  year={2001}
}

@article{hartigan1979algorithm,
  title={Algorithm AS 136: A k-means clustering algorithm},
  author={Hartigan, John A and Wong, Manchek A},
  journal={Journal of the royal statistical society. series c (applied statistics)},
  volume={28},
  number={1},
  pages={100--108},
  year={1979},
  publisher={JSTOR}
}

@article{ahmed2020k,
  title={The k-means algorithm: A comprehensive survey and performance evaluation},
  author={Ahmed, Mohiuddin and Seraj, Raihan and Islam, Syed Mohammed Shamsul},
  journal={Electronics},
  volume={9},
  number={8},
  pages={1295},
  year={2020},
  publisher={MDPI}
}

@article{ward1963hierarchical,
  title={Hierarchical grouping to optimize an objective function},
  author={Ward Jr, Joe H},
  journal={Journal of the American statistical association},
  volume={58},
  number={301},
  pages={236--244},
  year={1963},
  publisher={Taylor \& Francis}
}

@article{murtagh2014ward,
  title={Ward’s hierarchical agglomerative clustering method: which algorithms implement Ward’s criterion?},
  author={Murtagh, Fionn and Legendre, Pierre},
  journal={Journal of classification},
  volume={31},
  number={3},
  pages={274--295},
  year={2014},
  publisher={Springer}
}

@article{park2009simple,
  title={A simple and fast algorithm for K-medoids clustering},
  author={Park, Hae-Sang and Jun, Chi-Hyuck},
  journal={Expert systems with applications},
  volume={36},
  number={2},
  pages={3336--3341},
  year={2009},
  publisher={Elsevier}
}

@article{dempster1977maximum,
  title={Maximum likelihood from incomplete data via the EM algorithm},
  author={Dempster, Arthur P and Laird, Nan M and Rubin, Donald B},
  journal={Journal of the royal statistical society: series B (methodological)},
  volume={39},
  number={1},
  pages={1--22},
  year={1977},
  publisher={Wiley Online Library}
}

@article{skhosana2024modified,
  title={A modified EM-type algorithm to estimate semi-parametric mixtures of non-parametric regressions},
  author={Skhosana, Sphiwe B and Millard, Salomon M and Kanfer, Frans HJ},
  journal={Statistics and Computing},
  volume={34},
  number={4},
  pages={125},
  year={2024},
  publisher={Springer}
}

@article{bradley2021empirical,
  title={Empirical Bayesian analysis through the lens of a particular class of constrained Bayesian hierarchical models},
  author={Bradley, Jonathan R and Zong, Qingying},
  journal={Stat},
  volume={10},
  number={1},
  pages={e403},
  year={2021},
  publisher={Wiley Online Library}
}

@article{campbell1974multivariate,
  title={A multivariate study of variation in two species of rock crab of the genus Leptograpsus},
  author={Campbell, NA and Mahon, RJ},
  journal={Australian Journal of Zoology},
  volume={22},
  number={3},
  pages={417--425},
  year={1974},
  publisher={CSIRO Publishing}
}

@article{mcgrory2007variational,
  title={Variational approximations in Bayesian model selection for finite mixture distributions},
  author={McGrory, Clare A and Titterington, DM237087605559631},
  journal={Computational Statistics \& Data Analysis},
  volume={51},
  number={11},
  pages={5352--5367},
  year={2007},
  publisher={Elsevier}
}

@article{fan2012variational,
  title={Variational learning for finite Dirichlet mixture models and applications},
  author={Fan, Wentao and Bouguila, Nizar and Ziou, Djemel},
  journal={IEEE transactions on neural networks and learning systems},
  volume={23},
  number={5},
  pages={762--774},
  year={2012},
  publisher={IEEE}
}

@article{tentoni2004birthweight,
  title={Birthweight by gestational age in preterm babies according to a Gaussian mixture model},
  author={Tentoni, Stefania and Astolfi, Paola and De Pasquale, Antonio and Zonta, Laura A},
  journal={BJOG: An International Journal of Obstetrics \& Gynaecology},
  volume={111},
  number={1},
  pages={31--37},
  year={2004},
  publisher={Wiley Online Library}
}

@article{platt2001detecting,
  title={Detecting and eliminating erroneous gestational ages: a normal mixture model},
  author={Platt, Robert W and Abrahamowicz, Michal and Kramer, Michael S and Joseph, KS and Mery, Les and Blondel, B{\'e}atrice and Br{\'e}art, G{\'e}rard and Wen, SW},
  journal={Statistics in medicine},
  volume={20},
  number={23},
  pages={3491--3503},
  year={2001},
  publisher={Wiley Online Library}
}

@article{urquia2012mixture,
  title={A mixture model to correct misclassification of gestational age},
  author={Urquia, Marcelo L and Moineddin, Rahim and Frank, John W},
  journal={Annals of epidemiology},
  volume={22},
  number={3},
  pages={151--159},
  year={2012},
  publisher={Elsevier}
}

@article{schwartz2010joint,
  title={Joint Bayesian analysis of birthweight and censored gestational age using finite mixture models},
  author={Schwartz, Scott L and Gelfand, Alan E and Miranda, Marie L},
  journal={Statistics in medicine},
  volume={29},
  number={16},
  pages={1710--1723},
  year={2010},
  publisher={Wiley Online Library}
}

@article{charnigo2010thinking,
  title={Thinking outside the curve, part I: modeling birthweight distribution},
  author={Charnigo, Richard and Chesnut, Lorie W and LoBianco, Tony and Kirby, Russell S},
  journal={BMC pregnancy and childbirth},
  volume={10},
  number={1},
  pages={37},
  year={2010},
  publisher={Springer}
}

@article{huang1998extensions,
  title={Extensions to the k-means algorithm for clustering large data sets with categorical values},
  author={Huang, Zhexue},
  journal={Data mining and knowledge discovery},
  volume={2},
  number={3},
  pages={283--304},
  year={1998},
  publisher={Springer}
}

@article{linzer2011polca,
  title={poLCA: An R package for polytomous variable latent class analysis},
  author={Linzer, Drew A and Lewis, Jeffrey B},
  journal={Journal of statistical software},
  volume={42},
  pages={1--29},
  year={2011}
}

@article{krzanowski1993location,
  title={The location model for mixtures of categorical and continuous variables},
  author={Krzanowski, WJ},
  journal={Journal of Classification},
  volume={10},
  number={1},
  pages={25--49},
  year={1993},
  publisher={Springer}
}

@article{hennig2013find,
  title={How to find an appropriate clustering for mixed-type variables with application to socio-economic stratification},
  author={Hennig, Christian and Liao, Tim F},
  journal={Journal of the Royal Statistical Society Series C: Applied Statistics},
  volume={62},
  number={3},
  pages={309--369},
  year={2013},
  publisher={Oxford University Press}
}

@article{foss2018kamila,
  title={kamila: clustering mixed-type data in R and Hadoop},
  author={Foss, Alexander H and Markatou, Marianthi},
  journal={Journal of Statistical Software},
  volume={83},
  pages={1--44},
  year={2018}
}

@article{kim2017fast,
  title={A fast K-prototypes algorithm using partial distance computation},
  author={Kim, Byoungwook},
  journal={Symmetry},
  volume={9},
  number={4},
  pages={58},
  year={2017},
  publisher={MDPI}
}

@article{mcparland2016model,
  title={Model based clustering for mixed data: clustMD},
  author={McParland, Damien and Gormley, Isobel Claire},
  journal={Advances in Data Analysis and Classification},
  volume={10},
  number={2},
  pages={155--169},
  year={2016},
  publisher={Springer}
}

@article{oja1991fitting,
  title={Fitting mixture models to birth weight data: a case study},
  author={Oja, Hannu and Koiranen, Markku and Rantakallio, Paula},
  journal={Biometrics},
  pages={883--897},
  year={1991},
  publisher={JSTOR}
}

@article{gage2002modeling,
  title={Modeling birthweight and gestational age distributions: Additive vs. multiplicative processes},
  author={Gage, Timothy B},
  journal={American Journal of Human Biology},
  volume={14},
  number={6},
  pages={728--734},
  year={2002},
  publisher={Wiley Online Library}
}

@article{slaughter2009bayesian,
  title={A Bayesian latent variable mixture model for longitudinal fetal growth},
  author={Slaughter, James C and Herring, Amy H and Thorp, John M},
  journal={Biometrics},
  volume={65},
  number={4},
  pages={1233--1242},
  year={2009},
  publisher={Oxford University Press}
}

@article{wilcox1983birthweight,
  title={Birthweight and perinatal mortality: I. On the frequency distribution of birthweight},
  author={Wilcox, Allen J and Russell, Ian T},
  journal={International Journal of Epidemiology},
  volume={12},
  number={3},
  pages={314--318},
  year={1983},
  publisher={Oxford University Press}
}

@article{goldenberg2008epidemiology,
  title={Epidemiology and causes of preterm birth},
  author={Goldenberg, Robert L and Culhane, Jennifer F and Iams, Jay D and Romero, Roberto},
  journal={The lancet},
  volume={371},
  number={9606},
  pages={75--84},
  year={2008},
  publisher={Elsevier}
}

@article{hubert1985comparing,
  title={Comparing partitions},
  author={Hubert, Lawrence and Arabie, Phipps},
  journal={Journal of classification},
  volume={2},
  number={1},
  pages={193--218},
  year={1985},
  publisher={Springer}
}

@article{dahl2022search,
  title={Search algorithms and loss functions for Bayesian clustering},
  author={Dahl, David B and Johnson, Devin J and M{\"u}ller, Peter},
  journal={Journal of Computational and Graphical Statistics},
  volume={31},
  number={4},
  pages={1189--1201},
  year={2022},
  publisher={Taylor \& Francis}
}

@article{dy2004feature,
  title={Feature selection for unsupervised learning},
  author={Dy, Jennifer G and Brodley, Carla E},
  journal={Journal of machine learning research},
  volume={5},
  number={Aug},
  pages={845--889},
  year={2004}
}

@article{friedman2004clustering,
  title={Clustering objects on subsets of attributes (with discussion)},
  author={Friedman, Jerome H and Meulman, Jacqueline J},
  journal={Journal of the Royal Statistical Society Series B: Statistical Methodology},
  volume={66},
  number={4},
  pages={815--849},
  year={2004},
  publisher={Oxford University Press}
}

@article{witten2010framework,
  title={A framework for feature selection in clustering},
  author={Witten, Daniela M and Tibshirani, Robert},
  journal={Journal of the American Statistical Association},
  volume={105},
  number={490},
  pages={713--726},
  year={2010},
  publisher={Taylor \& Francis}
}

@article{wang2008variable,
  title={Variable selection for model-based high-dimensional clustering and its application to microarray data},
  author={Wang, Sijian and Zhu, Ji},
  journal={Biometrics},
  volume={64},
  number={2},
  pages={440--448},
  year={2008},
  publisher={Oxford University Press}
}
}

\end{document}